\RequirePackage{atbegshi}  % fix weird "\box255" error for large two-column figures
\documentclass[a4paper,twocolumn,10pt,accepted=2024-02-08]{quantumarticle}
\pdfoutput=1

\usepackage{hyperref}

\usepackage[utf8]{inputenc} % enable UTF-8

\usepackage{graphicx}

\usepackage[T1]{fontenc}

\IfFileExists{headers/config/showoverfull.config}{
	\overfullrule=1cm
}{
}

\usepackage{etoolbox}

\newbool{includeappendix}
\setbool{includeappendix}{true} % default value
\IfFileExists{headers/config/noappendix.config}{
	\setbool{includeappendix}{false}
}{}

\newif\ifincludeappendixx
\ifbool{includeappendix}{
	\includeappendixxtrue
}{
	\includeappendixxfalse
}

\usepackage{xr} % for \externaldocument{...}
\usepackage{filecontents}

\ifbool{includeappendix}{}{
	\input{appendix-labels-loader}

	\externaldocument{appendix-labels}
}

\usepackage{acro} % for \ac

\usepackage[usenames, dvipsnames]{color} % for textcolor
\usepackage{xcolor} % for definecolor

\definecolor{my-full-blue}{HTML}{1F77B4}
\definecolor{my-full-orange}{HTML}{FF7F0E}
\definecolor{my-full-green}{HTML}{2CA02C}
\definecolor{my-full-red}{HTML}{d62728}
\definecolor{my-full-purple}{HTML}{9467bd}
\definecolor{my-full-brown}{HTML}{8c564b}
\definecolor{my-full-pink}{HTML}{e377c2}
\definecolor{my-full-gray}{HTML}{7f7f7f}
\definecolor{my-full-olive}{HTML}{bcbd22}
\definecolor{my-full-cyan}{HTML}{17becf}
\colorlet{my-blue}{my-full-blue!30}
\colorlet{my-orange}{my-full-orange!30}
\colorlet{my-green}{my-full-green!30}
\colorlet{my-red}{my-full-red!30}
\colorlet{my-purple}{my-full-purple!30}
\colorlet{my-brown}{my-full-brown!30}
\colorlet{my-pink}{my-full-pink!30}
\colorlet{my-gray}{my-full-gray!30}
\colorlet{my-olive}{my-full-olive!30}
\colorlet{my-cyan}{my-full-cyan!30}

\definecolor{blue-node}{HTML}{1a237e}
\colorlet{blue-node-text}{white}

\definecolor{red-node}{HTML}{b71c1c}
\colorlet{red-node-text}{white}

\colorlet{green-node}{my-full-green}
\colorlet{green-node-text}{black}

\usepackage{listings}

\usepackage{textcomp}

\usepackage{xcolor}

\usepackage[scaled=0.8]{beramono}

\definecolor{ckeyword}{HTML}{7F0055}
\definecolor{ccomment}{HTML}{3F7F5F}
\definecolor{cstring}{HTML}{2A0099}

\lstdefinestyle{numbers}{
	numbers=left,
	framexleftmargin=20pt,
	numberstyle=\tiny,
	firstnumber=auto,
	numbersep=1em,
	xleftmargin=2em
}

\lstdefinestyle{layout}{
	frame=none,
	captionpos=b,
}

\lstdefinestyle{comment-style}{
	morecomment=[l]//,
	morecomment=[s]{/*}{*/},
	commentstyle={\color{ccomment}\itshape},
}

\lstdefinestyle{string-style}{
	morestring=[b]",%
	morestring=[b]',%
	stringstyle={\color{cstring}},
	showstringspaces=false,%
}

\lstdefinestyle{keyword-style}{
	keywordstyle={\ttfamily\bfseries},
	morekeywords={
		function,
		constructor,
		int,
		bool,
		return,
		returns,
		uint
	},
	morekeywords = [2]{},
	keywordstyle = [2]{\text},
	sensitive=true,
}

\lstdefinestyle{input-encoding}{
	inputencoding=utf8,
	extendedchars=true,
	literate=
	{ℝ}{$\reals$}1%
	{→}{$\rightarrow$}1%
	{α}{$\alpha$}1%
	{β}{$\beta$}1%
	{λ}{$\lambda$}1%
	{θ}{$\theta$}1%
	{ϕ}{$\phi$}1%
}

\lstdefinestyle{escaping}{
	moredelim={**[is][\color{blue}]{\%}{\%}},
	escapechar=|,
	mathescape=true
}

\lstdefinestyle{default-style}{
	basicstyle=\fontencoding{T1}\ttfamily\footnotesize,
	style=numbers,
	style=layout,
	style=comment-style,
	style=string-style,
	style=keyword-style,
	style=input-encoding,
	style=escaping,
	tabsize=2,
	upquote=true
}

\lstdefinelanguage{BASIC}{
	language=C++,
	style=default-style
}[keywords,comments,strings]%

\usepackage{tikz}

\usetikzlibrary{arrows}
\usetikzlibrary{automata}
\usetikzlibrary{calc}
\usetikzlibrary{backgrounds}
\usetikzlibrary{decorations.markings}
\usetikzlibrary{decorations.pathmorphing}
\usetikzlibrary{decorations.pathreplacing}
\usetikzlibrary{fit}
\usetikzlibrary{patterns}
\usetikzlibrary{positioning}
\usetikzlibrary{shadows}
\usetikzlibrary{shapes}
\usetikzlibrary{shapes.geometric}
\usetikzlibrary{trees}
\usetikzlibrary{intersections}

\usetikzlibrary{calligraphy}

\usetikzlibrary{decorations}

\usetikzlibrary{arrows.meta}
\usetikzlibrary{shapes.misc}

\tikzset{
    big arrow fill/.code={\definecolor{big arrow fill}{named}{#1}},
    big arrow fill={white},
    big arrow/.style 2 args={
        line width=#1, 
        -{Triangle[length={.5*(#1+#2)},width={#1+#2}]},
        postaction={
            -{Triangle[length={.5*(#1+#2)-sqrt(2)*1pt},width={#1+#2-sqrt(2)*2pt}]},
            draw=big arrow fill,
            line width={#1-1pt},
            shorten >={sqrt(2)*.5pt},
            shorten <=.5pt
        },
        every node/.style={
            text width=#1,
            align=center
        }
    }
}

\tikzstyle{gate} = [
	rectangle,
	inner ysep=0mm,
	minimum height=0.5cm,
	minimum width=1cm,
	fill=gray!30,
	font=\footnotesize,
]
\tikzstyle{tightgate} = [
	rectangle,
	inner ysep=0mm,
	minimum height=0.5cm,
	minimum width=0.8cm,
	fill=gray!30,
	font=\footnotesize,
]

\tikzstyle{ttightgate} = [
	rectangle,
	inner xsep=0.4mm,
	inner ysep=0mm,
	minimum height=0.5cm,
	minimum width=0.65cm,
	fill=gray!30,
	font=\footnotesize,
]

\tikzstyle{tightqballoc} = [
	rectangle,
	inner ysep=0mm,
	minimum height=0.5cm,
	minimum width=0.65cm,
	draw=gray, dashed,
	font=\footnotesize,
]

\tikzstyle{qballoc} = [
	rectangle,
	inner ysep=0mm,
	minimum height=0.5cm,
	minimum width=0.8cm,
	draw=gray, dashed,
	font=\footnotesize,
]
\tikzstyle{gtightgate} = [
	rectangle,
	inner ysep=0mm,
	minimum height=0.5cm,
	minimum width=0.8cm,
	fill=gray!30,
	font=\footnotesize,
	text=gray,
]
\tikzstyle{copygate} = [
	rectangle,
	inner ysep=0mm,
	minimum height=0.5cm,
	minimum width=1cm,
	fill=my-full-green,
	font=\footnotesize,
]
\tikzstyle{init} = [
	rectangle,
	rounded corners=7,
	fill=gray!30,
	minimum width=1cm,
	minimum height=0.5cm,
	inner ysep=0mm,
	font=\footnotesize,
]
\tikzstyle{tightinit} = [
	rectangle,
	rounded corners=7,
	fill=gray!30,
	minimum width=0.7cm,
	minimum height=0.5cm,
	inner ysep=0mm,
	font=\footnotesize,
]
\tikzstyle{component} = [
	rectangle,
	rounded corners=7,
	fill=gray!50,
	minimum width=3.5cm,
	minimum height=1cm,
	inner ysep=0mm,
	font=\footnotesize,
]

\tikzstyle{tightcomponent} = [
	rectangle,
	rounded corners=7,
	fill=gray!50,
	minimum width=2cm,
	minimum height=0.7cm,
	inner ysep=0mm,
	font=\footnotesize,
]

\tikzstyle{edge-label} = [
	midway,
	font=\footnotesize,
]

\tikzstyle{edge} = [line width=0.15mm]
\tikzstyle{ctrl} = [edge, Circle-{Latex[length=1mm,width=1mm]}]
\tikzstyle{targ} = [edge, -{Latex[length=1mm,width=1mm]}, line width=0.15mm]
\tikzstyle{avail} = [edge, -{Latex[length=1mm,width=1mm]}, dashed, draw=gray,
line width=0.15mm]
\tikzstyle{several} = [edge, -{Latex[length=1mm,width=1mm]}, densely dotted, line width=0.2mm]

\tikzstyle{gedge} = [color=gray, line width=0.15mm]
\tikzstyle{gctrl} = [gedge, Circle-{Latex[length=1mm,width=1mm]}]
\tikzstyle{gtarg} = [gedge, -{Latex[length=1mm,width=1mm]}, line width=0.15mm]

\tikzstyle{cc-edge} = [edge, -{Latex[open,length=1.5mm,width=1.5mm]}, densely dotted, line width=0.3mm]

\tikzstyle{steps} = [edge, -{Latex[length=1mm,width=1mm]}, dashed]

\tikzstyle{split} = [
	rounded corners=1,
	fill=gray!12,
	draw=gray!12,
	line width=#1
]
\tikzset{split/.default=0.8cm}

\tikzstyle{arrow-shape} = [
	single arrow,
	fill=gray!50,
	single arrow head extend=0.6ex,
	inner sep=2pt,
	inner xsep=1pt,
	minimum height=8pt,
]

\tikzstyle{qubit} = [
	rectangle,
	fill=gray!30,
	minimum size=0.3cm,
]

\tikzstyle{circuit-gate} = [
	rectangle,
	fill=gray!30,
	minimum size=0.4cm,
	inner sep=0pt,
	anchor=center
]
\tikzstyle{circuit-noop} = [
	rectangle,
	fill=none,
	opacity=0,
	minimum size=0.4cm,
	inner sep=0pt,
	anchor=center
]

\tikzstyle{circuit-ctrl} = [
	circle,
	fill=black,
	inner sep=0pt,
	minimum size=1.5mm,
	anchor=center
]

\tikzset{meter/.append style={
    rectangle,
	anchor=center,
    fill=gray!30,
    font=\vphantom{A},
    minimum size=0.4cm,
    line width=.7,
    path picture={
        \draw[black]
            ([shift={(.06,.2)}]path picture bounding box.south west)
            to[bend left=80]
            ([shift={(-.06,.2)}]path picture bounding box.south east);
        \draw[black,-{Latex[length=0.8mm,width=1.0mm]}]
            ([shift={(0,.18)}]path picture bounding box.south)
            --
            ([shift={(.14,-.11)}]path picture bounding box.north);
    }
}}

\tikzstyle{square-box}= [
	split=0.01cm, color = my-full-blue, fill=none, line width=0.02cm
]

\def\tikzRowSepDefault{0.5mm}
\def\tikzColumnSepDefault{0.5mm}
\def\tikzWireLeftDefault{0.6cm}
\def\tikzWireRightDefault{0.35cm}
\def\tikzWireLabelScaleDefault{0.9}
\def\tikzWireLabelShiftXDefault{-0.1}
\def\tikzWireLabelShiftYDefault{0.17}

\newcommand{\silent}[1]{{\color{black!40}#1}}

\usepackage[english]{babel}
\usepackage{amsmath}
\usepackage{graphicx}
\usepackage{stmaryrd}
\usepackage[noend]{algpseudocode}
\usepackage{xcolor}
\usepackage{tabu}
\usepackage{csvsimple}
\usepackage{array}
\usepackage{stackengine}
\usepackage{mathtools}
\usepackage{cancel}
\usepackage{blindtext}

\usepackage[numbers]{natbib}
\usepackage{siunitx}
\DeclareSIUnit[number-unit-product={}]{\percent}{\%}

\usepackage{ifthen} % provides \ifthenelse test  
\usepackage{xifthen} % provides \isempty test

\usepackage{xparse} % for \NewDocumentCommand

\usepackage{adjustbox} % fit within page

\usepackage{xspace} % for \xspace
\usepackage{wrapfig} % for \begin{wrapfigure}
\usepackage{subcaption}

\usepackage{makecell} % allow newline in cells

\usepackage{enumitem} % for costum enuemrates

\usepackage{amssymb}% http://ctan.org/pkg/amssymb
\usepackage{pifont}% http://ctan.org/pkg/pifont
\usepackage{amsthm}

\usepackage{proof} % for \infer
\usepackage{centernot}
\usepackage{setspace} % set line spacing
\usepackage{physics}
\usepackage{scalerel} % for scalerel
\usepackage{qcircuit}

\usepackage{thmtools}
\usepackage{thm-restate}

\usepackage{balance}
\usepackage{varwidth}

\newcommand{\bools}{\{0, 1\}}

\usepackage{multirow}
\newcommand{\fme}[1]{\llbracket #1 \rrbracket} %effect as a function on quantum states
\newcommand{\depgraph}{G}
\newcommand{\ctrlEdge}{\mathbin{\hbox{$\bullet$}\kern-1.5pt\hbox{$\to$}}}
\newcommand{\targEdge}{\to}
\newcommand{\availEdge}{\dashrightarrow} % maybe ---> ?

\newcommand{\inputNode}{\hbox{$\square$}\kern-1.5pt\hbox{$\to$}}
\newcommand{\outputNode}{\hbox{$\--$}\kern-1.5pt\hbox{$\square$}}

\newcommand{\unqomp}{Unqomp\xspace}
\newcommand{\reqomp}{Reqomp\xspace}
\newcommand{\squarerelwork}{\textsc{Square}\xspace}

\newcommand{\extended}[1]{E(#1)}

\newcommand{\legendboxbigwtxt}[2]{%
\raisebox{0.2em}{%
\fcolorbox{white!0}{#1}{#2}%
}%
}

\newcommand{\gatebox}[1]{\legendboxbigwtxt{gray!20}{#1}}

\algdef{SE}[SUBALG]{Indent}{EndIndent}{}{\algorithmicend\ }%
\algtext*{Indent}
\algtext*{EndIndent}

\algrenewcommand\textproc{\text}

\algblock[TryCatchFinally]{try}{endtry}
\algcblock[TryCatchFinally]{TryCatchFinally}{finally}{endtry}
\algcblockdefx[TryCatchFinally]{TryCatchFinally}{catch}{endtry}
[1]{\textbf{catch} #1}
{\textbf{end try}}

\makeatletter
\ifthenelse{\equal{\ALG@noend}{t}}%
  {\algtext*{endtry}}
  {}%
\makeatother

\newcounter{ALG@line@recordedLineNumber}
\newcommand{\recordLineNumber}[0]{\setcounter{ALG@line@recordedLineNumber}{\value{ALG@line}}}
\newcommand{\continueLineNumber}[0]{\setcounter{ALG@line}{\value{ALG@line@recordedLineNumber}}}

\DeclareMathOperator*{\argmin}{arg\,min}
\renewcommand{\i}[0]{{\rm{i}}}

\newcommand{\graph}[0]{{\ensuremath{G}}}
\newcommand{\un}[1]{\ensuremath{\overline{#1}}}
\newcommand{\graphu}[0]{{\ensuremath{\un{G}}}}
\newcommand{\graphval}[0]{{\ensuremath{g^{\mathit{val}}}}}
\newcommand{\hil}[1]{\mathcal{H}_{#1}}

\newcommand{\ganc}[0]{\tvar{G_\text{anc}}}
\newcommand{\gadeps}[0]{\tvar{G_\text{ancdeps}}}

\newcommand{\na}[0]{\tvar{n_\text{anc}}}
\newcommand{\nq}[0]{\tvar{n_\text{qbits}}}
\newcommand{\true}[0]{\textbf{true}}
\newcommand{\false}[0]{\textbf{false}}

\newcommand{\complex}[0]{\mathbb{C}}
\newcommand{\copyqb}[1]{{\underline{#1}}}
\newcommand{\copysem}[1]{\llparenthesis #1 \rrparenthesis}

\algrenewcommand\algorithmicfunction{\textbf{func}}

\newcommand{\tvar}[1]{\ensuremath{\mathit{#1}}}
\newcommand{\tfield}[1]{\ensuremath{\text{#1}}}
\newcommand{\tfunc}[1]{\textproc{#1}}
\newcommand{\tprim}[1]{\text{#1}}
\newcommand{\tstring}[1]{\text{#1}}

\newcommand{\graphPrim}[3]{\ensuremath{\tprim{#1}_{#2}(#3)}}

\usepackage[capitalize]{cleveref}

\makeatletter
\renewcommand\theHALG@line{\thefigure.\arabic{ALG@line}}
\makeatother

\crefformat{section}{\S#2#1#3}

\newcommand{\secref}[1]{\S#1}

\crefrangeformat{section}{\S#3#1#4\crefrangeconjunction\S#5#2#6}

\crefmultiformat{section}{\S#2#1#3}{\crefpairconjunction\S#2#1#3}{\crefmiddleconjunction\S#2#1#3}{\creflastconjunction\S#2#1#3}

\newcommand{\crefrangeconjunction}{--}

\crefname{listing}{Lst.}{listings}
\crefname{line}{Lin.}{Lin.}
\crefname{appendix}{App.}{App.}
\crefname{mydef}{Def.}{Def.}
\crefname{mylem}{Lem.}{Lem.}

\newcommand{\tabref}[1]{Tab.~#1}
\newcommand{\appref}[1]{%
	\ifbool{includeappendix}{\cref{#1}}{the appendix}%
}
\newcommand{\Appref}[1]{%
	\ifbool{includeappendix}{\cref{#1}}{The appendix}%
}

\input{headers/lqa/headings}

\begin{document}

\title{Reqomp: Space-constrained Uncomputation for Quantum Circuits}

% https://github.com/quantum-journal/quantum-journal/blob/master/quantum-template.tex#L18
\author{Anouk Paradis}
\affiliation{ETH Zurich, Switzerland}
\email{anouk.paradis@inf.ethz.ch}

\author{Benjamin Bichsel}
\affiliation{ETH Zurich, Switzerland}
\email{benjamin.bichsel@inf.ethz.ch}

\author{Martin Vechev}
\affiliation{ETH Zurich, Switzerland}
\email{martin.vechev@inf.ethz.ch}

\maketitle

% make sure the abstract is in a separate file (abstract.tex). This file will be
% used to generate a markdown version of your abstract
\begin{abstract}
Quantum circuits must run on quantum computers with tight limits on qubit and gate counts.
To generate circuits respecting both limits, a promising opportunity is
exploiting \emph{uncomputation} to trade qubits for gates.

We present \reqomp, a method to automatically synthesize correct and efficient uncomputation of ancillae while respecting hardware constraints.
For a given circuit, \reqomp can offer a wide range of trade-offs between tightly constraining qubit count or gate count.

Our evaluation demonstrates that \reqomp can significantly reduce the number of required ancilla qubits by up to 96\%. On 80\% of our benchmarks, the ancilla qubits required can be reduced by at least 25\% while never incurring a gate count increase beyond 28\%.

\end{abstract}

%\maketitle

%%%%%%%%
% BODY %
%%%%%%%%

\section{Introduction}
\label{sec:intro}
Quantum computers will remain tightly resource-constrained for the foreseeable
future, both in terms of available qubits and number of operations applicable
before an error occurs.
Running quantum programs hence requires compiling them to circuits with a limited qubit and gate count.
A promising opportunity to achieve this goal is to exploit the need for
\emph{uncomputation} as an opening to trade qubits for gates.

\paragraph{What is Uncomputation?}
Just as classical programs, quantum circuits often leverage temporary values, called \emph{ancilla variables}.
Whereas classical programs can discard temporary values whenever convenient,
temporary values in quantum circuits must be carefully managed to avoid
side effects on other values through
entanglement~\cite[\secref{3}]{paradis_unqomp_2021}.
Uncomputation is the process of preventing such side effects by reverting
ancilla variables to state $\ket{0}$ after their last use, thus ensuring that
they are disentangled from the remainder of the state.
For instance, \cref{fig:pb-stmt-big-circ} shows a circuit implementing $CCCCH$: the $H$ gate on qubit $t$ with four control qubits $o$, $p$, $q$, and $r$.
\cref{fig:pb-stmt-big-circ-no} uses three ancillae variables $a$, $b$, $c$,
stored in the respective \emph{ancilla qubits} $u_0$, $u_1$, $u_2$.
% \begin{align}\label{eq:cch}
%     \ket{o}\ket{p}\ket{q}\ket{r}\ket{t} &\xmapsto{CCH} \ket{o}\ket{p}\ket{q}\ket{r}H^{opqr}\ket{t},\\
%     \text{where } H^{opqr} &= \begin{cases}
%         I & \text{if } o \cdot p \cdot q \cdot r = 0 \\
%         H & \text{if } o \cdot p \cdot q \cdot r = 1
%     \end{cases}.
% \end{align}
%
The first ancilla $a$ holds $o \cdot p$, $b$ holds $o \cdot p \cdot q$, and $c$
holds $o \cdot p \cdot q \cdot r$.
We then use this last ancilla $c$ to control the $H$ gate on $t$, only applying
$H$ if all of $o$, $p$, $q$, and $r$ hold state $\ket{1}$.
In \cref{fig:pb-stmt-big-circ-no}, these ancilla variables are not uncomputed,
and may result in unexpected interactions if this circuit is used as part of a
bigger computation.
They must therefore be uncomputed, as shown in \cref{fig:pb-stmt-big-circ-lazy}:
the operations applied to each of them are reverted at the end of the circuit, ensuring that all ancilla qubits are reset to $\ket{0}$.

\para{Reducing Qubits}
After uncomputing an ancilla variable, its qubit can be reused by another
ancilla variable, therefore reducing the overall number of qubits used by the
circuit.
Sometimes, it is even beneficial to uncompute an ancilla variable (too) early,
allowing its qubit to be reused at the cost of later \emph{recomputing} the
ancilla variable when it is needed again.

\begin{figure}
    \centering

    \begin{subfigure}[t]{0.38\linewidth}
        \centering
        \begin{tikzpicture}
            \input{figures/pb_stmt_circ.tex}
            \node(last-vert) at ($(v0-X-v0-0.west)!1.17!(v0-3.east)$) {};
            \node(txt-vert) at ($(v0-X-v0-0.west)!1!(v0-3.east)$) {};

            %highlighting ancillae lifetimes
            \begin{scope}[on background layer]
                \draw [split=0.2cm, color= my-full-red, opacity=0.2]
                (v0-X-v0-0.west) -- % top left
                (last-vert|- v0-3) -- % top left
                ($(v0-3.west)!0.5!(v0-3.east)$)--
                (v0-X-v0-0) -- % top left
                cycle;

                \node (ahoriz) at ($(v0-X-v0-0.west)!0.5!(v0-3.east)$) {};
                \node (avert) at ($(v0-3)!0.5!(s-3)$) {};
                \node[color=my-full-red!70] (a) at (txt-vert |- avert) {$a$};

                \draw [split=0.2cm, color= my-full-green, opacity=0.2]
                (v1-X-v1-0.west) -- % top left
                (last-vert|- v1-3) -- % top left
                %(v1-3) -- % top left
                ($(v1-3.west)!0.5!(v1-3.east)$)--
                (v1-X-v1-0) -- % top left
                cycle;

                \node (bhoriz) at ($(v1-X-v1-0.west)!0.5!(v1-3.east)$) {};
                \node (bvert) at ($(v1-3)!0.5!(v0-3)$) {};
                \node[color=my-full-green] (b) at (txt-vert |- bvert) {$b$};

                \draw [split=0.2cm, color= my-full-blue, opacity=0.2]
                (v2-X-v2-0.west) -- % top left
                (last-vert|- v2-H-r-0) -- % top left
                ($(v2-H-r-0.west)!0.5!(v2-H-r-0.east)$)--
                (v2-X-v2-0) -- % top left
                cycle;

                \node (choriz) at ($(v2-X-v2-0.west)!0.5!(v2-H-r-0.east)$) {};
                \node (cvert) at ($(v2-H-r-0)!0.5!(v1-3)$) {};
                \node[color=my-full-blue!70] (c) at (txt-vert |- cvert) {$c$};
            \end{scope}
        \end{tikzpicture}
        \caption{No uncomputation.}
        \label{fig:pb-stmt-big-circ-no}
    \end{subfigure}%
    ~%
    \begin{subfigure}[t]{0.6\linewidth}
        \centering
        \begin{tikzpicture}
            \input{figures/pb_stmt_lazy_circ.tex}
            
            %highlighting ancillae lifetimes
            \begin{scope}[on background layer]
                \draw [split=0.2cm, color= my-full-red, opacity=0.2]
                (v0-X-v0-0.west) -- % top left
                (v0-X-v0-1) -- % top left
                (v0-X-v0-1) -- % top left
                ($(v0-X-v0-1.west)!0.5!(v0-X-v0-1.east)$)--
                (v0-X-v0-0) -- % top left
                cycle;

                \node (ahoriz) at ($(v0-X-v0-0.west)!0.5!(v0-X-v0-1.east)$) {};
                \node (avert) at ($(v0-3)!0.5!(s-3)$) {};
                \node[color=my-full-red!70] (a) at (ahoriz |- avert) {$a$};

                \draw [split=0.2cm, color= my-full-green, opacity=0.2]
                (v1-X-v1-0.west) -- % top left
                (v1-X-v1-1) -- % top left
                (v1-X-v1-1) -- % top left
                ($(v1-X-v1-1.west)!0.5!(v1-X-v1-1.east)$)--
                (v1-X-v1-0) -- % top left
                cycle;

                \node (bhoriz) at ($(v1-X-v1-0.west)!0.5!(v1-X-v1-1.east)$) {};
                \node (bvert) at ($(v1-3)!0.5!(v0-3)$) {};
                \node[color=my-full-green] (b) at (bhoriz |- bvert) {$b$};

                \draw [split=0.2cm, color= my-full-blue, opacity=0.2]
                (v2-X-v2-0.west) -- % top left
                (v2-X-v2-1) -- % top left
                (v2-X-v2-1) -- % top left
                ($(v2-X-v2-1.west)!0.5!(v2-X-v2-1.east)$)--
                (v2-X-v2-0) -- % top left
                cycle;

                \node (choriz) at ($(v2-X-v2-0.west)!0.5!(v2-X-v2-1.east)$) {};
                \node (cvert) at ($(v2-H-r-0)!0.5!(v1-3)$) {};
                \node[color=my-full-blue!70] (c) at (choriz |- cvert) {$c$};
            \end{scope}
        \end{tikzpicture}
        \caption{3-qubit uncomputation.}
        \label{fig:pb-stmt-big-circ-lazy}
    \end{subfigure}

    %%%%%%%%%%%%%%%%%%%%%%%%%%%%%%%%%%%%%%%%%%%%%%%%%%%%%%%%%%%%%%%%%%%%%%%%%%%%
    \begin{subfigure}[b]{0.75\linewidth}
        \centering
        \begin{tikzpicture}
            \input{figures/pb_stmt_eager_circ.tex}
            \begin{scope}[on background layer]

                \draw [split=0.2cm, color= my-full-red, opacity=0.2]
                (v0-X-v0-0.west) -- % top left
                (v0-X-v0-1) -- % top left
                (v0-X-v0-1) -- % top left
                ($(v0-X-v0-1.west)!0.5!(v0-X-v0-1.east)$)--
                (v0-X-v0-0) -- % top left
                cycle;

                \node (ahoriz) at ($(v0-X-v0-0.west)!0.40!(v0-X-v0-1.east)$) {};
                \node (avert) at ($(v0-3.north)!0.5!(s-3.south)$) {};
                \node[color=my-full-red!70] (a) at (ahoriz |- avert) {$a$};

                \draw [split=0.2cm, color= my-full-red, opacity=0.2]
                (v0-X-v0-4.west) -- % top left
                (v0-X-v0-5) -- % top left
                (v0-X-v0-5) -- % top left
                ($(v0-X-v0-5.west)!0.5!(v0-X-v0-5.east)$)--
                (v0-X-v0-4) -- % top left
                cycle;
                
                \node (ahoriz2) at ($(v0-X-v0-4.west)!0.40!(v0-X-v0-5.east)$) {};
                \node (avert2) at ($(v0-3.north)!0.5!(s-3.south)$) {};
                \node[color=my-full-red!70] (a2) at (ahoriz2 |- avert2) {$a$};

                \draw [split=0.2cm, color= my-full-green, opacity=0.2]
                (v1-X-v1-0.west) -- % top left
                (v1-X-v1-1) -- % top left
                (v1-X-v1-1) -- % top left
                ($(v1-X-v1-1.west)!0.5!(v1-X-v1-1.east)$)--
                (v1-X-v1-0) -- % top left
                cycle;

                \node (bhoriz) at ($(v1-X-v1-0.west)!0.25!(v1-X-v1-1.east)$) {};
                \node (bvert) at ($(v1-3.north)!0.5!(r-H-r-0.south)$) {};
                \node[color=my-full-green] (b) at (bhoriz |- bvert) {$b$};

                \draw [split=0.2cm, color= my-full-blue, opacity=0.2]
                (v0-X-v0-2.west) -- % top left
                (v0-X-v0-3) -- % top left
                (v0-X-v0-3) -- % top left
                ($(v0-X-v0-2.west)!0.5!(v0-X-v0-3.east)$)--
                (v0-X-v0-2) -- % top left
                cycle;
                
                \node (choriz) at ($(v0-X-v0-2.west)!0.5!(v0-X-v0-3.east)$) {};
                \node (cvert) at ($(v0-3.north)!0.5!(s-3.south)$) {};
                \node[color=my-full-blue!70] (c) at (choriz |- cvert) {$c$};

            \end{scope}
        \end{tikzpicture}
        \caption{2-qubit uncomputation.}
        \label{fig:pb-stmt-big-circ-eager}
    \end{subfigure}~%
    \begin{subfigure}[b]{0.24\linewidth}
        \footnotesize
        \setlength{\tabcolsep}{1.5pt}
        \centering
        \begin{tabular}{@{}ccc@{}}
            Fig. & \#qb & \#gt \\\hline
            \ref{fig:pb-stmt-big-circ-no} & \textcolor{gray}{8} & \textcolor{gray}{4} \\
            \ref{fig:pb-stmt-big-circ-lazy} & 8 & 7 \\
            \ref{fig:pb-stmt-big-circ-eager} & 7 & 9
        \end{tabular}
        \vspace{1.045em}
        \caption{Resources.}
        \label{fig:pb-stmt-resources}
    \end{subfigure}

    \caption{Two uncomputation strategies for $CCCCH$.}
    \label{fig:pb-stmt-big-circ}
\end{figure}

% \para{Early Uncomputation Trades Space for Time}

\cref{fig:pb-stmt-big-circ-lazy} simply uncomputes ancilla variables in the
reverse order of their computation, namely $c$--$b$--$a$.
As no ancilla qubit can be reused, \cref{fig:pb-stmt-big-circ-lazy} requires $8$
qubits and $7$ gates overall (see \cref{fig:pb-stmt-resources}).
\cref{fig:pb-stmt-big-circ-eager} shows an alternative implementation of $CCCCH$
leveraging \emph{recomputation}.
It uncomputes ancilla variable $a$ early, making its qubit $u_0$ free for the
computation of ancilla variable $c$. However, uncomputing $b$ requires $a$
again, forcing us to \emph{recompute} it and subsequently uncompute it for a
second time, at the cost of $2$ additional gates.
Overall, \cref{fig:pb-stmt-big-circ-eager} thus trades qubits for
gates compared to \cref{fig:pb-stmt-big-circ-lazy}, as summarized in \cref{fig:pb-stmt-resources}.

\para{Correctness}
Clearly, uncomputation is only useful if it correctly resets ancillae to
$\ket{0}$ without modifying the remainder of the state.
However, this is difficult to achieve, as uncomputing an ancilla may require
some preprocessing on its controls, to ensure that they are in the right state
(for details, see \cref{sec:square}).
Synthesizing the right gates to achieve uncomputation is thus a fundamental
challenge, as evidenced by correctness issues in
\squarerelwork~\cite{ding_square_2020} which attempts to automate the placement
of programmer-defined computation and uncomputation blocks (see
\cref{sec:square}).

\para{Our Work}
We present \reqomp, a method to automatically synthesize and place correct yet
efficient uncomputation while respecting hardware constraints.
\reqomp takes as inputs a quantum circuit $C$ without uncomputation (such as
\cref{fig:pb-stmt-big-circ-no}), its ancilla variables and a space constraint
specifying the number of available ancilla qubits.
If possible, \reqomp extends $C$ to a circuit $\un{C}$ which uncomputes all
ancilla variables, using only the number of ancilla qubits specified.

To ensure \reqomp is both correct (whenever a circuit $\un{C}$ is returned, it
is a correct uncomputation of $C$) and practical (for most input circuits $C$
and size constraints, it finds a circuit $\un{C}$), we build it out of two
distinct components. First, to ensure correctness, we introduce
\emph{well-valued circuit graphs}, which extend circuit
graphs~\cite{paradis_unqomp_2021} (a graph representation of quantum circuits)
by the new concept of \emph{value indices} tracking the state of qubits. We
further formalize \tfunc{evolveVertex}, a method to safely introduce computation
or uncomputation in these graphs. Second, to ensure practicality, we use multiple
heuristics picking which steps of computation and uncomputation should be
inserted (through calls to \tfunc{evolveVertex}) and in which order.

\para{Evaluation}
Our experimental evaluation shows that \reqomp can significantly reduce the
number of required ancilla qubits by up to 96\% compared to the most relevant
previous work \unqomp~\cite{paradis_unqomp_2021}.
Many algorithms are amenable to a significant reduction: for $16$ of $20$
benchmarks, \reqomp can reduce the number of ancilla qubits by $25\%$ compared
to \unqomp, without incurring a gate count increase beyond $28\%$.
For the remaining $4$ examples, \reqomp strictly outperforms \unqomp, albeit by a smaller margin.
In some cases, \reqomp achieves an impressive ancilla qubit reduction at very
low cost: for one example, by $75\%$ at the cost of increasing gate count
by~$17.6\%$.

Note that \unqomp already showed that manual uncomputation is both
error-prone and less efficient than automatically synthesized
uncomputation~\cite[\secref{7}]{paradis_unqomp_2021}.

\pagebreak

\para{Main Contributions}
Our main contributions are:
\begin{itemize}
    \item Well-valued circuit graphs, a graph representation of circuits
    allowing for accurate value tracking, and a method evolveVertex to
    modify them~(\cref{sec:circ-graphs});
    \item \reqomp, a method using well-valued circuit graphs to synthesize and place uncomputation in circuits
    under space constraints~(\cref{sec:heuristics});
    \item A correctness proof for \reqomp~(\cref{sec:correctness});% \todo{Is it     actually for Reqomp? Maybe for evolveVertex, which is sufficient to ensure     correctness of Reqomp, or similar?};
    \item An implementation\footnote{\reqomp is publicly available at \url{https://github.com/eth-sri/Reqomp}.}
    and evaluation of \reqomp demonstrating it outperforms
    previous work~(\cref{sec:eval}).
\end{itemize}

\section{Background}
\label{sec:background}

We now introduce the necessary background on quantum computation.

\para{Quantum States}
We write the quantum state $\varphi$ of a system with qubits $p$ and $q$ as:
\begin{align}\label{eq:quantum-state}
    \sum_{j=0}^1 \sum_{k=0}^1 \gamma_{j,k} \ket{j}_p \otimes \ket{k}_q
    = \sum_{l \in \{0,1\}^2} \gamma_{l} \ket{l}_{pq} \in \hil{2},
\end{align}
where $\gamma_{j, k}, \gamma_{l} \in \complex$ and $\otimes$ is the Kronecker product.
If $\varphi$ factorizes into
$
    \left(\sum_{j} \gamma_j' \ket{j}_p \right)
    \otimes
    \left(\sum_{k} \gamma_k'' \ket{k}_q \right),
$
$p$ and $q$ are \emph{unentangled}, otherwise they are \emph{entangled}.
Whenever convenient, we omit $\otimes$ and write $\ket{j}$ instead of
$\ket{j}_{p}$.
We use latin letters $\ket{j}$ to denote computational basis states from the
canonical basis $\{\ket{0}, \ket{1}\}$ and greek letters~$\varphi$ to denote
arbitrary states.

\para{Gates}
A gate applies a unitary operation to a quantum state.
Here, we only consider gates with a single target qubit in state $\varphi$ and
potentially multiple control qubits $C = \{c_1, ...\}$ in state $\ket{j}$ for $j
\in \{0,1\}^m$, mapping
$
    \ket{j}_C
    \otimes
    \varphi
    \text{ to }
    %%%%%%%%%%%%%%%%%%%%%%%%%%%%%%
    \ket{j}_C
    \otimes
    \phi,
$
where the mapping from $\varphi$ to $\phi$ may depend on the control $j$. Specifically, only the value of the target qubit may be changed, while control qubits are preserved.
Note that this mapping can be naturally extended to superpositions (i.e., linear combinations as in \cref{eq:quantum-state}) by linearity.
Further, because any circuit can be decomposed into single-target gates, not
considering multi-target gates is not a fundamental restriction.

A gate is \emph{qfree} if its mapping can be fully described by operations on
computational basis states, i.e., if for control qubits $C$ and target qubit $t$ it is of the form
\begin{align*}
    \ket{j}_C \ket{k}_t
    \mapsto
    \ket{j}_C \ket{F(j,k)}_t,
\end{align*}
for $F \colon \{0,1\}^m \times \{0,1\} \to \{0,1\}$.
For example, the NOT gate $X$, the controlled NOT gate $CX$, and the Toffoli
gate $CCX$ are qfree, while the Hadamard gate $H$ and the controlled Hadamard
gate $CH$ are not qfree.
Qfree gates are known to be critical for synthesizing
uncomputation~\cite{paradis_unqomp_2021,bichsel_silq_2020,rand_reqwire_2019}.

\para{Uncomputation}
The task of uncomputation is to revert all ancilla variables in a circuit to
their initial state $\ket{0}$, while preserving the circuit effect on the other
variables. Formally, given a circuit $C$, we want to synthesize $\un{C}$ which
resets ancillae variables to $\ket{0}$ without affecting the remainder of the
state:
\begin{restatable}[Correct Uncomputation, \cite{paradis_unqomp_2021,bichsel_silq_2020}]{mydef}{defcorect}
	\label{def:correct}
    $\un{C}$ correctly uncomputes the ancillae $A$ in $C$ if whenever
    \begin{align*}
        \ket{0 \cdots 0}_A \otimes \varphi &\xmapsto{\llbracket C \rrbracket} \sum_{j \in \{0,1\}^{|A|}} \gamma_j \ket{j}_A \otimes \phi_j, \text{then}\\
        \ket{0 \cdots 0}_A \otimes \varphi &\xmapsto{\llbracket \un{C} \rrbracket} \sum_{j \in \{0,1\}^{|A|}} \gamma_j \ket{0 \cdots 0}_A \otimes \phi_j.
    \end{align*}
\end{restatable}
Here, $\llbracket C \rrbracket$ denotes the semantics of circuit $C$ acting on a
given input state. We refer to \cite{paradis_unqomp_2021} for a more thorough
introduction to uncomputation.

\section{Circuit Graphs}
\label{sec:circ-graphs}
As discussed in \cref{sec:intro}, \reqomp does not work directly on circuits, but instead relies on \emph{well-valued circuit graphs}. 
In this section, we first intuitively introduce these graphs
(\cref{sec:circ-graphs:circ-g}) and the method \tfunc{evolveVertex} to
manipulate them (\cref{sec:circ-graphs:uncomp}). Finally, we formalize the
definition of well-valued circuit graphs and show how \tfunc{evolveVertex} preserves their well-formedness
(\cref{sec:circ-graph:formal}).

Circuit graphs were introduced and formalized in \unqomp~\cite{paradis_unqomp_2021}. We here extend them to precisely track qubits values. We discuss the differences between well-valued circuit graphs and \unqomp circuit graphs in \cref{sec:circ-graph:formal}.

\subsection{Circuit Graph Intuition}
\label{sec:circ-graphs:circ-g}
\begin{figure}
    \begin{tikzpicture}
		\def\tikzTitleColor{gray!30}
		\def\tikzArrowHeadLen{11pt}
        \def\tikzArrowWidth{9pt}
		\node[fill=gray!12,anchor=north west] (circ) at (0, 3){
			\begin{tikzpicture}
				\node[anchor=north west] (input_circ_circ){
                    \begin{tikzpicture}
                        \input{figures/circGraphsIntro/circuitSimple.tex}
                    \end{tikzpicture}
                };
                \node[fill=\tikzTitleColor, anchor=north] (input_circ_title) at (input_circ_circ.south) {a) Circuit $C$};
			\end{tikzpicture}
			
		};

		\node[fill=gray!12,anchor=north east] (value-graph) at ($ (circ.south) + (-0.1cm, -0.4cm) $){
			\begin{tikzpicture}
				\node[anchor=north west] (input_circ_circ){
                    \begin{tikzpicture}
                        \input{figures/circGraphsIntro/val_graph_simple.tex}
                    \end{tikzpicture}
                };
                \node[fill=\tikzTitleColor, anchor=north] (input_circ_title) at (input_circ_circ.south) {b) Value graph $\graphval$};
			\end{tikzpicture}
		};

        \node[fill=gray!12,anchor=north west] (circ-graph) at ($ (circ.south) + (0.1cm, -0.4cm) $) {
			\begin{tikzpicture}
				\node[anchor=north west] (input_circ_circ){
                    \begin{tikzpicture}
                        \input{figures/circGraphsIntro/circgraph-simple.tex}
                    \end{tikzpicture}
                };
                \node[fill=\tikzTitleColor, anchor=north] (input_circ_title) at (input_circ_circ.south) {c) Circuit graph $\graph$};
			\end{tikzpicture}
		};

		\draw[big arrow={\tikzArrowWidth}{\tikzArrowHeadLen}, big arrow fill= gray, draw=none] (circ.west) -- node[midway] {}  (value-graph.north |- circ.west) -- (value-graph.north);

		\draw[big arrow={\tikzArrowWidth}{\tikzArrowHeadLen}, big arrow fill= gray, draw=none] (circ.east) -- node[midway] {}  (circ-graph.north |- circ.east) -- (circ-graph.north);

        \end{tikzpicture}
		\begin{subfigure}[t]{0pt}
            \phantomcaption\label{fig:intro-cg-circ}
        \end{subfigure}
		\begin{subfigure}[t]{0pt}
            \phantomcaption\label{fig:intro-cg-vg}
        \end{subfigure}
        \begin{subfigure}[t]{0pt}
            \phantomcaption\label{fig:intro-cg-cg}
        \end{subfigure}
        \caption{A Circuit $C$ and the corresponding value graph $\graphval$ and circuit graph $G$}
	\label{fig:intro-circuit-graphs}
\end{figure}

As an example, we consider the circuit depicted in \cref{fig:intro-cg-circ}.
This circuit uses an ancilla qubit $a$ to compute some output on qubit $t$,
based on the value of qubit $s$. The circuit graph $G$ corresponding to this
circuit is shown in \cref{fig:intro-cg-cg}.

\para{Vertices and Edges}
We first focus on the structure of the circuit graph $G$. $G$ contains one \emph{init vertex} per qubit (e.g.,
\gatebox{$s_{0.0}$} for qubit $s$), and one \emph{gate vertex} per gate (e.g.,
\gatebox{$s_{1.0}$} for the first $X$ gate on $s$).
It also connects consecutive vertices on the same qubit by a \emph{target edge},
e.g., \gatebox{$s_{0.0}$} $\targEdge$ \gatebox{$s_{1.0}$}. %Note how target edges can be seen as a wire, liking all 
Further, as $a_{1.0}$ represents a $CX$ gate controlled by qubit $s$,
the circuit graph $G$ also contains a \emph{control edge} between the
corresponding vertices on $s$ and $a$: $\gatebox{$s_{1.0}$} \ctrlEdge \gatebox{$a_{1.0}$}$.
Finally, the circuit graph $G$ also contains \emph{anti-dependency edges} to
enforce correct ordering between otherwise unordered vertices. For example,
\gatebox{$a_{1.0}$} $\availEdge$ \gatebox{$s_{0.1}$} ensures that the second $X$
gate on $s$ (represented by $s_{0.1}$) can only be applied after the $CX$ gate
targeting $a$ (represented by $a_{1.0}$). 
Anti-dependency edges can be reconstructed from the target and control edges:
for any three vertices $n, c, d$ such that $c \targEdge d$ and $c
\ctrlEdge n$, there is an edge $n \availEdge d$ ensuring that $n$ must be applied before $d$.

\para{Valid Circuit Graphs}
A circuit graph is \emph{valid} iff it corresponds to a valid circuit. Most
importantly, all valid circuit graphs must be acyclic\footnote{We recall the definition of valid circuit graphs from \cite{paradis_unqomp_2021} in
\cref{def:valid}}.
Any valid circuit graph $G$ can be converted into a circuit. The resulting circuit has one qubit per init vertex in $G$. We then pick any total order on the gate vertices of $G$ that is consistent with the partial order induced by its edges, and add gates to the circuit following this total order. \cite{paradis_unqomp_2021} showed that any of the circuits $G$ can be converted to (depending on the choice of total order) have equivalent semantics. We hence define the semantics of a valid circuit graph $G$, denoted $\fme{G}$, as
the semantics of any circuit it can be converted to.

\para{Tracking Values}
While the above construction follows \unqomp~\cite{paradis_unqomp_2021}, we
additionally introduce a new vertex naming convention to track qubit values.
Specifically, each vertex (e.g., $s_{1.0}$) is identified by its qubit (here
$s$), its \emph{value index} (here $1$) and its \emph{instance index} (here
$0$).
The value index is chosen such that intuitively, two vertices with the same
qubit and value index hold the same value, even in the presence of
entanglement.
The instance index is used to ensure uniqueness of vertex names.
For example, as $X$ is self-inverse, the value on qubit $s$ is the same in the
very beginning of the circuit (\gatebox{$s_{0.0}$}) as after applying the two
$X$ gates to $s$ (\gatebox{$s_{0.1}$}).
More precisely, if the input state to the circuit is $\ket{0}_s \otimes
\varphi$, after the two $X$ gates on $s$ have been applied, the final state is
$\ket{0}_s \otimes \varphi'$ for some $\varphi'$. Reflecting this in the circuit
graph, vertices $s_{0.0}$ and $s_{0.1}$ share the same value index $0$.

\para{Value Graph $\graphval$}
To track value indices during circuit graph construction and later during
uncomputation, we rely on the value graph $\graphval$, shown in
\cref{fig:intro-cg-vg}. It records for each qubit and value
index the possible value transitions.
$\graphval$ contains one init vertex per qubit but without an instance index, for example $s_0$ for qubit $s$.
When encountering a new gate, for example the first $X$ gate on qubit $s$,
we pick a fresh value index for this qubit and
extend $\graphval$. The value graph $\graphval$ also records which operations can be safely uncomputed. For instance, as $X$ is a qfree gate, the $X$ on $s_0$ can be uncomputed: applying $X$ on $s_1$ yields $s_0$ (note that
$X$ is self-inverse). We materialize this with the reverse
edge $s_1 \targEdge s_0$, giving:
\begin{center}
	\begin{tikzpicture}
		%%%%%%%%%
		% NODES %
		%%%%%%%%%
	
		\matrix[row sep=2mm,column sep=0mm, inner sep=0mm, nodes={inner sep=1mm}] {
			\node (s0)  {$s_0$}; &\node (r0p)  {\phantom{llll}}; & \node (s1)  {$s_1$}; \\
		};
	
		\path[<->]          (s0)  edge  node[above] {\scriptsize $X$} (s1);
		%\path[->]          (s1)  edge  [bend left=20] node[below] {\scriptsize $X$} (s0);
	\end{tikzpicture}
\end{center}
Similarly, the $CX$ gate from $a_{0}$ with control $s_1$
yields $a_1$ (see \cref{fig:intro-cg-vg}). Note that we do not specify the instance index of $s_1$: as any vertex on qubit $s$ with value index $1$ carries the same value, any of them can be used as a control. As $CX$ is qfree, we also record in $\graphval$ the reverse edge from $a_1$ to $a_0$, with gate $CX$ and control $s_1$. In contrast, the $CCH$ gate on qubit $r$ cannot be safely uncomputed as it is not qfree~\cite{bichsel_silq_2020}. We therefore only record the forward edge from $t_0$ to $t_1$.

\begin{figure*}[t]
    \footnotesize
    %\begin{minipage}[t]{0.4\textwidth}
    \begin{subfigure}[b]{.37\linewidth}
            \centering
            \input{algorithms/evolveVertexShortened.tex}
            \caption{Function \tfunc{evolveVertex}}
	\label{alg:evolve-vertex-simpl}
    \end{subfigure}
	\begin{subfigure}[b]{.63\linewidth}
        \centering
            \begin{tikzpicture}
        \def\tikzMatrRowSep{6mm}
        \def\tikzMatrColSep{3mm}
        \def\tikzXOffset{0.4cm}
        \def\tikzYOffset{0.6cm}
        \def\tikzMinWidth{3.05cm} % 2 * \tikzMatrixRowSep + 3 * min width of tightqballoc (=6.5mm) = 2 * \tikzMatrixRowSep + 18.5 = 34.5
        \def\tikzArrowHeadLen{8pt}
        \def\tikzArrowWidth{7pt}

        \node[fill=gray!12,anchor=north west] (cg0) at (0, 0){
			\input{figures/circgraphuncomp/c6_0.tex}
		};

		\node[fill=gray!12,anchor=north west] (cg1) at ($ (cg0.north east) + (\tikzXOffset, 0) $){
			\input{figures/circgraphuncomp/c6_1.tex}
		};

        \node[fill=gray!12,anchor=north west] (cg2) at ($ (cg1.north east) + (\tikzXOffset, 0) $){
			\input{figures/circgraphuncomp/c6_2.tex}
		};

        \node (y-pos) at ($(cg1.north)!0.85!(cg1.south)$){};
        \draw[big arrow={\tikzArrowWidth}{\tikzArrowHeadLen}, big arrow fill= gray, draw=none] (cg0.south) -- node[midway] {}  (cg0.south |- y-pos) -- (cg1.west |- y-pos);
        \draw[big arrow={\tikzArrowWidth}{\tikzArrowHeadLen}, big arrow fill= gray, draw=none] (cg1.east |- y-pos) -- (cg2.west |- y-pos);

        \end{tikzpicture}
        \caption{Using \tfunc{evolveVertex} to uncompute $a_{1.0}$ in $\un{G}$}
        \label{fig:evolveVertex-showcase}
\end{subfigure}

    \caption{\tfunc{evolveVertex} algorithm and demonstration.}
    \label{fig:demo-evolveVertex}
\end{figure*}
\subsection{Modifying a Circuit Graph}
\label{sec:circ-graphs:uncomp}

We now discuss the core operation of \reqomp, the safe extension of a circuit graph in a stepwise manner through function $\tfunc{evolveVertex}$.

\para{evolveVertex}
We show the algorithm for \tfunc{evolveVertex} in \cref{alg:evolve-vertex-simpl}. As its name suggests, it is used to evolve vertices, i.e., to bring qubits from one value index to another.
It uses the value graph $\graphval$ as a guide, and iteratively modifies a circuit graph $\graphu$. In \cref{fig:demo-evolveVertex}, $\un{G}$ is a copy of the circuit graph $G$ from \cref{fig:intro-cg-cg}, and its value graph $\graphval$ is shown in \cref{fig:intro-cg-vg}. Note that \reqomp, and hence also \tfunc{evolveVertex}, never modifies the circuit graph $G$ corresponding to its input circuit. Instead, they both work on a new circuit graph $\un{G}$, built by following $G$, as we will explain in \cref{sec:heuristics}.

\para{Calling \tfunc{evolveVertex}}\tfunc{evolveVertex} takes as input three arguments. First, $\tvar{qb}$ is the qubit on which we will insert the uncomputation. Second, $\tvar{nVId}$ is the value index we will evolve this qubit to. The last argument is $I$, a set of qubits on which vertices are currently being added---this argument is needed to avoid
infinite recursion (see also \cref{lin:evolveVertex-assertacyclic}, discussed
later). In \cref{fig:evolveVertex-showcase}, we demonstrate the example call \tfunc{evolveVertex}($a$, $0$, $\emptyset$), which uncomputes a
single gate on qubit $a$: it will bring qubit $a$ from its current value index
$1$ to $0$. The argument $\emptyset$ indicates that no vertices on any other
qubit are in the process of being added to $\un{G}$.

\para{Building the New Vertex}
\tfunc{evolveVertex} proceeds as follows.
\cref{lin:evolveVertex-last} gets the last vertex \tvar{\un{last}} on qubit $a$, that is the lowest one following target edges. \cref{lin:evolveVertex-vidx} stores its value index in variable $\tvar{oVId}$. Here \tvar{\un{last}} is $a_{1.0}$ and $\tvar{oVId}$ is $1$.
\cref{lin:evolveVertex-onestepawayassert} then checks that $nVId$, the value index we want to add to the graph, here $0$, can be reached
in just one gate step. This is the case as $a_0$ is just one $CX$ gate away
from $a_1$, as evidenced by the edge $a_{0} \xrightarrow{CX, b_0} a_{1}$ in $\graphval$ (see \cref{fig:intro-cg-vg}). \cref{lin:evolveVertex-v} then inserts a new vertex in $\un{G}$ on qubit $a$ with value index $0$ and gate $CX$. As there is already a vertex $a_{0.0}$ in $\un{G}$, it picks a new instance index, resulting in vertex $a_{0.1}$. \cref{lin:evolveVertex-e} finally links it to its predecessor with a target edge and adds the resulting anti-dependency edge $t_{1.0} \availEdge a_{0.1}$. This results in the second graph in \cref{fig:evolveVertex-showcase} (ignoring the red edges $s_{1.0} \targEdge a_{0.1} \availEdge s_{0.1}$).

\para{Adding Control Edges}
To ensure $a_{0.1}$ indeed uncomputes $a_{1.0}$, we must control $a_{0.1}$ with qubits holding the same values as were used to control $a_{1.0}$. More precisely, $a_{0.1}$ should be controlled by vertices with the same qubit and value index as those controlling $a_{1.0}$. As the set \tvar{ctrls} contains exactly those qubits and value indices (see \cref{lin:evolveVertex-onestepaway-def}), \cref{lin:evolveVertex-for-controls} simply iterates over all controls \tvar{c} in $\tvar{ctrls}$.
Through a call to the auxiliary function \tfunc{getAvailCtrl} it then gets a (potentially new) vertex \tvar{\un{c}}
(\cref{lin:callgetavailablectrl}), which should have the same value index and
qubit as \tvar{c} and be available as a
control for \tvar{\un{v}} (that is adding the control edge $\un{c} \ctrlEdge \un{v}$, and resulting anti-dependency edges does not create a cycle in $\un{G}$). Many implementations of \tfunc{getAvailCtrl} are possible, each following different heuristics. The only restriction is that any modification to $\un{G}$ must be done through a call to \tfunc{evolveVertex}. This and the assertion in \cref{lin:evolveVertex-assert-ctrls} are enough to ensure correctness of \tfunc{evolveVertex} as we will discuss in \cref{sec:circ-graph:formal}.

\para{Choosing a Control}
Let us manually follow the implementation of \tfunc{getAvailCtrl}\footnote{We describe this function in \cref{sec:heuristics:compl-graph} and show it in \cref{alg:reqomp-convenience}.}. The only control in $\tvar{ctrls}$ is $s_1$. We first check if an existing vertex in $\un{G}$ with qubit $s$ and value index $1$ could be used. This is not the case, as using $s_{1.0}$ as control for $a_{0.1}$ would result in a cycle, as shown in the second graph in \cref{fig:evolveVertex-showcase}. We must hence compute a new vertex on qubit $s$ with value index $1$. We do so by calling $\tfunc{evolveVertex}(s, 1, \{a\})$. Note how the set of qubits under construction contains $a$, as this is a recursive call within the computation of $a_{0.1}$. We can finally link $s_{1.1}$ to $a_{0.1}$ with a control edge, concluding the computation and yielding the third graph in \cref{fig:evolveVertex-showcase}.

\para{Avoiding Infinite Recursion}
The assertion in \cref{lin:evolveVertex-assertacyclic} ensures that we never
call \tfunc{evolveVertex} recursively on the same qubit.
This avoids infinite recursion where two qubits keep triggering recomputation of
the other.
To this end, we propagate the set \tvar{I} of qubits currently under
construction through \tfunc{getAvailCtrl} to potential recursive calls into
\tfunc{evolveVertex} (see \cref{lin:callgetavailablectrl}).

\para{Modified Control}
Here the uncomputation of ancilla $a$ (introducing $a_{0.1}$ in $\un{G}$) resulted in the modification of its control qubit $s$ (introducing $s_{1.1}$). The circuit $\un{C}$ corresponding\footnote{This is a slight simplification. Multiple circuits $\un{C}$ may correspond to the circuit graph $\un{G}$. They all have the same semantics, and will all be incorrect uncomputations of $C$.} to $\un{G}$ is hence not a correct uncomputation of $C$. While ancilla $a$ has been correctly brought back to its initial value $0$, the value of qubit $s$ has been modified and does not hold the same value as it would in $C$.  
%This illustrates the need for a precise tracking of qubit values through the circuit. 
We will show in \cref{sec:heuristics:compl-graph} how \reqomp notices and fixes such a value mismatch to ensure correct uncomputation.

\subsection{Formalizing Value Indices}
\label{sec:circ-graph:formal}
%The above valid circuit graph definition given above is not enough to precisely follow qubit values through uncomputation and recomputation. We therefore need to make sure that the vertices value indices track qubit values. However, this cannot be done in isolation, as a qubit may be entangled with the remainder of the state. 
So far, we relied on an intuitive understanding of well-valued circuit graphs, and used it to build uncomputation for a circuit. Let us now formalize this intuition in \cref{def:well-valued}.

\begin{restatable}[Well-valued Circuit Graph]{mydef}{defvalid}
	\label{def:well-valued}
    We say a valid circuit graph is well valued iff:
	\begin{enumerate}[label=(\roman*)]
		\item all vertex names are of the form $q_{s.i}$ where $q$ is the name
		of the vertex qubit, $s$ and $i$ are natural numbers
		\item there are no duplicate vertices
		\item the init vertex on each qubit is named $q_{0.0}$ and for any
		$q_{s.i}$ in $G$, $q_{s.0}$ is in $G$
		\item \label{def:well-val-4}any gate vertex $q_{s.i}$ with $s>0$ satisfies one of the following:
		\subitem \textbf{(fwd)}
		$\tfunc{valIdx}(\tfunc{pred}(q_{s.i})) =
		\tfunc{valIdx}(\tfunc{pred}(q_{s.0}))$ and $q_{s.i}$ and $q_{s.0}$ have
		the same gate and same control vertices (up to their instance indices)
		\subitem \textbf{(bwd)} if we denote $s' = \tfunc{valIdx}(\tfunc{pred}(q_{s.i}))$,
		we have that  (i)~$\tfunc{valIdx}(\tfunc{pred}(q_{s'.0})) = s$, (ii)~$q_{s.i}.gate$ is qfree and equal to $q_{s'.0}.gate^\dagger$, and (iii)~both $q_{s.i}$ and $q_{s'.0}$ have the same controls (up to instance indices).
\end{enumerate}
\end{restatable}
Here, $\tfunc{pred}(v)$ is the unique $v'$ such that $v' \targEdge v$ and $\tfunc{valIdx}(v)$ is the value index of $v$.
We now give some intuition of the last condition (iv).
Case \textbf{(fwd)} corresponds to a (forward) computation, for instance
$s_{1.0}$ and $s_{1.1}$ in the last graph in \cref{fig:evolveVertex-showcase}. Here, case (fwd) ensures that $s_{0.0} \targEdge s_{1.0}$ and $s_{0.1} \targEdge s_{1.1}$ apply the same
operation to the same starting state, i.e. in both cases qubit $s$ holds the same value before the operation is applied. This is the case as both $s_{1.0}$ and
$s_{1.1}$ have the same gate $X$, and their predecessors ($s_{0.0}$ and
$s_{0.1}$) have the same value index $0$.
Case \textbf{(bwd)} corresponds to a (backward) uncomputation, for instance
$a_{1.0}$ and $a_{0.1}$. Here, case (bwd) ensures
that the operations $a_{0.0} \targEdge a_{1.0}$ and $a_{1.0} \targEdge
a_{0.1}$ are exact inverses of each other.
Specifically, it ensures that (i)~$a_{0.1}$ and the predecessor of $a_{1.0}$
(here $a_{0.0}$) have the same value index~$0$, (ii)~the gates of $a_{1.0}$ and
$a_{0.1}$ are inverses of each other, and (iii)~their controls ($s_{1.0}$ and
$s_{1.1}$) have the same qubit and value index.

\para{Preserving Values}
Any circuit graph verifying this definition ensures the following: if a qubit is in some basis state at a value index, then if the
qubit reaches the same value index at a later point in time, it will again be in
this same basis state.
Or, more formally:
for every pair of vertices
$q_{s.i}$ and $q_{s.i'}$, applying all gates $\graph'$ between these vertices
should preserve $q$, in the following sense:
\begin{minipage}{\linewidth}
	\footnotesize
	\begin{align}\label{eq:value-index-meaning}
		\begin{array}{l}
			\forall b \in \{0,1\}. \\
			\forall \varphi \in \hil{n-1}.
		\end{array}
		\exists \psi \in \hil{n-1}.
		\ket{b}_q \otimes \varphi \xmapsto{\llbracket \graph' \rrbracket} \ket{b}_q \otimes \psi,
	\end{align}
\end{minipage}
where $\hil{n-1}$ denotes the set of quantum states over $n-1$ qubits.
As we can write any state as a sum of computational basis states, \cref{eq:value-index-meaning} allows us to reason about any state.
We show in \cref{app:proof} that any well-valued circuit graph ensures
\cref{eq:value-index-meaning} (more precisely, \cref{lem:nullprojcoeffs} implies
this).

\para{evolveVertex}
In \cref{sec:circ-graphs:uncomp}, we claimed that \tfunc{evolveVertex} preserves the well-valuedness of a circuit graph. More precisely, if \tfunc{evolveVertex} terminates without error, the resulting circuit graph $\un{G}$ is still well-valued (assuming it was before the call).
This is ensured through the two assertions in \cref{sub@lin:evolve-vtx-assert-well-val} and \cref{lin:evolveVertex-assert-ctrls}. We now specify those assertions.
The first, \textsc{well\_valued\_vertex} requires that one of the following conditions is satisfied:
\begin{enumerate}[itemindent=2em]
	\item[\textbf{(fwd)}]there exists $i$ such that $\tvar{qb_{oVId.i}} \targEdge \tvar{qb_{nVId.0}}$ is in $\un{G}$ and $gt$ and $ctrls$ correspond to the gate and controls of \tvar{qb_{nVId.0}};
	\item[\textbf{(bwd)}] there exists $i$ such that $\tvar{qb_{nVId.i}} \targEdge \tvar{qb_{oVId.0}}$ is in $\un{G}$, $gt$ is qfree and equal to $\tvar{qb_{oVId.i}}.gate^\dagger$ and $ctrls$ correspond to the controls of \tvar{qb_{oVId.0}};
\end{enumerate}
The second, \textsc{correct\_control} requires that $\un{c}$ has the same qubit and value index as $c$, and that adding the control edge $\un{c} \ctrlEdge \un{v}$ (and all resulting anti-dependency edges) does not create a cycle in $\un{G}$.
Taken together, these assertions correspond exactly to the definition of a well-valued graph, and ensuring that it stays acyclic. Further, neither of those assertions refer to the value graph $\graphval$ and both avoid reliance on the function \tfunc{getAvailCtrl}. This allows for a self-contained definition of well-valued graphs and simpler correctness proof, which we discuss in \cref{sec:correctness}. 
%These assertions may seem redundant with the assertion in \cref{lin:evolveVertex-onestepawayassert} and what the function \tfunc{getAvailCtrl} should ensure. 

%\textsc{well\_valued\_vertex} is redundant with the assertion in \cref{lin:evolveVertex-onestepawayassert} and \textsc{correct\_control} is redundant with what \tfunc{getAvailCtrl} should ensure. 
%Further, neither of them refer to the value graph $\graphval$, allowing for a simpler correctness proof, which we discuss in \cref{sec:correctness}. 

\para{Contributions to Circuit Graphs}
In the following, we briefly elaborate on the main differences between the circuit graphs introduced in \unqomp and our new notion of a well-valued circuit graph.
\unqomp does not use value indices. Instead, uncomputation is built by adding to the circuit graph built from the input circuit an uncomputation vertex for each vertex on an ancilla. Correctness of the uncomputation then relies on this one computation vertex to one uncomputation vertex correspondence in the final circuit graph. This one to one correspondence fundamentally does not allow for recomputation, where three or more vertices may correspond to the same value. In contrast, we introduce the notion of value indices, and prove formally (see \cref{sec:correctness}) that they accurately track values in qubits. We further introduce the notion of a value graph and the function \tfunc{evolveVertex}, which leverages value indices to build correct computation and uncomputation in a circuit graph.

\section{Reqomp}
\label{sec:heuristics}
\begin{figure*}
    \begin{tikzpicture}
        \def\tikzTitleColor{gray!30}
        \def\tikzArrowHeadLen{11pt}
        \def\tikzArrowWidth{9pt}
		\node[fill=gray!12,anchor=north west] (input_circ) at (0, 5){
            \begin{tikzpicture}
                \node[anchor=north west] (input_circ_circ){
                    \begin{tikzpicture}
                        \input{figures/reqomp-fig/circuit.tex}
                    \end{tikzpicture}
                };
                \node[fill=\tikzTitleColor, anchor=north] (input_circ_title) at (input_circ_circ.south) {a) Input Circuit $C$};
            \end{tikzpicture}
		};

        \node[fill=gray!12,anchor=north west] (circ-graph) at ($ (input_circ.north east) + (1cm, 0) $) {
            \begin{tikzpicture}
                \node[anchor=south west] (circ_graph_graph){
                    \input{figures/reqomp-fig/circgraph.tex}
                };
                \node[fill=\tikzTitleColor, anchor=north] (circ_graph_title) at (circ_graph_graph.south |- input_circ.south) {b) Circuit Graph $G$ and value graph $g^{val}$ (omitted)};
            \end{tikzpicture}
			
		};

        \node[fill=gray!12,anchor=north west] (n_qbs) at ($ (input_circ.south west) + (0, -0.7cm) $) {
            \begin{tikzpicture}
            \node[anchor=north west] (n_qbs_n){ 2 };

            \node[fill=\tikzTitleColor, anchor=north, align=center] (n_qbs_title) at (n_qbs_n.south) {\# Ancilla\\ Qubits};
        \end{tikzpicture}
		};

        \node[fill=gray!12,anchor=north west] (dep_graph) at ($ (n_qbs.north east) + (1cm, 0) $) {
            \begin{tikzpicture}
            \node[anchor=north west] (dep_graph_graph){
			\begin{tikzpicture}[node distance=0.3cm]
                \node[tightinit, fill=my-full-red!40] (a) {$a$};
                \node[tightinit,right= of a, fill=my-full-red!40] (b) {$b$};
                \node[tightinit,right= of b, fill=my-full-red!40] (c) {$c$};
                \node[right= of c] (phantom) {};
                \node[tightinit,right= of phantom,fill=my-full-red!40] (d) {$d$};
                \node[tightinit,right= of d, fill=my-full-red!40] (e) {$e$};

                \draw[-,line width=0.2cm, color = my-full-blue!60] (a) -- (b);
                \draw[-,line width=0.2cm, color = my-full-blue!60] (b) -- (c);
                \draw[-,line width=0.2cm, color = my-full-blue!60] (d) -- (e);
        
                \draw[style=targ] (a) -- (b);
                \draw[style=targ] (b) -- (c);
                \draw[style=targ] (d) -- (e);

                \begin{scope}[on background layer]
                    \node[tightcomponent, anchor = center, fill=gray!30] (comp1) at ($(a)!0.5!(b)$) {};
                    \node[tightcomponent, anchor = center,fill = gray!30] (comp1b) at ($(b)!0.5!(c)$) {};
                    \node[tightcomponent, anchor = center, fill = gray!30] (comp2) at ($(d)!0.5!(e)$) {};
                    \draw[style=cc-edge] (comp1b) -- (comp2);
                \end{scope}
            \end{tikzpicture}
            };

            \node[fill=\tikzTitleColor, anchor=north] (dep_graph_title) at (dep_graph_graph.south) {c) Ancilla Dependencies};
        \end{tikzpicture}
		};

        \node[fill=gray!12,anchor=north west] (uncomp_strat) at ($ (n_qbs.south west) + (0, -0.7cm) $) {
            \begin{tikzpicture}
            \node[anchor=north west] (strat){$a, b, a^\dagger, c, c^\dagger, a, b^\dagger, a^\dagger, d, e, e^\dagger, d^\dagger$};

            \node[fill=\tikzTitleColor, anchor=north] (strat_title) at (strat.south) {d) Uncomputation Strategy};
        \end{tikzpicture}
		};

        \node[fill=gray!12,anchor=north west] (circ_g_uncomp) at ($ (uncomp_strat.north east) + (0.5cm, 0) $) {
            \begin{tikzpicture}
            \node[anchor=north west] (cg_uncomp){(omitted)\phantom{$^\dagger$}};

            \node[fill=\tikzTitleColor, anchor=north] (cg_uncomp_title) at (cg_uncomp.south) {e) Circuit Graph $\un{G}$};
        \end{tikzpicture}
		};

		\node[fill=gray!12,anchor=north east] (unqomp_circ) at (dep_graph.north -| circ-graph.east){
            \begin{tikzpicture}
                \node[anchor=south west] (uncomp_circ_circ){
                    \begin{tikzpicture}
                        \input{figures/reqomp-fig/circuit_w_unqomp_highlights.tex}
                    \end{tikzpicture}
                };
                \node[fill=\tikzTitleColor, anchor = north] (uncomp_circ_title_title) at ($ (uncomp_circ_circ.south) + (0, -0.2cm) $) {f) Circuit With Uncomputation $\un{C}$};
            \end{tikzpicture}
		};

        \draw[big arrow={\tikzArrowWidth}{\tikzArrowHeadLen}, big arrow fill= gray, draw=none] (input_circ.east) -- (circ-graph.west);

        \node (arrow1-x-start) at ($(circ-graph.west)!0.1!(circ-graph.east)$) {};
        \draw[big arrow={\tikzArrowWidth}{\tikzArrowHeadLen}, big arrow fill= gray, draw=none] (arrow1-x-start |- circ-graph.south) -- (arrow1-x-start |- dep_graph.north);

        \node (arrow2-x-start) at ($(dep_graph.west)!0.2!(dep_graph.east)$) {};
        \draw[big arrow={\tikzArrowWidth}{\tikzArrowHeadLen}, big arrow fill= gray, draw=none] (n_qbs.south) -- (n_qbs.south |- uncomp_strat.north);

        \draw[big arrow={\tikzArrowWidth}{\tikzArrowHeadLen}, big arrow fill= gray, draw=none] (arrow2-x-start |- dep_graph.south) -- (arrow2-x-start |- uncomp_strat.north);

        \draw[big arrow={\tikzArrowWidth}{\tikzArrowHeadLen}, big arrow fill= gray, draw=none] (uncomp_strat.east) -- (circ_g_uncomp.west |- uncomp_strat.east);

        \draw[big arrow={\tikzArrowWidth}{\tikzArrowHeadLen}, big arrow fill= gray, draw=none] (circ_g_uncomp.east |- uncomp_strat.east) -- (unqomp_circ.west |- uncomp_strat.east);

        %\node (arrow1-x-start) at ($ (circuitc.east) + (0, \tikzTitleSepDefault) $) {};
		%\node (arrow1-y-end) at ($ (circ-graph.north) + (0, 0.2cm) $) {};
		%\node (arrow1-x-end) at ($(circ-graph-title.west)!0.2!(circ-graph-title.east)$) {};
		%\draw[big arrow={\tikzArrowWidth}{\tikzArrowHeadLen}, big arrow fill= gray, draw=none] (arrow1-x-start) -- node[midway] {}  (arrow1-x-start -| arrow1-x-end) -- node[midway] {} (arrow1-x-end |- arrow1-y-end);

        %\include{figures/circGraphsIntro/val_graph_simple}
        \end{tikzpicture}
        \begin{subfigure}[t]{0pt}
            \phantomcaption\label{fig:reqomp-circ}
        \end{subfigure}
        \begin{subfigure}[t]{0pt}
            \phantomcaption\label{fig:reqomp-cg}
        \end{subfigure}
        \begin{subfigure}[t]{0pt}
            \phantomcaption\label{fig:reqomp-anc-deps}
        \end{subfigure}
        \begin{subfigure}[t]{0pt}
            \phantomcaption\label{fig:reqomp-strat}
        \end{subfigure}
        \begin{subfigure}[t]{0pt}
            \phantomcaption\label{fig:reqomp-cgu}
        \end{subfigure}
        \begin{subfigure}[t]{0pt}
            \phantomcaption\label{fig:reqomp-circu}
        \end{subfigure}
        \caption{Overview of Reqomp}
	\label{fig:reqomp}
\end{figure*}

The previous section presented our notion of well-valued circuit graphs, and how they can be used to insert computation or uncomputation on any qubit. We now take a step back and present the complete \reqomp procedure, and how it leverages well-valued circuit graphs and the \tfunc{evolveVertex} function to tackle the problem of ancilla variables uncomputation under space constraints. As mentioned in \cref{sec:intro}, \reqomp takes as input a quantum circuit and a number of available ancilla qubits.
If successful, it returns a quantum circuit where all ancilla variables from the original circuit are uncomputed and all other variables are preserved, using only the number of available ancilla qubits.

\para{Example Circuit}
\cref{fig:reqomp} gives an overview of Reqomp and applies it to an example circuit with five ancilla variables, $a, b, c, d$, and $e$, and two non ancilla variables $r$ and $t$. We note that while this circuit does not implement a relevant algorithm, it allows showcasing the key features of Reqomp on a
simple example.
%
%We first present a high level overview of Reqomp in
% \cref{sec:overview-workflow}, before going into more details about each of its
% steps on a specific example in \cref{sec:overview-example} to
% \cref{sec:overview-conv-to-circ}.

\para{\reqomp Workflow}
\cref{fig:reqomp} shows the steps performed by \reqomp, which we detail below.

First, \reqomp converts the circuit $C$ into a circuit graph $G$ and a value graph $\graphval$ (see \cref{fig:reqomp-cg}).
Using this representation, \reqomp identifies the dependencies among ancilla
variables in the circuit (see \cref{sec:reqomp-id-anc} and \cref{fig:reqomp-anc-deps}), and uses them to derive an uncomputation strategy
respecting the number of available ancilla qubits (see \cref{sec:reqomp-derive-strategy} and \cref{fig:reqomp-strat}).
\reqomp then applies this strategy to build a new circuit graph $\graphu$
containing uncomputation (see \cref{sec:heuristics:compl-graph} and \cref{fig:reqomp-cgu}). % (\cref{sec:overview-applying-strategy} and\legendboxbig{my-orange} in \cref{fig:overview}).
Finally, \reqomp converts the resulting circuit graph into a circuit $\un{C}$ (see \cref{sec:reqomp-conv-to-circ} and \cref{fig:reqomp-circu}).

\subsection{Identifying Ancilla Variables Dependencies}
\label{sec:reqomp-id-anc}
The first step of \reqomp is converting the input circuit into a circuit graph $G$, as we described in \cref{sec:circ-graphs}.
Using this circuit graph $G$, \reqomp then identifies ancilla dependencies.

On the circuit graph $G$ from \cref{fig:reqomp-cg}, \reqomp identifies all
ancilla variables vertices (highlighted in red) and their dependencies (highlighted in
blue), and extracts the ancilla dependencies shown in \cref{fig:reqomp-anc-deps}.
There, each vertex corresponds to an ancilla variable
and each solid edge corresponds to a control edge among gate vertices between these
respective ancilla variables. We will discuss the dotted edge shortly.

\subsection{Deriving the Uncomputation Strategy}
\label{sec:reqomp-derive-strategy}
Based on the ancilla dependencies derived above, \reqomp will derive an \emph{uncomputation strategy}.

\paragraph{What is an Uncomputation Strategy?}
The uncomputation strategy describes in which order the ancillae in the circuit should be computed and uncomputed, to satisfy the space constraints that were given as input, while minimizing the number of gates in the circuit. For instance \cref{fig:pb-stmt-big-circ} showcases two different such strategies. The first one, shown in \cref{fig:pb-stmt-big-circ-lazy}, is to compute ancilla $a$, then $b$, then $c$, then uncompute $c$, then $b$, then finally $a$. We typically write such a strategy as $a, b, c, c^\dagger, b^\dagger, a^\dagger$, where we write $a$ to denote "computing ancilla $a$", and $a^\dagger$ to denote "uncomputing ancilla $a$". The second strategy, shown in \cref{fig:pb-stmt-big-circ-eager} is $a, b, a^\dagger, c, c^\dagger, a, b^\dagger, a^\dagger$. %Those strategies have different trade-offs. The first one is shorter, and therefore results in less gates in the circuit, but requires 3 ancilla qubits. In contrast, the second strategy is longer, but requires only 2 ancilla qubits, as it is able to reuse one.

\para{Partitioning Ancillae}
The first step in deriving such an uncomputation strategy from the ancilla dependencies is to distinguish groups of ancillae variables that depend on each other; in other words, partitioning the ancillae according to their dependencies.
More precisely, \reqomp identifies ancilla variables that do not interact with each
other (i.e., lie in different connected components of the ancilla dependency
graph), and can therefore be computed and uncomputed independently. For instance in \cref{fig:reqomp-anc-deps}, ancillae variables $a, b$ and $c$ belong to the same connected component highlighted in dark gray, while ancillae variables $d$ and $e$ are in another component.
%
%This allows us to then process the individual partitions in a suitable order,
%deriving and applying a separate uncomputation strategy for each.
%

\paragraph{Why Partition Ancillae?}
\reqomp aims at balancing ancilla qubits and gates.
For two ancillae that do not interact, such a trade-off is easy: we should
always uncompute the first ancilla early, making its qubit available for the
latter one.
As the ancillae are independent, the latter one does not need the earlier one,
so the early uncomputation will not induce extra gates, i.e., no recomputation is necessary.
For instance, in \cref{fig:reqomp-anc-deps}, ancillae $\{a, b, c\}$ and $\{d, e\}$
are independent. Therefore, it is strictly better to uncompute $a, b$ and $c$ before
computing $d$ and $e$, thereby reusing the physical ancilla qubits initially
holding $a, b$ and $c$ for $d$ and $e$.
This is in contrast to ancillae that are part of the same partition. For
instance, in \cref{fig:pb-stmt-big-circ}, we saw that for the 3 linked ancillae $a, b, c$, uncomputing early yields a different trade-off than uncomputing late.

% Partitioning ancillae allows us to clearly separate independent clusters of
% ancillae to be computed and uncomputed completely, before moving on to the next
% cluster. 

\paragraph{Strategy for a Connected Component}
Now that we have split the ancillae variables in two components, let us derive the uncomputation strategy for each of them. We first note that within each of the components in \cref{fig:reqomp-anc-deps}, the ancilla variables exhibit a linear
dependency. Formally, we say ancillae $a^{1}, \dots, a^{n}$ are linearly dependent if all gates
targeting $a^{i}$ for $i>1$ are only controlled by $a^{i-1}$ and non-ancillae. This corresponds to a component that forms a simple path. In this case, we can derive an optimal uncomputation strategy (in terms of number of computation/uncomputation steps) using dynamic programming~\cite{knill_analysis_1995}~\footnote{We show our implementation of this method in \cref{alg:linear-steps}.}. For the first component ($\{a, b, c\}$) on at most two ancilla qubits, the following optimal strategy is found:
\begin{align} \label{eq:uncomputation-strategy}
    a, b, a^\dagger, c, c^\dagger, a, b^\dagger, a^\dagger.
\end{align}
For the second component ($\{d, e\}$), also on at most two ancilla qubits, the following optimal strategy is found:
\begin{align} \label{eq:uncomputation-strategy}
    d, e, e^\dagger,d^\dagger.
\end{align}

If within a component the ancilla variables are not linearly dependent, we abort the current procedure, and fall back on an alternative one, Reqomp-Lazy, which we describe in \cref{sec:reqomp-lazy}. We note that we avoid solving the general problem of finding an uncomputation
strategy for any ancilla dependency, as it is P-SPACE
complete~\cite{pspace-compl}.

\paragraph{Combining Strategies}
%
%\cref{fig:partition-graph} shows an example graph with four ancillae $a$, $b$, $c$, and $d$.
%
%We can see that ancillae $\{a,b\}$ do not directly interact with ancillae $\{c,d\}$, and should therefore be two different partitions.
%
%We now explain how we identify such partitions based on the circuit graph.

\newcommand{\inlineinit}[1]{%
	\scalebox{0.6}{\begin{tikzpicture}\node[init] {#1};\end{tikzpicture}}%
}
\newcommand{\inlinetargetedge}[0]{%
	\begin{tikzpicture}
		\node[] (a) at (0,0) {};
		\node[] (b) at (0.7,0) {};
		\draw[style=targ] (a) -- (b);
	\end{tikzpicture}%
}
\newcommand{\inlineccedge}[0]{%
	\begin{tikzpicture}
		\node[] (a) at (0,0) {};
		\node[] (b) at (0.7,0) {};
		\draw[style=cc-edge] (a) -- (b);
	\end{tikzpicture}%
}

After determining the optimal uncomputation strategy for each connected component, we must combine those strategies to yield our complete strategy. To this end, we determine in which order the components should be processed, ensuring that if an ancilla $d$ transitively depends
on ancilla $c$, $c$'s component is processed before $d$'s component (captured by
edges \inlineccedge in \cref{fig:reqomp-anc-deps}). In our case, we must process the component with ancilla variables $a, b, c$ before the one with ancilla variables $d, e$, as $d$ depends on $c$ through the qubit $r$. Combining the respective strategies in this order, we finally get the complete uncomputation strategy shown in \cref{fig:reqomp-strat}. 
If ordering the connected
components is impossible due to cycles, we again fall back to Reqomp-Lazy.%mentioned above. %For completeness, \cref{app:partitioning}
%provides an implementation of \tfunc{PartitionAncillae}.

\subsection{Applying the Uncomputation Strategy}
\label{sec:heuristics:compl-graph}
\begin{figure}
    \footnotesize
	\begin{algorithmic}[1]
		%\State%%%%%%%%%%%%%%%%%%%%%%%%%%%%%%%%%%%%%%
		\continueLineNumber
		\Function{ApplyingStrategy}{\tvar{stages}}
			\State $\un{G} \gets \tfunc{initGraph}(G.\tvar{qbs})$ \label{lin:apply-strat-initg}
            \For{$\tvar{anc},\tvar{fwd} \in \tvar{stages}$}\label{lin:reqomp-steps-start}
					\If{\tvar{fwd}}
						\State $\tvar{ancQb} \gets \tfunc{getFreeAncillaQubit}()$ \label{lin:reqomp-getphysQb}
						\State $i \gets \graphPrim{freshInstanceId}{\graphu}{\tvar{anc}, 0}$\label{lin:reqomp-getfreshinstance}
						\State $\graphPrim{addVertex}{\graphu}{\tvar{anc}_{0. i}, \tvar{ancQb}}$\label{lin:reqomp-initnodeaddedtograph}
						\State $\graphPrim{addEdge}{\graphu}{\graphPrim{getLastOnQb}{\graphu}{\tvar{ancQb}} \targEdge \tvar{anc}_{0, i}}$\label{lin:reqomp-initnodelinkedtolast}
                    \EndIf
                    \If{\tvar{fwd}}\label{lin:reqomp-detailed-steps-beg}
					\State \tfunc{evolveVtxUntil}($\tvar{anc}, \tfunc{max}\{s \mid anc_s \in \graphval\}$)
                    \Else
                    \State \tfunc{evolveVtxUntil}($\tvar{anc}, 0$)\label{lin:reqomp-detailed-steps-end}
					\EndIf
                    
                    \If{\tvar{s} is last fwd stage in \tvar{stages} acting on $\tvar{anc}$}\label{lin:reqomp-nonAncillas-start}
					%\Comment{\todo{show necessity on example? It is similar to evolveNonAncillas, so we can explain it that way}}
					\For{$t_{s}$ such that $\exists t_{s'} \xrightarrow{\tvar{gate}, .., \tvar{anc}_{i}} t_s$ in \graphval}\label{lin:reqomp:non-anc-beg}
					\If{$t$ is not an ancilla and $t_{s.0} \notin \un{G}$}
                    \State \tfunc{evolveVtxUntil}($\tvar{t}, s$)\label{lin:reqomp:non-anc-end}
					\EndIf
                    \EndFor
					\EndIf
					\If{not \tvar{fwd}}
					\State $\tfunc{freeAncillaQubit}(\tvar{ancQb})$\label{lin:reqomp-freephysqb}
					\EndIf
			\EndFor%\label{lin:reqomp-steps-end}
			\For{$v_{s.i}$ final vertex in $G$}\label{lin:reqomp:non-anc2-beg} %\Comment{\todo{move "and" to an if? Probably more clear}}
			\If{$v$ not an ancilla}
			\State \tfunc{evolveVtxUntil}($v$, $s$)\label{lin:reqomp:non-anc2-end}
			\EndIf
			\EndFor
			\State \tfunc{assertFullyEvolved}()\label{lin:reqomp-assertFullyEvolved}
		\EndFunction

		\State%%%%%%%%%%%%%%%%%%%%%%%%%%%%%%%%%%%%%%

		\Function{evolveVtxUntil}{$\tvar{var}$, \tvar{to}}\label{lin:evolvevtxUntil}
		\State $\tvar{from} \gets \graphPrim{lastOnQb}{\graphu}{\tvar{var}}.\tvar{valIdx}$\label{lin:evolvevtxUntil:getfrom}
		\State $\tvar{steps} \gets \tfunc{getPath}(\tvar{var}_\tvar{from}, \tvar{var}_{\tvar{to}}, \graphval)$\label{lin:evolvevtxUntil:getpath}
		\For{$\tvar{step} \in \tvar{steps}$}\label{lin:evolvevtxUntil:applypath-beg}
                    \State $\tfunc{evolveVertex}(\tvar{anc}, \tvar{step}, \emptyset)$\label{lin:evolvevtxUntil:applypath-end}
        \EndFor
		\EndFunction
		\recordLineNumber
	\end{algorithmic}
    \caption{Applying the uncomputation strategy. We assume $G, \un{G}$, and $\graphval$ are globally available.}
	\label{alg:apply-start}
\end{figure}

We showed in the previous section how \reqomp derives the uncomputation strategy for a given circuit. Further, we have shown in \cref{sec:circ-graphs:uncomp} how we could use \tfunc{evolveVertex} to insert computation or uncomputation in a circuit graph. 
However, there is a gap between the uncomputation strategy and \tfunc{evolveVertex}: the former does not mention any non-ancilla variables, nor any value indices, which are a required argument of \tfunc{evolveVertex}. 
\cref{alg:apply-start} bridges this gap by showing how \reqomp translates the uncomputation strategy into a series of calls to \tfunc{evolveVertex}, which will build a new circuit graph $\un{G}$ with uncomputation.

\para{A New Circuit Graph}
\reqomp does not insert the uncomputation directly in $G$, the circuit graph built from the input circuit $C$. Instead, it builds a new graph $\un{G}$ from scratch, adding computation and uncomputation on all qubits step by step. $\un{G}$ is initialized in \cref{lin:apply-strat-initg}. Initially, it contains one init vertex for each non ancilla qubit in $G$. For the circuit graph $G$ shown in \cref{fig:reqomp-cg}, this results in the following graph $\un{G}$:

\begin{center}
    \begin{tikzpicture}
		%%%%%%%%%
		% NODES %
		%%%%%%%%%
	
		\matrix[row sep=0mm,column sep=1mm] {
			\node (s00) [style=tightgate] {$s_{0.0}$}; &
			%\node (a00) [style=qballoc] {\phantom{$a_{0.0}$}}; &
			%\node (b00) [style=qballoc] {\phantom{$b_{0.0}$}}; &
			\node (t00) [style=tightgate] {$r_{0.0}$}; \\
		};
	
	\end{tikzpicture}
	
\end{center}

\para{Qubits Allocation}
The strategy consists of a sequence of stages, where each stage fully computes an ancilla or fully uncomputes it. Each stage is described by the ancilla it concerns, and a boolean $\tvar{fwd}$ that is true for a computation step, and false for an uncomputation one. For instance the first stage in the strategy shown in \cref{fig:reqomp-strat} is ($a$, True), that is computing ancilla $a$. For a computation stage, the first step is to allocate a qubit (\cref{lin:reqomp-getphysQb}) and create a new vertex on this qubit (\Crefrange{lin:reqomp-getfreshinstance}{lin:reqomp-initnodeaddedtograph}). This new vertex is then linked to the last vertex on the same qubit, if it exists, with a target edge in \cref{lin:reqomp-initnodelinkedtolast}. This is typically the case if the qubit was previously used to compute and uncompute another ancilla. Further, at the very end of each stage, if it was an uncomputation stage, the qubit is marked as freed and therefore can be reused for later stages, see \cref{lin:reqomp-freephysqb}.

\para{Detailed Steps for Ancilla (Un)computation}
Now that the ancilla has been allocated a qubit if necessary, \reqomp computes the detailed computation or uncomputation steps the current stage requires, in \crefrange{lin:reqomp-detailed-steps-beg}{lin:reqomp-detailed-steps-end}. 
First, \reqomp decides on what is the objective value index for the current stage. If this is an uncomputation, it is 0, as the ancilla should be uncomputed back to its initial value. If this is a computation step, the ancilla should be computed until its maximum value index in $G$. For instance, for ancilla $a$, this maximum value is $1$.

\para{EvolveVtxUntil}
To compute an ancilla to the objective value index chosen above, \reqomp relies on the function \tfunc{evolveVtxUntil}, shown in \cref{lin:evolvevtxUntil}. This function first determines what the current value index of the variable is in $\un{G}$, that is to say what is the value index of the last vertex with qubit $\tvar{var}$. Using $\graphval$, the function then determines the intermediate computation steps required to bring $\tvar{var}$ from its current value index $\tvar{from}$ to the objective one $\tvar{to}$. This is simply the shortest path in $\graphval$ from $\tvar{var}_\tvar{from}$ to $\tvar{var}_\tvar{to}$. Those steps can then be applied in order, using the function $\tfunc{evolveVertex}$, which we discussed in \cref{sec:circ-graphs:uncomp}.

\para{Non Ancilla Variables}
We explained above how the uncomputation strategy can be detailed for non ancilla variables. Let us now describe how non ancilla variables are computed. This is done in two places. First, at the end of the last forward stage on an ancilla $\tvar{anc}$, anything controlled by this ancilla is computed. More precisely, we want to compute all $t_s$, where $t$ is a non ancilla variable and $s$ a value index such that the computation of $t_s$ is controlled by some $\tvar{anc}_i$, that is to say that there exists some edge $t_{s'} \xrightarrow{\tvar{gate}, .., \tvar{anc}_{i}} t_s$ in $\graphval$. This is shown in \crefrange{lin:reqomp:non-anc-beg}{lin:reqomp:non-anc-end}, and the required computation steps are again computed and applied through the function \tfunc{evolveVtxUntil}\footnote{As detailed in
\cref{app:reqomp-convenience-methods} (\cref{alg:reqomp-convenience}), we may
force extra steps in the computation to ensure some values are computed at least once for non ancilla variables.}. Non ancilla variables are also computed at the very end of the strategy, in \crefrange{lin:reqomp:non-anc2-beg}{lin:reqomp:non-anc2-end}. Here any non ancilla variable whose final value index in $\un{G}$ is not the same as in $G$ is computed to its final value index in $G$.

\para{\tfunc{getAvailCtrl}}
When we introduced \tfunc{evolveVertex} in \cref{sec:circ-graphs:uncomp}, we mentioned that it relies on an auxiliary function \tfunc{getAvailCtrl} to get the controls required for a vertex. Any heuristic can be used for this function, as long as it only modifies $\un{G}$ through \tfunc{evolveVertex}. \reqomp uses the following heuristic\footnote{\cref{alg:reqomp-convenience} shows this implementation of \tfunc{getAvailCtrl}.}. Suppose we need a control $c_s$ (that is value index $s$ on qubit $c$) to be used to control some vertex $\un{v}$. We first find the latest (that is the lowest following target edges) vertex with this qubit and value index in $\un{G}$. If this vertex is available for $\un{v}$, that is adding a control edge from this vertex to $\un{v}$ does not create a cycle in $\un{G}$, we return it. If this vertex is not available, or if no such vertex exists, we recursively call \tfunc{evolveVertex} to build a new vertex on qubit $c$ with value index $s$ from the latest vertex on qubit $c$.

\para{Asserting Uncomputation is Complete}
After the uncomputation strategy has been applied as described above, \cref{lin:reqomp-assertFullyEvolved} performs a final check~\footnote{\cref{alg:reqomp-convenience} shows the implementation of \tfunc{assertFullyEvolved}.}. It asserts that all variables are fully evolved, either back to their initial state index $0$ (for ancilla variables), or to their final state index in $\graph$ (for non-ancillae). If not, \reqomp falls back to the alternative Reqomp-Lazy strategy (see \cref{sec:reqomp-lazy}).

\subsection{Obtaining the Final Circuit}
\label{sec:reqomp-conv-to-circ}
If the above check succeeded, the final step of the algorithm converts the circuit graph $\un{G}$ to a circuit $\un{C}$.
During this step, we perform a generic post-processing optimization, previously
discussed in \cite{paradis_unqomp_2021}.
Specifically it replaces in $\un{G}$ all CCX gates which are later uncomputed by RCCX gates.
While RCCX gates introduce an additional phase change, replacing pairs of CCX
gates ensures that this phase change is also reverted.

As RCCX gates can be implemented more efficiently than CCX gates (the latter
require more T gates), this can lead to a substantial efficiency improvement.
This is particularly appealing in our setting, were we encounter many CCX gates,
and most of them are uncomputed.

We note that \unqomp could only apply this optimization to gates it had itself
uncomputed, whereas \reqomp can also identify uncomputation that is already in
place in the original circuit, by leveraging value indices. 

The updated circuit graph $\un{G}$ is then converted back to a circuit, as described in \cref{sec:circ-graphs:circ-g}. For our example, this results in the circuit $\un{C}$ in
\cref{fig:reqomp-circu}.
Importantly, the resulting circuit uses the same physical ancilla qubit to hold
both $a$ and $c$, saving one qubit at the cost of an extra uncomputation and
recomputation of qubit~$a$. The same physical qubit is also used to hold $d$, at no extra recomputation cost, as $d$ does not depend directly on $a$.

\subsection{Fallback procedure: Reqomp-Lazy}
\label{sec:reqomp-lazy}
\begin{figure}
    \footnotesize
	\begin{algorithmic}[1]
		\continueLineNumber
		\Function{Reqomp-Lazy}{\null}
			\State $\graphu \gets \tfunc{DeepCopy}(\graph)$ \label{lin:unqomp-graphu}
			\State $U \gets \left\{q_{s.0} \in \un{G} \mid q \tfunc{ ancilla in } \un{G}.qbs \text{ and } s > 0\right\}$\label{lin:unqomp-u}
%			\State $U \gets \cup_{a \tfunc{ ancilla in } \un{G}.qbs}\{a_{s.0} \mid s \in [1... \tfunc{valIdx}(\graphPrim{last}{\graphu}{a})]\}$
			\For{$\tvar{v} \in \tfunc{reverse}(\tfunc{topologicalSort}(U))$} \label{lin:unqomp-evolve-start}
			\State $l \gets \graphPrim{last}{\un{G}}{v.\tfunc{qubit}}$
			%\State $p \gets \graphPrim{predecessor}{\graphval}{v}$ 
			\State \textbf{assert} $l.\tfunc{valIdx} == v.\tfunc{valIdx}$\label{lin:unqomp-assert-last-good}
			%\State \textbf{assert} $p.\tfunc{valIdx} == v.\tfunc{valIdx} - 1$\label{lin:unqomp-assert-pred-good}
			%\State $\tfunc{evolveVertexUntil}(\tvar{l}.\tfield{qbit}, p.\tfunc{valIdx})$\label{lin:unqomp-evolve}
			\State $\tfunc{evolveVertex}(v.\tfunc{qubit}, v.\tfunc{valIdx} - 1, \emptyset)$\label{lin:unqomp-evolve}
			\EndFor \label{lin:unqomp-evolve-end}
			%\For{$v_{s.i}$ final vertex in $G$}\label{lin:unqomp:non-anc2-beg}
			%\If{$v$ not an ancilla}
			%\State \tfunc{evolveVtxUntil}($v$, $s$)\label{lin:unqomp:non-anc2-end}
			%\EndIf
			%\EndFor %%%%%%%%%%this is not in the code -> i dropped it
			\State $\tfunc{assertFullyEvolved}()$\label{lin:unqomp-assert}
			\State $\tfunc{reuseAncillaRegisters}()$\label{lin:unqomp-reuse}
		\EndFunction
		\recordLineNumber
	\end{algorithmic}
    \caption{Fallback algorithm closely following
    Unqomp~\protect{\cite{paradis_unqomp_2021}}.}
    \label{alg:unqomp2}
\end{figure}
In general, uncomputation according to \cref{def:correct} is not always physically possible~\cite[\secref{6.2}]{paradis_unqomp_2021}.
Because we cannot always achieve uncomputation, \reqomp applies heuristics
to succeed as frequently as possible. However, we must accept that they may fail
in some cases. 
First, the ancillae within a partition may not be linearly dependent, or the ancillae partitions may have cyclic dependencies. In such cases, \reqomp falls back to the heuristic Reqomp-Lazy, which we will describe shortly. Second, assertions in \tfunc{evolveVertex} or \tfunc{assertFullyEvolved} may fail, both in \reqomp and Reqomp-Lazy. In such cases, \reqomp returns an error.
%
%For instance, the ancillae may not be linearly dependent, the call to \tfunc{assertFullyEvolved} may fail or an assertion in \tfunc{evolveVertex} may fail.
%
When this happens, it may indicate that no approach can achieve
uncomputation, hinting at a possible implementation mistake or misconception by
the programmer.
If uncomputation is possible, but no available approach can synthesize it
automatically, a programmer can always uncompute manually instead.

\para{Overview}
Reqomp-Lazy is inspired by \unqomp~\cite{paradis_unqomp_2021}, but leverages
the augmented circuit graphs and \tfunc{evolveVertex}.
In particular, \reqomp can uncompute and recompute controls for a vertex when
they are not directly available.
In contrast, \unqomp would have returned an error anytime this happens. We
provide an example in~\cref{fig:extra_gate_unqomp} (\cref{sec:eval}).

\cref{alg:unqomp2} shows the algorithm Reqomp-Lazy.
\cref{lin:unqomp-graphu} initializes $\graphu$ with a copy of $\graph$, to be
extended by adding vertices that perform uncomputation.
\cref{lin:unqomp-u} defines $U$ as the set of all vertices to be
uncomputed: it contains the first instance of each value index on each
ancilla qubit.
Then, \crefrange{lin:unqomp-evolve-start}{lin:unqomp-evolve-end} step through
$U$ in reverse topological order and revert all operations on ancillae one step at a time  by
calling \tfunc{evolveVertex} (\cref{lin:unqomp-evolve}).
%
%Note that \crefrange{lin:unqomp-assert-last-good}{lin:unqomp-assert-pred-good} ensure that the current last node on the qubit has the same value index as $v$ and that its predecessor has value index exactly $v.\tfunc{valIdx} - 1$, and therefore that the call to \tfunc{evolveVtxUntil} will result in exactly one call to \tfunc{evolveVertex}, uncomputing a single step on the given ancilla variable.
%
Then, analogously to \reqomp,
% \crefrange{lin:unqomp:non-anc2-beg}{lin:unqomp:non-anc2-end} evolves all non-ancillae (which may have been changed by calls to \tfunc{getAvailCtrl} in the calls to \tfunc{evolveVertex}  \cref{lin:unqomp-evolve}) and 
 \cref{lin:unqomp-assert} asserts all qubits
are fully evolved.
Finally, as specified in the original version of
\unqomp~\cite[\secref{5.4}]{paradis_unqomp_2021}, \cref{lin:unqomp-reuse}
allocates ancillae to the same physical qubits if their lifetimes do not
overlap. The resulting circuit graph is then finally converted to a circuit, using the procedure detailed in \cref{sec:reqomp-conv-to-circ}.

\para{Custom Control Strategy}
While Reqomp-Lazy reuses $\tfunc{evolveVertex}$, it uses a different implementation of \tfunc{getAvailCtrl}.
This new implementation aims at using controls that are as early (in terms of target edges) as possible, therefore keeping later controls available for later uncomputations.
Specifically, to find a control $c_s$ for a vertex $\un{v}$, it finds the earliest vertex in $\un{G}$ on qubit $c$ and value index $s$ that is available for $\un{v}$. Recall that in contrast, we used the latest such vertex when using \tfunc{evolveVertex} to apply an uncomputation strategy (see \cref{sec:heuristics:compl-graph}). If no such control vertex can be found, we recursively call \tfunc{evolveVertex} to evolve the last vertex on qubit $c$  until it has the state index $s$, just as we did in \cref{sec:heuristics:compl-graph}.

\colorlet{extended}{black!30}

\begin{figure}
    \begin{tikzpicture}
        % GATES
        \matrix[row sep=\tikzRowSepDefault,column sep=\tikzColumnSepDefault] (grid) {
            % q
            \node (q-init) [style=circuit-ctrl,extended] {}; &
            \node (q-H) [style=circuit-gate] {$H$}; &
            \node (q-H-ctrl) [style=circuit-ctrl,extended] {}; &
            \node [style=circuit-noop] {}; &
            \node (q-cx-1) [style=circuit-ctrl] {}; &
            \node [style=circuit-noop] {}; &
            \node (q-cx-2) [style=circuit-ctrl] {}; &
            \node (q-last) [style=circuit-noop] {}; \\

            % a
            \node (a-init) [style=circuit-noop] {}; &
            \node [style=circuit-noop] {}; &
            \node [style=circuit-noop] {}; &
            \node (a-init-ctrl) [style=circuit-ctrl,extended] {}; &
            \node (a-cx-1) [style=circuit-gate] {$X$}; &
            \node (a-cx-1-ctrl) [style=circuit-ctrl,extended] {}; &
            \node (a-cx-2) [style=circuit-gate] {$X$}; &
            \node (a-cx-2-ctrl) [style=circuit-ctrl,extended] {}; \\

            \node (dummy) [style=circuit-noop] {};
            \\

            % q-init-copy
            \node (q-init-copy-first) [style=circuit-gate] {$X$};
            \node [style=circuit-noop] {}; &
            \node [style=circuit-noop] {}; &
            \node [style=circuit-noop] {}; &
            \node [style=circuit-noop] {}; &
            \node [style=circuit-noop] {}; &
            \node [style=circuit-noop] {}; &
            \node [style=circuit-noop] {}; &
            \node (q-init-copy-last) [style=circuit-noop] {}; & \\

            % q-H-copy
            \node (q-H-copy-first) [style=circuit-noop] {}; &
            \node [style=circuit-noop] {}; &
            \node (q-H-copy) [style=circuit-gate] {$X$}; &
            \node [style=circuit-noop] {}; &
            \node [style=circuit-noop] {}; &
            \node [style=circuit-noop] {}; &
            \node [style=circuit-noop] {}; &
            \node (q-H-copy-last) [style=circuit-noop] {}; \\

            % a-init-copy
            \node (a-init-copy-first) [style=circuit-noop] {}; &
            \node [style=circuit-noop] {}; &
            \node [style=circuit-noop] {}; &
            \node (a-init-copy) [style=circuit-gate] {$X$}; &
            \node [style=circuit-noop] {}; &
            \node [style=circuit-noop] {}; &
            \node [style=circuit-noop] {}; &
            \node (a-init-copy-last) [style=circuit-noop] {}; \\

            % a-cx-1-copy
            \node (a-cx-1-copy-first) [style=circuit-noop] {}; &
            \node [style=circuit-noop] {}; &
            \node [style=circuit-noop] {}; &
            \node [style=circuit-noop] {}; &
            \node [style=circuit-noop] {}; &
            \node (a-cx-1-copy) [style=circuit-gate] {$X$}; &
            \node [style=circuit-noop] {}; &
            \node (a-cx-1-copy-last) [style=circuit-noop] {}; \\

            % a-cx-2-copy
            \node (a-cx-2-copy-first) [style=circuit-noop] {}; &
            \node [style=circuit-noop] {}; &
            \node [style=circuit-noop] {}; &
            \node [style=circuit-noop] {}; &
            \node [style=circuit-noop] {}; &
            \node [style=circuit-noop] {}; &
            \node [style=circuit-noop] {}; &
            \node (a-cx-2-copy) [style=circuit-gate] {$X$}; \\
        };

        % POSITIONS

        % q
        \coordinate[left=1cm of q-init.center] (q-left);
        \coordinate[right=0.5cm of q-last.center] (q-right);

        % a
        \coordinate[left=1cm of a-init.center] (a-left);
        \coordinate[right=0.5cm of a-cx-2-ctrl.center] (a-right);

        % q-init-copy
        \coordinate[left=1cm of q-init-copy-first.center] (q-init-copy-left);
        \coordinate[right=0.5cm of q-init-copy-last.center] (q-init-copy-right);

        % q-H-copy
        \coordinate[left=1cm of q-H-copy-first.center] (q-H-copy-left);
        \coordinate[right=0.5cm of q-H-copy-last.center] (q-H-copy-right);

        % a-init-copy
        \coordinate[left=1cm of a-init-copy-first.center] (a-init-copy-left);
        \coordinate[right=0.5cm of a-init-copy-last.center] (a-init-copy-right);

        % a-cx-1-copy
        \coordinate[left=1cm of a-cx-1-copy-first.center] (a-cx-1-copy-left);
        \coordinate[right=0.5cm of a-cx-1-copy-last.center] (a-cx-1-copy-right);

        % a-cx-1-copy
        \coordinate[left=1cm of a-cx-2-copy-first.center] (a-cx-2-copy-left);
        \coordinate[right=0.5cm of a-cx-2-copy.center] (a-cx-2-copy-right);

        % states
        \node[anchor=east] at (q-left) {$\ket{0}$};
        \node[anchor=east] at (a-left) {$\ket{0}$};
        \node[anchor=east] at (q-init-copy-left) {$\ket{0}$};
        \node[anchor=east] at (q-H-copy-left) {$\ket{0}$};
        \node[anchor=east] at (a-init-copy-left) {$\ket{0}$};
        \node[anchor=east] at (a-cx-1-copy-left) {$\ket{0}$};
        \node[anchor=east] at (a-cx-2-copy-left) {$\ket{0}$};

        % wire labels
        \node[anchor=base west,scale=0.9] at ($(q-left) - (0.1,-0.17)$) {$q$};
        \node[anchor=base west,scale=0.9] at ($(a-left) - (0.1,-0.17)$) {$a$};
        \node[anchor=base west,scale=0.9] at ($(q-init-copy-left) - (0.1,-0.17)$)
        {$\copyqb{q_{0.0}}$};
        \node[anchor=base west,scale=0.9] at ($(q-H-copy-left) - (0.1,-0.17)$) {$\copyqb{q_{1.0}}$};
        \node[anchor=base west,scale=0.9] at ($(a-init-copy-left) - (0.1,-0.17)$)
        {$\copyqb{a_{0.0}}$};
        \node[anchor=base west,scale=0.9] at ($(a-cx-1-copy-left) - (0.1,-0.17)$)
        {$\copyqb{a_{1.0}}$};
        \node[anchor=base west,scale=0.9] at ($(a-cx-2-copy-left) - (0.1,-0.17)$) {$\copyqb{a_{0.1}}$};

        % square
        \begin{scope}[on background layer]
            \node[rectangle, fill=black!4, minimum width=7cm,
            minimum height = 1.9cm,anchor=north west] (circuit) at ($(q-left.north west) +
            (-0.7,1.1)$) {};
            \node[anchor=north] at (circuit.north) {\footnotesize(a)
            Circuit};

            \node[rectangle, fill=black!4, minimum width=7cm,
            minimum height = 3.05cm,anchor=south west] (extended) at ($(a-cx-2-copy-left.south west) +
            (-0.7,-0.82)$) {};
            \node[anchor=south] at (extended.south) {\footnotesize(b)
            Extension};
        \end{scope}

        % wires
        \begin{scope}[on background layer]
            \draw[-] (q-left) -- (q-right);
            \draw[-] (a-left) -- (a-right);
            \draw[-] (q-init-copy-left) -- (q-init-copy-right);
            \draw[-] (q-H-copy-left) -- (q-H-copy-right);
            \draw[-] (a-init-copy-left) -- (a-init-copy-right);
            \draw[-] (a-cx-1-copy-left) -- (a-cx-1-copy-right);
            \draw[-] (a-cx-2-copy-left) -- (a-cx-2-copy-right);

            \draw[-,extended] (q-init) -- (q-init-copy-first);
            \draw[-,extended] (q-H-ctrl) -- (q-H-copy);
            \draw[-,extended] (a-init-ctrl) -- (a-init-copy);
            \draw[-] (q-cx-1) -- (a-cx-1);
            \draw[-,extended] (a-cx-1-ctrl) -- (a-cx-1-copy);
            \draw[-] (q-cx-2) -- (a-cx-2);
            \draw[-,extended] (a-cx-2-ctrl) -- (a-cx-2-copy);
        \end{scope}

        % final states
        \node[anchor=west] (ket-left) at (q-right) {$\ket{0}$};
        \node[anchor=west] at (a-right) {$\ket{0}$};
        \node[anchor=west] at (q-init-copy-right) {$\ket{0}$};
        \node[anchor=west] at (q-H-copy-right) {$\ket{0}$};
        \node[anchor=west] at (a-init-copy-right) {$\ket{0}$};
        \node[anchor=west] at (a-cx-1-copy-right) {$\ket{0}$};
        \node[anchor=west] at (a-cx-2-copy-right) {$\ket{0}$};

        \node[anchor=west] (ket-right) at ($(q-right) + (0.6,0)$) {$\ket{1}$};
        \node[anchor=west] at ($(a-right) + (0.6,0)$) {$\ket{0}$};
        \node[anchor=west] at ($(q-init-copy-right) + (0.6,0)$) {$\ket{0}$};
        \node[anchor=west] at ($(q-H-copy-right) + (0.6,0)$) {$\ket{1}$};
        \node[anchor=west] at ($(a-init-copy-right) + (0.6,0)$) {$\ket{0}$};
        \node[anchor=west] at ($(a-cx-1-copy-right) + (0.6,0)$) {$\ket{1}$};
        \node[anchor=west] at ($(a-cx-2-copy-right) + (0.6,0)$) {$\ket{0}$};

        \node[] at ($(ket-left.center) + (0,0.7)$) {$\tfrac{1}{\sqrt{2}}$};
        \node[] at ($(ket-right.center) + (0,0.7)$) {$\tfrac{1}{\sqrt{2}}$};
        \node[scale=0.8] at ($(ket-left.center)!0.5!(ket-right.center) + (0,0.73)$) {$+$};

        % equalities
        \draw[-,line width=0.2mm,color=my-full-red]
            ($(a-init-copy-right) + (1.2,-0.05)$) --
            ($(a-init-copy-right) + (1.3,-0.05)$) --
            ($(a-cx-2-copy-right) + (1.3,0)$) --
            ($(a-cx-2-copy-right) + (1.2,0)$);
        
        \draw[-,line width=0.2mm,color=my-full-brown]
            ($(a-right) + (1.2,0)$) --
            ($(a-right) + (1.3,0)$) --
            ($(a-init-copy-right) + (1.3,+0.05)$) --
            ($(a-init-copy-right) + (1.2,+0.05)$);

        \draw[-,line width=0.2mm,color=my-full-blue]
            ($(q-right) + (1.2,0)$) --
            ($(q-right) + (1.25,0)$) --
            ($(q-H-copy-right) + (1.25,+0.05)$) --
            ($(q-H-copy-right) + (1.2,+0.05)$);
    \end{tikzpicture}
    
	\begin{subfigure}[t]{0pt}
		\phantomcaption\label{fig:correctness-circuit}
	\end{subfigure}
	\begin{subfigure}[t]{0pt}
		\phantomcaption\label{fig:correctness-extended}
	\end{subfigure}

    \caption{Intuition on the correctness of Reqomp.}
    \label{fig:correctness}
\end{figure}

\section{Correctness}
\label{sec:correctness}

We prove in \cref{app:proof} that Reqomp synthesizes correct uncomputation. In
this section, we provide an intuition of this proof.

\para{Value Index Assertions}
The correctness of Reqomp relies on value indices.
At the end of the algorithm (\cref{lin:reqomp-assertFullyEvolved} when \tfunc{ApplyingStrategy} succeeded, \cref{lin:unqomp-assert} when \reqomp falls back to Reqomp-Lazy), we assert
that the last vertex on all ancilla qubits has value index 0, and that for any
non-ancilla qubit, the value indices of the last vertex are the same for the
original graph and the synthesized graph.
Intuitively, this ensures that ancillae are reset to $\ket{0}$, while other
qubits are preserved.

\pagebreak
Correctness hence relies on the precise formal interpretation of value indices. %which we discuss in the following. 
Intuitively, we claim that two vertices on the same qubit with the same value
index \emph{hold the same value}.

\para{Extended Circuits}
To formally define this notion, we introduce the notion of an extended circuit.
We conceptually \emph{extend} a given circuit to allow us to
compare the value of all vertices occurring in the circuit.

\cref{fig:correctness} exemplifies this by extending the example circuit in
\cref{fig:correctness-circuit}, which applies an $H$ gate and two controlled $X$
gates to qubits $q$ and $a$. Overall, the circuit in
\cref{fig:correctness-circuit} yields state
$$
\tfrac{1}{\sqrt{2}} \ket{0}_q\ket{0}_a + \tfrac{1}{\sqrt{2}} \ket{1}_q\ket{0}_a,
$$
which we write in a column-by-column format in \cref{fig:correctness-circuit}
(right).

\cref{fig:correctness-extended} shows our extension of
\cref{fig:correctness-circuit}, copying\footnote{Note that copying using a
controlled $X$ gate does not violate the no-cloning theorem.} the value of each
vertex from the corresponding circuit graph to a fresh qubit.
The name of these \emph{copy qubits} is the same as their corresponding vertex
but underlined, e.g., $\copyqb{q_{0.0}}$ holds the initial state of $q$,
corresponding to vertex $q_{0.0}$.

\para{Value Index}
Intuitively, copy qubits with the same value index and qubit hold the same
value. More precisely, if we write the state produced by the extended circuit as
a sum of computational basis states, in each summand (with a non-null
coefficient), copy qubits with the same value index and qubit hold the same
value.
For example, in every summand (i.e., column) of the
final state in \cref{fig:correctness-extended}, $\copyqb{a_{0.0}}$ and
$\copyqb{a_{0.1}}$ hold value $\ket{0}$ (see red bracket in
\cref{fig:correctness}).

Similarly, each qubit holds the same value as its last copy qubit. For example,
in every summand (i.e., column) of the final state in
\cref{fig:correctness-extended}, $q$ and $\copyqb{q_{1.0}}$ both hold either
value $\ket{0}$ or $\ket{1}$ (see blue bracket in
\cref{fig:correctness}).

In \cref{lem:nullprojcoeffs}~(\cref{app:proof}) we formally prove that these two facts hold for any well-valued circuit graph, as defined in \cref{def:well-valued}.

We further show in \cref{app:proof} that any circuit graph built with
\tfunc{evolveVertex}~(\cref{alg:evolve-vertex-simpl}) is well-valued.

\para{Final Values in the Extended Graph}
The assertion in \cref{lin:reqomp-assertFullyEvolved} (resp. \cref{lin:unqomp-assert}) ensures that in the circuit graph $\un{G}$ built by applying the uncomputation strategy (resp. the circuit graph $\un{G}$ built by Reqomp-Lazy), the last vertex on all ancilla qubits has
value index 0.
Hence, those qubits hold the same value as the initial value of that qubit, i.e.,
$\ket{0}$.
More precisely, consider a circuit graph $G$ with ancilla qubits $A$ and non
ancilla qubits $Q$, and denote $\un{G}$ the circuit graph after uncomputation.
We then have that any summand in the final state after applying the extended
version of $\un{G}$ is of the form $\ket{0...0}_A \otimes \ket{i}_Q \otimes
\ket{...}_{\un{V}}$, where we use $\un{V}$ to denote all the copy qubits in
$E(G)$.

The assertions in \cref{lin:reqomp-assertFullyEvolved} and \cref{lin:unqomp-assert} further check that the
value indices of non-ancilla qubits match their respective last vertices in
$\graph$.
As we show more formally in \cref{app:proof}, this means that if the effect of
the extended version $E(\graph)$ of $\graph$ on some initial state can be
written as
\begin{minipage}{\linewidth}
\footnotesize
\begin{align*}
    \ket{0 \cdots 0}_A \otimes \varphi &\xmapsto{\llbracket E(\graph) \rrbracket} \sum_{\substack{j \in \{0,1\}^{|A|}\\ k \in \{0,1\}^{|Q|}}} \gamma_{jk} \ket{j}_A \otimes \ket{k}_Q \otimes \ket{...}_{\copyqb{V}},
\end{align*}
\end{minipage}
then the effect of the extended version $E(\un{G})$ of $\un{G}$ on the same state can be written as:
\begin{minipage}{\linewidth}
\footnotesize
\begin{align*}
    \ket{0 \cdots 0}_A \otimes \varphi &\xmapsto{\llbracket E(\graphu) \rrbracket} \sum_{\substack{j \in \{0,1\}^{|A|}\\ k \in \{0,1\}^{|Q|}}} \gamma_{jk} \ket{0 \cdots 0}_A \otimes \ket{k}_Q \otimes \ket{...}_{\copyqb{V'}},
\end{align*}
\end{minipage}
where we denote $\un{V'}$ the set of copy qubits in $E(\un{G})$.

\para{Circuit Graph Semantics}
Importantly, the semantics of the unextended circuit follows straight-forwardly
from the semantics of the extended circuit. In \cref{fig:correctness}, simply
ignoring the rows from \cref{fig:correctness-extended} yields the correct final
state. If we similarly ignore the values of the copy qubits $\copyqb{V}$ and
$\copyqb{V'}$ in the two equations above, we recover the correct uncomputation
theorem, for circuits $C$ and $\un{C}$:
\defcorect*

\para{Multiple Graphs}
Note that here we assumed that both $G$ and $\un{G}$ have the same effect, as
they apply the same gates for the same value indices. Proving this formally
requires extra work, done in \cref{lem:valprojcoeffs}~(\cref{app:proof}).

\begin{table*}[!htb]
    \caption{%
        \reqomp results when targeting a specific ancilla qubit reduction compared to
        \unqomp (e.g., $-66.7$ indicates a reduction by $66.7\%$).
        Gate counts are reported as compared to \unqomp (e.g., $70.5$ indicates an increase by $70.5\%$).
        Columns \textbf{Max} and \textbf{Min} report the results for the most
        aggressive settings, respectively optimizing only for number of qubits
        and optimizing only for number of gates.
        Columns \textbf{-75\%}, \textbf{-50\%}, and \textbf{-25\%} report the
        gate counts when achieving the respective ancilla qubit reductions.
        Entries "x" indicate that a given ancilla qubit reduction was not achieved. 
        }
    \label{tab:trade-offsrange}

    %\vspace{-0.2cm}
    \centering
    \footnotesize
    \newcolumntype{U}{S[table-format=2.1,table-column-width=0.4cm, group-digits = false]}
    \begin{tabular}{l|rr|r|r|r|rr|}%
    &    \multicolumn{7}{c|}{\textbf{Ancilla Reduction}} \\
    &    \multicolumn{2}{c}{\textbf{Max}} & \multicolumn{1}{c}{\textbf{-75\%}} & \multicolumn{1}{c}{\textbf{-50\%}} & \multicolumn{1}{c}{\textbf{-25\%}}& \multicolumn{2}{c|}{\textbf{Min}}\\
	\textbf{Algorithm}  &  \multicolumn{1}{c}{{anc (as \%)}}
	&\multicolumn{1}{c}{{gates (as \%)}} & \multicolumn{1}{c}{gates (as \%)}  &
	\multicolumn{1}{c}{gates (as \%)}  & \multicolumn{1}{c}{gates (as \%)}
	&\multicolumn{1}{c}{anc (as \%)} & \multicolumn{1}{c|}{gates (as \%)}  \\
	\hline
    \textbf{Small} & & & & & & & \\
        % to produce the file: plots.py run gets_vals_all_examples() then run rewrites_csv_output.py then copy paste tradeoffsclean.csv here. Examples names can be changed in plots.py
        Adder&-66.7&70.5&x&39.2&15.7&-8.3&0.0\\
Deutsch-Jozsa&-50.0&40.9&x&40.9&20.4&0.0&0.0\\
Grover&-33.3&12.9&x&x&12.9&0.0&0.0\\
IntegerComparator&-63.6&47.1&x&36.1&11.6&0.0&-5.2\\
MCRY&-63.6&80.0&x&53.3&26.7&0.0&0.0\\
MCX&-60.0&55.2&x&46.0&27.6&0.0&0.0\\
Multiplier&0.0&-5.1&x&x&x&0.0&-5.1\\
PiecewiseLinearR&-50.0&7.5&x&7.5&2.5&0.0&-3.3\\
PolynomialPauliR&-33.3&9.5&x&x&9.5&0.0&0.0\\
WeightedAdder&0.0&-9.7&x&x&x&0.0&-9.7\\

        \textbf{Big} & & & & & & & \\
        Adder&-93.0&314.9&64.7&42.9&21.0&-1.0&0.0\\
Deutsch-Jozsa&-92.9&327.1&69.0&45.7&23.3&0.0&0.0\\
Grover&-50.0&37.8&x&37.8&16.2&0.0&0.0\\
IntegerComparator&-92.9&310.8&60.2&37.6&14.7&0.0&-6.7\\
MCRY&-96.0&515.4&75.3&50.2&25.1&0.0&0.0\\
MCX&-96.0&509.9&74.9&49.8&25.1&0.0&0.0\\
Multiplier&0.0&-5.4&x&x&x&0.0&-5.4\\
PiecewiseLinearR&-85.0&47.1&17.6&10.2&2.8&0.0&-3.8\\
PolynomialPauliR&-50.0&14.4&x&14.4&1.5&0.0&0.0\\
WeightedAdder&0.0&-8.0&x&x&x&0.0&-8.0\\

        \end{tabular}
        \vspace{-0.2cm}
\end{table*}

\section{Evaluation}\label{sec:eval}
We have evaluated \reqomp on an existing benchmark to answer the following
research questions:
\begin{enumerate}[label={\bf Q\arabic*},leftmargin=2em]   
    \item Circuit Efficiency: Can \reqomp create efficient circuits in terms of
    number of qubits and gates, while allowing to trade one for the other?
    \label{q:efficiency}
    \item Usability: Is \reqomp fast and directly applicable to a wide range of
    circuits? \label{q:usability}
\end{enumerate}

\para{Implementation}
We implemented \reqomp as a language extension of Qiskit, using Qiskit's
built-in \texttt{AncillaRegister} type to mark ancilla variables in the circuit. As
Qiskit, our extension is implemented in Python.

\begin{figure*}[!htb]
    \centering
    \begin{subfigure}{0.49\linewidth}
        \includegraphics[width = \linewidth]{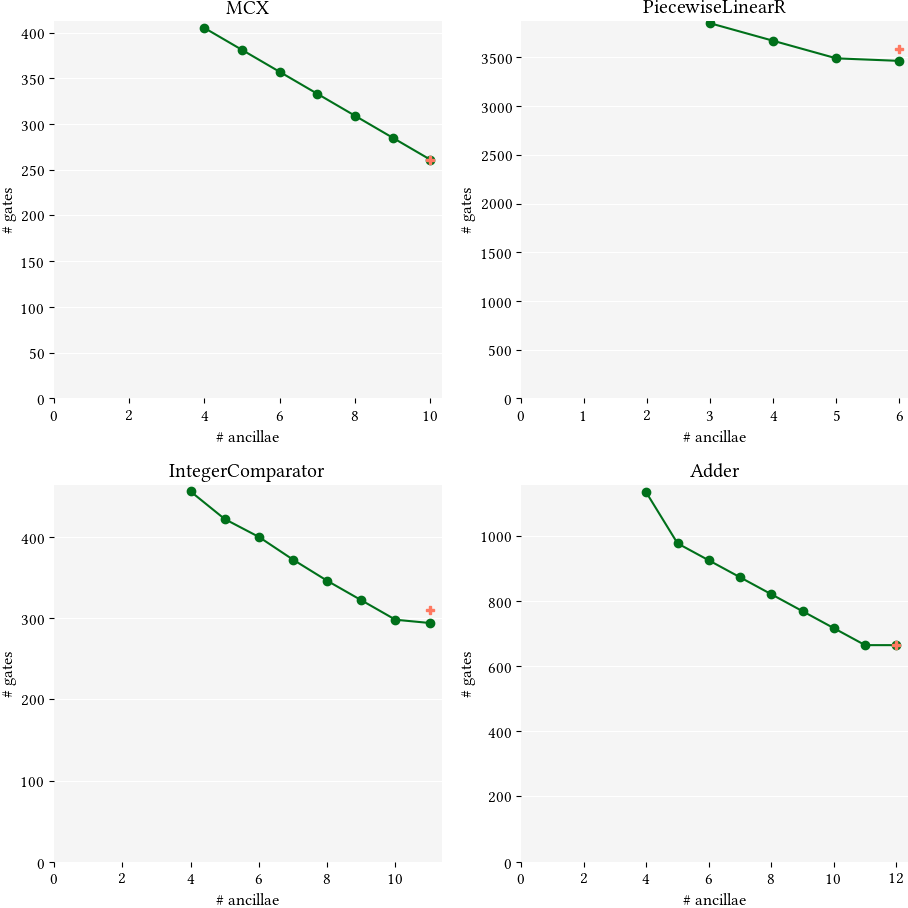}
        \caption{Gate counts for selected small circuits.}
        \label{fig:ancillae-gates-small}
    \end{subfigure}~
    \begin{subfigure}{0.49\linewidth}
        \includegraphics[width = \linewidth]{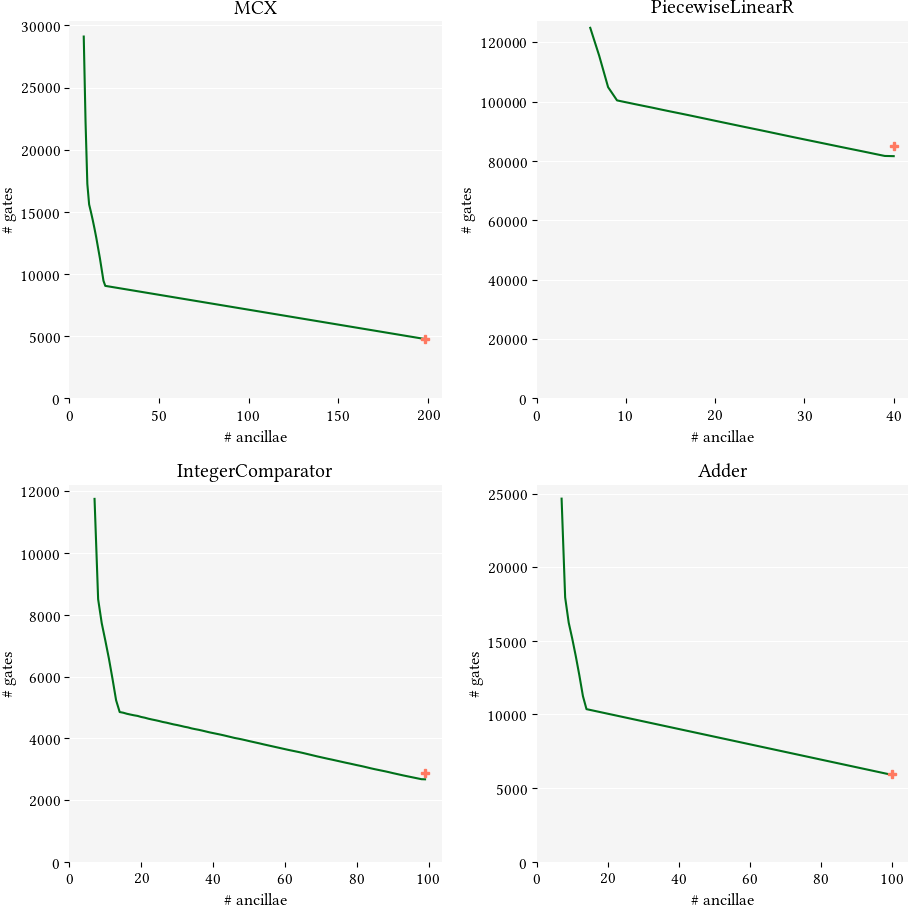}
        \caption{Gate counts for selected big circuits.}
        \label{fig:ancillae-gates-big}
    \end{subfigure}%
    \\%
    \begin{subfigure}{0.49\linewidth}
        \includegraphics[width = \linewidth]{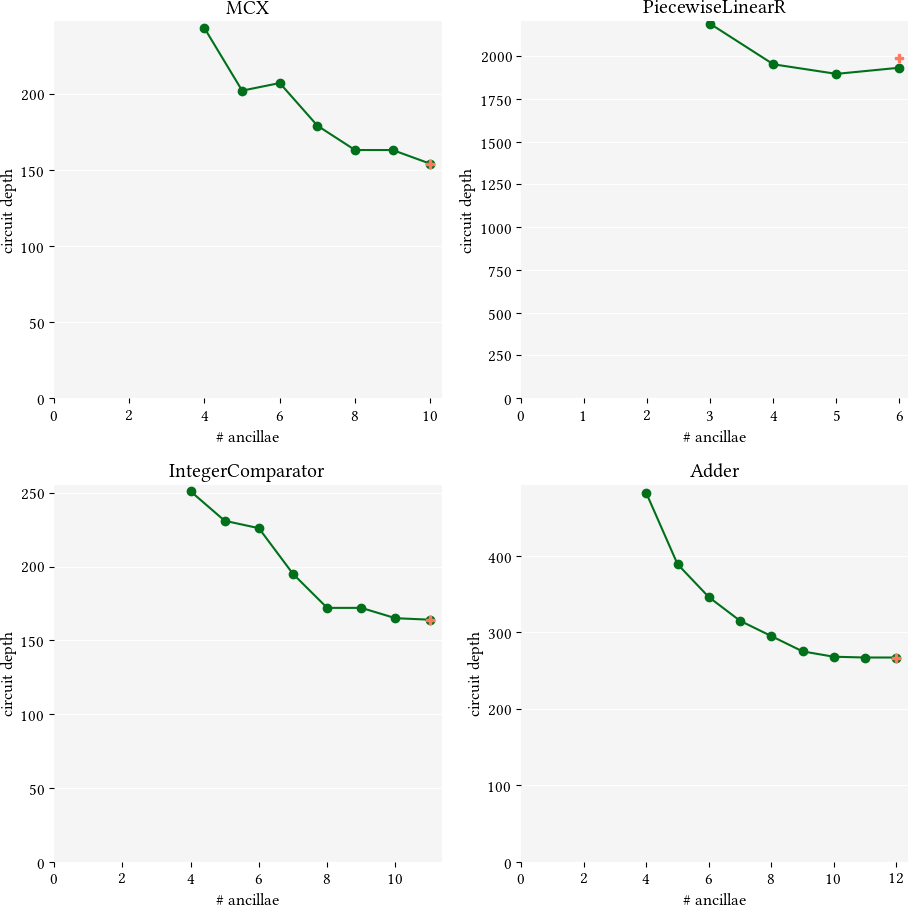}
        \caption{Circuit depth for selected small circuits.}
        \label{fig:ancillae-depth-small}
    \end{subfigure}~
    \begin{subfigure}{0.49\linewidth}
        \includegraphics[width = \linewidth]{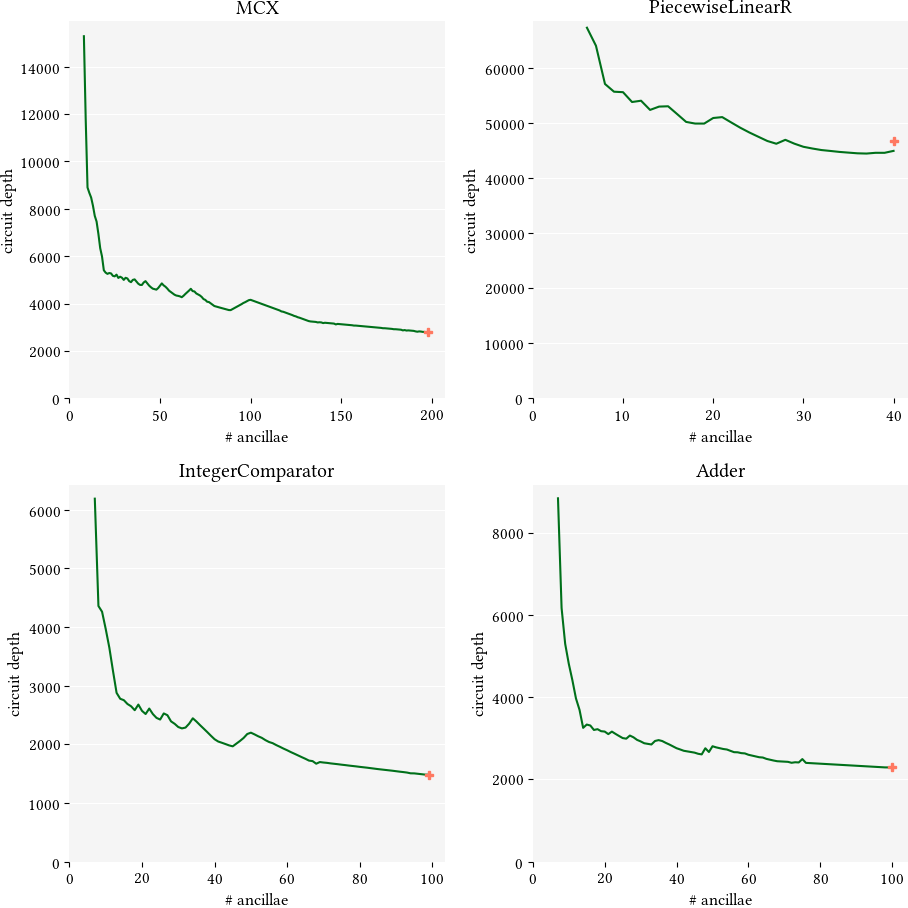}
        \caption{Circuit depths for selected big circuits.}
        \label{fig:ancillae-depth-big}
    \end{subfigure}
    \caption{Gate counts
    (\subref{fig:ancillae-gates-small}--\subref{fig:ancillae-gates-big}) and
    circuit depths
    (\subref{fig:ancillae-depth-small}--\subref{fig:ancillae-depth-big}) for
    given numbers of ancillae, using
    \reqomp~(\raisebox{-0.1em}{\protect\includegraphics[height=0.6em]{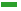}})
    and
    \unqomp~(\raisebox{-0.1em}{\protect\includegraphics[height=0.8em]{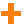}}).
    %
    % for small (\subref{fig:ancillae-gates-small}\&\subref{fig:ancillae-depth-small}) and big (\subref{fig:ancillae-gates-big}\&\subref{fig:ancillae-depth-big}) circuits
    %
    }
    \label{fig:trade-off}
\end{figure*}
\pagebreak
\subsection{Benchmarks and Baseline}
To evaluate \reqomp, we used the benchmark from
\unqomp~\cite{paradis_unqomp_2021}.
The first column in \cref{tab:trade-offsrange} summarizes the circuits in our
benchmark, separated into "small" and "big" circuits.
While the "small" circuits were taken directly from \unqomp, we have generated
the "big" circuits by re-parametrizing the original circuits to yield bigger
circuits. This allows us to demonstrate the \reqomp also performs well on larger
circuits.

For completeness, we provide the exact parameters for each circuit in
\cref{app:eval}, including the resulting circuit sizes.

\para{Circuits}
To provide an intuition on our benchmark, we explain selected circuits (see
\cite[\secref{7.1}]{paradis_unqomp_2021} for details).

IntegerComparator takes a constant parameter $n$ and multiple input qubits
encoding a value $v$, and flips its output qubit if and only if $v \geq n$.
MCX flips its output qubit if and only if all its input qubits are one.
MCRY applies a rotation to its output qubit if and only if all its control
qubits are one.
PiecewiseLinearR applies a rotation $f(x)$ to its output qubit, where $x$ is the
value on its input qubits and $f$ is piecewise linear.
PolynomialPauliR works analogously, but for polynomial $f$.
% %
WeightedAdder takes as parameters a list of weights $\lambda_0, ... \lambda_n$
and outputs $\sum \lambda_i q_i$ where the $q_i$ are the input values.

% Each of these examples is a parametric program, that produces a quantum circuit
% for a given input size (given as a number of qubits). 
\para{Selecting a Baseline}
We provide a thorough overview of related work in \cref{sec:rel-work}. 
Of the many works discussed there, only four can take circuits as input: 
\squarerelwork~\cite{ding_square_2020}, Quipper~\cite{green_quipper_2013}, 
ReQWire~\cite{rand_reqwire_2019} and Unqomp~\cite{paradis_unqomp_2021}. 
Of these, ReQWire can only verify uncomputation in circuits and not synthesize it\footnote{It can synthesize circuits with uncomputation from a boolean formula but we focus here on its possibilities when working directly on circuits.}, and we show in \cref{sec:square} that due to various shortcomings \squarerelwork is not a viable option for uncomputation. This leaves only Quipper and \unqomp. As \cite{paradis_unqomp_2021} showed that \unqomp generally outperforms Quipper, we choose \unqomp as our baseline.

Other works take as input boolean formulas (to be compiled to circuits)~\cite{parent_reversible_2015,parent_revs_2017,bhattacharjee_reversible_2019}, focus on building uncomputation strategies without explaining how to apply them~\cite{meuli_reversible_2019,bennett_timespace_1989}, or do not compile to circuits~\cite{qsharp}. 

\subsection{\ref*{q:efficiency}: Circuit Efficiency}
We now discuss the efficiency of circuits produced by \reqomp in terms of qubits
and gates and compare them to circuits produced by
\unqomp\cite{paradis_unqomp_2021}.

\para{Approach}
For each circuit, we ran \reqomp targeting all possible number of ancilla qubits
\texttt{nAncillaQubits}. We then recorded, for all calls
that terminated without error, the number of qubits and gates of the resulting
circuit (with uncomputation).

We note that \reqomp had to fall back to Reqomp-Lazy for
circuits Multiplier and WeightedAdder, as the ancilla dependencies of these
circuits are not linear. While Reqomp-Lazy succeeds on these circuits and even
outperforms \unqomp, it cannot offer multiple space-time trade-offs.

\para{Results}
\cref{tab:trade-offsrange} summarizes our results.
Note that gate counts are expressed as a percentage of \unqomp gate counts.
For all examples, using the maximum number of ancilla qubits (column \textbf{Min} as
this is the minimal reduction) yields better results than \unqomp for
$10$ circuits, and equivalent results for the remaining $10$ circuits.
For example, \reqomp saves 5.2\% of gates on circuit IntegerComparator, without
requiring additional qubits. This is because \reqomp can identify uncomputation
\emph{already present in the original circuit}, allowing it to avoid unnecessary
operations when uncomputing or recomputing an ancilla or even a control.
Analogous effects occur for PiecewiseLinearR, WeightedAdder, and
Multiplier, where the last two are handled by Reqomp-Lazy.

More importantly, \cref{tab:trade-offsrange} demonstrates that \reqomp can
significantly reduce the number of ancilla qubits compared to \unqomp: by up to
$96\%$ for two examples, and by at least $25\%$ for $16$ out of $20$ circuits.
Importantly, this reduction comes at only a moderate cost in gate count,
below $28\%$ for qubit reductions of $25\%$.
As most quantum computers are more limited in terms of qubits than gates, these
trade-offs are highly favorable.
Further, for some examples the reduction in qubits comes at almost no cost in
gates: for PiecewiseLinearR, reducing by 75\% the number of ancilla qubits only increases the number of gates by~$17.6\%$.

\para{Trade-Offs}
To further demonstrate the gate count cost incurred by these reductions,
\crefrange{fig:ancillae-gates-small}{fig:ancillae-gates-big} show a more fine-grained visualization of the trade-offs between ancilla qubits and gate count.

Overall, we immediately observe that on all circuits, reducing the number of
available ancilla qubits can only increase (and never decrease) the gate count of the
resulting circuit. However, the rate of this increase varies among the different
circuits, as discussed next.

For some benchmarks such as PiecewiseLinearR
(\crefrange{fig:ancillae-gates-small}{fig:ancillae-gates-big}) and
PolynomialPauliR (\cref{tab:trade-offsrange}), \reqomp can drastically reduce
the number of ancillae at almost no cost in terms of gates.

For other benchmarks such as MCX
(\crefrange{fig:ancillae-gates-small}{fig:ancillae-gates-big}) and MCRY
(\cref{tab:trade-offsrange}), \reqomp can still reduce the number of ancillae
substantially, but at a significant cost in terms of gates. In such cases, the
appropriate ancilla reduction depends on the available hardware---a programmer
with access to \reqomp can then systematically select the right trade-off.

Other circuits fall somewhere between these two categories
(\crefrange{fig:ancillae-gates-small}{fig:ancillae-gates-big} and
\cref{tab:trade-offsrange}): \reqomp can reduce the number of ancilla qubits, at a
non-negligible cost in terms of gates.

\para{Very Small Number of Qubits}
\cref{fig:trade-off} further demonstrates that enforcing a very small number of
ancillae typically increases the number of applied gates significantly.
For instance, MCX with 200 controls can be implemented with only 8 ancilla
qubits, but this requires a staggering \num{21831} gates, compared to only
\num{3579} when 200 ancillae are used.
%
% Interestingly, this explosion is significantly weaker for PiecewiseLinearR.

Overall, we conclude that enforcing very small number of ancilla qubits is
typically not a good approach.

\para{Depth}
For completeness, \crefrange{fig:ancillae-depth-small}{fig:ancillae-depth-big}
shows the trade-off between ancilla qubits reduction and circuit depth.
As we do not optimize for circuit depth, reducing the number of ancillae
sometimes yields \emph{shorter} circuits. Still, overall, circuit depth behaves
analogously to gate count, generally increasing for reduced ancilla qubits counts, at
different rates depending on the circuit.

Interestingly, in some cases, we can reduce the ancilla count at almost no cost
in circuit depth, even though there is a cost in gate count. For example,
reducing ancillae from $99$ to $25$ on Adder only increases depth by $31\%$, even
though it increases the gate count by $64\%$.

\subsection{\ref*{q:usability}: \reqomp Usability}
\label{sec:eval-expressivity}
We also investigated the usability of \reqomp, showing that it is both fast and
directly applicable to many quantum circuits.

\para{\reqomp Runtime}
Our evaluation indicated that \reqomp is fast: it synthesized uncomputation for
all circuits in \cref{tab:trade-offsrange} within five seconds.

Furthermore, running \reqomp typically takes as much time as decomposing the
resulting circuit to basic gates using Qiskit's built-in \texttt{decompose()}
function.
We hence believe that \reqomp can be integrated into the programmer's workflow
without incurring a significant slowdown.

\para{Applicability}
Recall that even for a circuit where uncomputation is possible in principle,
\reqomp may raise an error.
We therefore investigated how frequently \reqomp succeeds in practice, comparing
it to other tools:

\begin{center}
\begin{tabular}{@{ }l|@{ }c@{ }c@{ }c@{ }}
    & Qfree only & \unqomp & \reqomp \\ \toprule

    \vrule width 0pt height 2.2ex 
    \% examples covered & $\leq$ 50\% & 60\% & 100\% \\ % 2020
\end{tabular}
\end{center}

We find that Reqomp (with the fallback strategy Reqomp-lazy) finds a circuit
with uncomputation for all input circuits. In contrast, \unqomp can only cover
60\% of those circuits directly.
We will explain shortly how we tweaked \unqomp to also cover the remaining 40\%.
Furthermore, only 50\% of the circuits in our benchmark are purely classical,
hence any tool that exclusively supports qfree gates can at most be used on 50\%
of the examples.

\begin{figure}
    \centering
    % \hspace{3cm}
    % \Qcircuit @C=0.4em @R=0.1em @!R {
	% 				& \lstick{{p}} & \gate{X}     & \ctrl{1} & \gate{X} & \ctrl{2} & \gate{X}  & \ctrl{1} & \gate{X} & \qw\\
	% 				& \lstick{{a}} & \qw & \targ    & \qw & \ctrl{1} & \qw & \targ    & \qw & \qw\\
    %                 & \lstick{{r}} & \qw & \qw   & \qw & \targ &\qw & \qw & \qw   & \qw\\
 	% 			}
    \begin{tikzpicture}
        \input{figures/eval_unqomp_cannot.tex}

        \node (top-line) at ($(p-X-p-0)!-1.2!(a-0)$) {};
        \node (bot-line) at ($(p-X-p-0) + (0, 0.1cm)$) {};

        \node (p00-x) at ($(p-X-p-0)!-1.2!(p-X-a-0)$) {};
        \node[color=my-full-blue] (p00) at (p00-x |- top-line) {\footnotesize$p_0$};
        \path[->, color=my-full-blue] (p00)  edge [] (p00-x |- bot-line);

        \node (p10-x) at ($(p-X-p-0)!0.9!(p-X-a-0)$) {};
        \node[color=my-full-orange] (p10) at (p10-x |- top-line) {\footnotesize$p_1$};
        \path[->, color=my-full-orange] (p10)  edge [] (p10-x |- bot-line);

        \node (p01-x) at ($(p-X-p-1)!0.3!(p-X-p-2)$) {};
        \node[color=my-full-blue] (p01) at (p01-x |- top-line) {\footnotesize$p_0$};
        \path[->, color=my-full-blue] (p01)  edge [] (p01-x |- bot-line);

        \node (p11-x) at ($(p-X-p-2)!0.4!(p-X-p-3)$) {};
        \node[color=my-full-orange] (p11) at (p11-x |- top-line) {\footnotesize$p_1$};
        \path[->, color=my-full-orange] (p11)  edge [] (p11-x |- bot-line);

        \node(last-vert) at ($(p-X-p-3.west)!2!(p-X-p-3.east)$) {};
        \node (p02-x) at ($(p-X-p-3)!0.7!(last-vert)$) {};
        \node[color=my-full-blue] (p02) at (p02-x |- top-line) {\footnotesize$p_0$};
        \path[->, color=my-full-blue] (p02)  edge [] (p02-x |- bot-line);

        \begin{scope}[on background layer]
            \draw [split=0.2cm, color= my-full-blue, opacity=0.2]
                (p-X-p-0.west) -- % top left
                (p-label |- p-X-p-0.west) --
                %(p-label |- p-X-p-0.west) --
                (p-X-p-0.west) -- % top left
                (p-label |- p-X-p-0.west) --
                cycle;
            \draw [split=0.2cm, color= my-full-orange, opacity=0.2]
                (p-X-p-0.east) -- % top left
                (p-X-p-1.west) --
                %(p-label |- p-X-p-0.west) --
                (p-X-p-0.east) -- % top left
                (p-X-p-1.west) --
                cycle;
            \draw [split=0.2cm, color= my-full-blue, opacity=0.2]
                (p-X-p-1.east) -- % top left
                (p-X-p-2.west) --
                %(p-label |- p-X-p-0.west) --
                (p-X-p-1.east) -- % top left
                (p-X-p-2.west) --
                cycle;
            \draw [split=0.2cm, color= my-full-orange, opacity=0.2]
                (p-X-p-2.east) -- % top left
                (p-X-p-3.west) --
                %(p-label |- p-X-p-0.west) --
                (p-X-p-2.east) -- % top left
                (p-X-p-3.west) --
                cycle;

            \draw [split=0.2cm, color= my-full-blue, opacity=0.2]
                (last-vert) --
                (p-X-p-3.east) -- % top left
                (last-vert) --
                %(p-label |- p-X-p-0.west) --
                (p-X-p-3.east) -- % top left
                cycle;
        \end{scope}
        
    % COMPUTE/UNCOMPUTE 0
                    %\node (tt) at (i-X-i-0) {tt};
                    \node (TopBoxComp) at ($(p-X-p-0)!-0.55!(a-0)$){\phantom{t}};
                    \node (BotBoxComp) at ($(a-X-a-0)!-0.55!(p-X-a-0)$){\phantom{b}};
                    \node (LeftBoxComp) at ($(p-X-p-0)!-0.7!(p-X-a-0)$){\phantom{l}};
                    \node (RightBoxComp) at ($(p-X-p-1)!-0.7!(p-X-a-0)$){\phantom{r}};

                    \path [style=square-box,  color = my-full-red]
                    (TopBoxComp -| LeftBoxComp) -- % top left
                    (TopBoxComp -| RightBoxComp) -- % top left
                    (BotBoxComp -| RightBoxComp) -- % top left
                    (BotBoxComp -| LeftBoxComp) -- % top left
                    cycle;
                    
                    \node (LeftBoxComp02) at ($(p-X-p-2)!-0.7!(p-X-a-1)$){\phantom{l}};
                    \node (RightBoxComp02) at ($(p-X-p-3)!-0.7!(p-X-a-1)$){\phantom{r}};

                    \path [style=square-box, color=my-full-green]
                    (TopBoxComp -| LeftBoxComp02) -- % top left
                    (TopBoxComp -| RightBoxComp02) -- % top left
                    (BotBoxComp -| RightBoxComp02) -- % top left
                    (BotBoxComp -| LeftBoxComp02) -- % top left
                    cycle;
    \end{tikzpicture}
	\caption{Gate requiring definition in Unqomp.}
	\label{fig:extra_gate_unqomp}
\end{figure}

\para{\unqomp Limitations}
\unqomp can only handle 60\% of the circuits in our evaluation directly, because
it cannot accurately handle uncomputation that already occurs in the input circuit.
\cref{fig:extra_gate_unqomp} illustrates this on a circuit applying a
$\overline{C}X$ gate (see red box on the left), where the bar over $C$ indicates
that the control is inverted.
To invert the controls, the circuit applies an $X$ gate to invert the control,
and another $X$ gate to restore the value of the control.
To uncompute $a$ it is hence necessary to track that after two $X$ gates, $p$ is
back to its original value shown as $p_0$ in \cref{fig:extra_gate_unqomp}, and
therefore applying a third $X$ gate will bring its value to $p_1$ again,
allowing to uncompute $a$.
Value indices allow \reqomp to precisely track those value changes, and insert
the uncomputation gates (in the green box on the right).
In contrast, \unqomp fundamentally cannot allow for recomputation, as its correctness relies on each operation being computed and uncomputed exactly once. It further does not recognize uncomputation or recomputation already present in the original circuit. Therefore, in \cref{fig:extra_gate_unqomp}, \unqomp cannot recognize that the second $X$ gate recovers the original value of $p$.
Even if it did, it could not recompute $p_1$ to uncompute $a$. 
%Further, in general \unqomp does not modify the value of any non ancilla qubit through uncomputation, as it could not recover the values later (since it would require recomputation).
%

In our evaluation (\cref{tab:trade-offsrange}), we bypassed this type of issue
by defining the red block on the left as a custom gate controlled by $p$.
\unqomp then never decomposes this new gate, assumes it keeps $p$ constant, and
places it to uncompute $a$.
Unfortunately, this approach makes \unqomp harder to use, and in some cases
makes the resulting circuit less efficient.

\section{Related Work}
\label{sec:rel-work}
We now discuss works related to \reqomp.

\subsection{\squarerelwork}
\label{sec:square}
\input{figures/loc-dep-uncomp.tex}
Even though it cannot synthesize uncomputation code,
\squarerelwork~\cite{ding_square_2020} looks very closely related to \reqomp at
first sight.
Specifically, it presents "a compiler that automatically [places uncomputation]
in order to manage the trade-offs in qubit savings and gate
costs"~\cite[\secref{1}]{ding_square_2020}.
Unfortunately, \squarerelwork suffers from various shortcomings that prevent a
meaningful comparison to \reqomp.

\para{Square Problem Statement}
\squarerelwork takes as input a program defining a qfree circuit (non qfree gates are not supported).
In this program, each function consists of the three blocks \emph{Compute}
(indicating forward computation), \emph{Store} (indicating computation of
outputs), and \emph{Uncompute} (indicating uncomputation).
\squarerelwork then compiles this program to a circuit by arranging these
blocks, possibly repeating blocks when recomputation is helpful.

\squarerelwork defines three different strategies for interleaving the
blocks. Lazy (uncompute as late as possible), Eager (uncompute as early as
possible),  and finally \squarerelwork itself, using a custom heuristic. For the example $CCCH$ in \cref{fig:pb-stmt-big-circ}, Lazy would correspond to the 3-qubit strategy shown in \cref{fig:pb-stmt-big-circ-lazy} and Eager to the 2-qubit strategy shown in \cref{fig:pb-stmt-big-circ-eager}. We now present the main shortcomings of \squarerelwork. 

\para{Constant Compute/Uncompute Blocks}
As mentioned in \cref{sec:intro}, the gates needed to uncompute an ancilla
variable may depend on where this uncomputation occurs in the circuit.
It is hence impossible to define fixed \emph{Compute} and \emph{Uncompute}
blocks to be applied anywhere.

For instance, consider the circuit in \cref{fig:loc-dep-lazy}.
It uses three ancilla variables $a$, $b$, and $c$ to compute the output variable $r$ from the input $i$.
\cref{fig:loc-dep-lazy} highlights the Compute and Uncompute blocks \squarerelwork would consider, namely blocks $a$, $b$ and, $c$ for computation and blocks $a^\dagger$, $b^\dagger$, and $c^\dagger$ for uncomputation.
Note how the value of qubit $i$ is changed by block $b$, and restored later by block $b^\dagger$, ensuring that qubit $i$ has the same value for the $CX$ gate in block $a^\dagger$ as it had in block $a$.
Now, if we want to save one ancilla qubit by uncomputing ancilla variable $a$ early, we get the circuit shown in \cref{fig:loc-dep-eager}.
Here, when uncomputing $a$ for the first time, the value of $i$ has been changed in block $b$ and is not yet restored.
To correctly uncompute $a$ in the block $a^\dagger_2$ (different from the block $a^\dagger$), it is hence necessary to restore $i$ using an $X$ gate before using it as a control to uncompute $a$.
Similarly, block $b_2^\dagger$ must change the value of $i$ again.

Not accounting for the above, \squarerelwork assumes that no matter its placement, uncomputation code can be kept unchanged.
In particular, its eager strategy would use the Compute and Uncompute blocks from \cref{fig:loc-dep-lazy}, yielding \cref{fig:loc-dep-square}.
This is clearly incorrect as this circuit has different semantics than the one in \cref{fig:loc-dep-lazy}.
For example, for input $\ket{0}_i\ket{0}_t$, \cref{fig:loc-dep-lazy} produces state $\ket{0}_i\ket{0}_t$ while \cref{fig:loc-dep-square} produces state $\ket{0}_i\ket{1}_t$ (assuming ancillae are in state $\ket{0}$).

We note that \squarerelwork does not exclude such patterns---in fact its
\texttt{little-belle}
benchmark contains an analogous pattern.~\footnote{Benchmark
\texttt{little-belle} is available at
\url{https://github.com/epiqc/Benchmarks/blob/master/bench/square-cirq/synthetic/little_belle.py}.
We note that different uncomputation strategies do not yield different results
on it, as it does not contain gates modifying the output and hence is
semantically equivalent to the identity.}

\para{Incomplete Uncomputation}
Besides only supporting fixed uncomputation code, \squarerelwork may also skip uncomputation of some ancilla variables.
For some examples evaluated in \cite{ding_square_2020}, the implementation of the lazy strategy does not insert any uncomputation code at all, leaving all ancilla variables dirty, while the eager strategy uncomputes all of them.
Specifically, we believe that the reported differences between strategies in the
\squarerelwork publication (\cite[\tabref{III}]{ding_square_2020}) on the benchmarks\footnote{Available at
\url{https://github.com/epiqc/Benchmarks/tree/master/bench/square-cirq/application}.}
RD53, 6SYM, 2OF5, and ADDER4 are only due to leaving some ancillae dirty---as
these benchmarks do not contain nested uncomputation, the order of uncomputation
should not make a difference.

\para{Additional Parameters}
Finally, the
implementation 
of
\squarerelwork is inconsistent with the system described in
\cite{ding_square_2020}. Specifically, using the interface to specify
\emph{Compute} blocks requires providing $7$ parameters, and some benchmarks
evaluated in~\cite{ding_square_2020} also contain \emph{Unrecompute} and
\emph{Recompute} blocks not mentioned in the
publication~\cite{ding_square_2020}. Even though the authors provided us with brief explanations of these parameters on request, we could not confidently derive correct parameters for new benchmarks.

\subsection{Purely Classical Circuits}
Most works synthesizing uncomputation cannot handle non-qfree
gates~\cite{rand_reqwire_2019, green_quipper_2013, parent_reversible_2015, parent_revs_2017, majumdar_verified_2017}.~\footnote{\cite{rand_reqwire_2019} can verify uncomputation for non qfree circuits, but can synthesize it only for qfree ones.} 
It has already been established~\cite{paradis_unqomp_2021} that using such works
on quantum circuits by separating out the qfree subparts typically yields
inefficient circuits, and is sometimes even impossible.
% %

In the following, we discuss works which only support qfree gates, and define a custom strategy allowing to trade qubits for gates. 
We have already discussed \squarerelwork in \cref{sec:square}.

\para{Boolean Functions}
\textsc{Revs}~\cite{parent_reversible_2015, parent_revs_2017} translates irreversible
classical functions to reversible circuits. It focuses on optimization
possibilities during the translation from boolean functions to reversible
circuits, but also offers an uncomputation strategy, however without the option
of trading qubits for gates.

Similarly, \cite{bhattacharjee_reversible_2019} also translates boolean
specifications to reversible circuit. While it introduces another uncomputation
heuristic, it also cannot trade qubits for gates.

We expect that both of those strategies could be incorporated into \reqomp,
possibly yielding more efficient circuits.

\para{Pebble Games}
Multiple works present uncomputation strategies for classical reversible
computation, which can be reduced to solving \emph{pebble
games}~\cite{bennett_timespace_1989}.
Importantly, while pebble games operate on dependency graphs on values, \reqomp
operates on quantum circuits.
In particular, pebble games assume all values can be uncomputed, which is
incorrect for non-qfree gates. Further, a direct translation of circuits to such
graphs would ignore repeated values, leading to issues analogous to
\cref{fig:extra_gate_unqomp}. In contrast, conflating repeated values can lead
to cyclic dependencies, which are not supported by pebble games.

Knill~\cite{knill_analysis_1995} provides an optimal yet efficient solution for
linear dependencies. As most circuits we encounter in practice exhibit linear
dependencies, \reqomp uses the same uncomputation strategy.
Meuli et al.~\cite{meuli_reversible_2019} suggest using a SAT-solver to handle
arbitrary dependencies, which may be a possible extension of \reqomp.

\subsection{Non-Qfree Circuits}
We now discuss works offering uncomputation for non-qfree circuits. 

\para{Language Level}
Quantum languages like Quipper~\cite{green_quipper_2013} and Q\#~\cite{qsharp} offer
convenience functions to automatically insert uncomputation. However, these
functions are often tedious to use, and may insert incorrect uncomputation (see
\cite[\secref{8}]{paradis_unqomp_2021} for details).

Silq~\cite{bichsel_silq_2020} uses a type system to detect which variables can be
safely uncomputed, but does not synthesize this uncomputation.
Overall, none of those works can constrain the number of ancillae used.

\para{Circuit Level}
We are aware of only two works supporting uncomputation for non-qfree circuits.
ReQWire~\cite{rand_reqwire_2019} can only verify user supplied uncomputation (in
the case of non-qfree circuits).
Unqomp~\cite{paradis_unqomp_2021} allows to synthesize uncomputation for quantum
circuits, but cannot trade qubits for gates. Further, as discussed in \cref{sec:eval}, it uses a notion of circuit graphs that does not allow to track qubit values and therefore is unable to uncompute directly many examples that Reqomp can handle.

\section{Conclusion}
We introduced \reqomp, a method to synthesize and place efficient uncomputation for quantum circuits with space constraints.
\reqomp is proven correct and can easily be integrated into circuit based quantum languages such as Qiskit.
We demonstrate in our evaluation that \reqomp is widely applicable and yields wide ranges of trade-offs in space and time, for instance allowing to generate tightly space constrained circuits by using only a few ancilla qubits.

%%%%%%%%%%%%%%
% BODY - END %
%%%%%%%%%%%%%%
% this message marks the end of the body (used to check if we are over the page
% limit)
\message{^^JLASTBODYPAGE \thepage^^J}

%%%%%%%%%%%%%%%%
% BIBLIOGRAPHY %
%%%%%%%%%%%%%%%%
% may be replaced by new bibliography
%\clearpage
\bibliography{references}
\bibliographystyle{quantum}

%%%%%%%%%%%%%%%%%%%%
% BIBLIOGRAPHY END %
%%%%%%%%%%%%%%%%%%%%
% this message marks the end of the bibliography (used to split the paper from
% the supplement)
\message{^^JLASTREFERENCESPAGE \thepage^^J}

%%%%%%%%%%%%
% APPENDIX %
%%%%%%%%%%%%

% may decide to not include appendix
\ifincludeappendixx
	%\clearpage
	\appendix
	% use \include to generate "appendix.aux". Needed for references into the
	% appendix
	\begin{table}[t]
	\footnotesize
	\begin{tabular}{lp{4.3cm}}
			\textbf{Symbol} & \textbf{Meaning} \\
			$\i$ & Imaginary unit \\
			$o, p, q, r, t, u, \dots$ & Qubits \\
			$a, a^{(0)}, b, c, d, \dots$ & Ancilla qubits \\
			$n$ & Number of qubits \\
%			$m$ & Number of gates \\
			$C$ & Circuit \\
			$\graph=(V,E)$ & Circuit graph \\
			$\graphval=(V^{\mathit{val}}, E^{\mathit{val}})$ & Value graph \\
			$\graphu=(\un{V},\un{E})$ & Synthesized circuit graph \\
			$v,v',\un{v},w,\dots$ & Vertex \\
			$s$ & State index \\
			$i$ & Instance index \\
			$q_{s.i}, p_{0.1}, r_{1.0}, \dots$ & Vertex with explicit
			$q$, $s$, $i$ \\
			$c, \un{c}$ & Control vertex \\
			$\varphi, \phi, \psi, \sum_{j} \gamma_j \ket{j}$ & Quantum state \\
			$\gamma,\gamma',\un{\gamma}, \lambda,\lambda',\un{\lambda}, \dots$ & Complex coefficient (see above)  \\
			$j, k, l$ & Variables to sum over (see above) \\
			$Q$ & Set of qubits \\
			$A$ & Set of ancilla qubits \\
			$R$ & Set of non-ancillae qubits (rest) \\
			$F \colon \{0,1\}^{n+1} \to \{0,1\}$ & Classical function
			defining qfree gate with $n$ controls \\
			$U$ & (Unitary) gate (e.g., $X$ or $CX$) \\
			$\fme{G}$ & Semantics of a circuit graph, as a function over quantum states\\
			$\extended{G}$ & Extended graph of $G$ \\
			$\copyqb{q_{s.i}}$ & Qubit in $\extended{G}$ holding a copy of $q_{s.i}$\\
			$\copysem{G}_p$ & Coefficient for the projection $p$ of the semantics of $\extended{G}$ 
	   \end{tabular}
	   \caption{Notational conventions used throughout this work.}
	   \label{tab:notation}
\end{table}

\section{Notational Conventions}
\cref{tab:notation} summarizes notational conventions used in this work.

%%%%%%%%%%%%%%%%%%%%%%%%%%%%%%%%%%%%%%%%%%%%

\section{Algorithms}

\begin{figure*}
    \footnotesize
	\begin{algorithmic}[1]
		\continueLineNumber
		\Function{PartitionAncillae}{\null}
		%\Comment{All graphs over ancillae}
			\State $\ganc \gets \tfunc{AncillaDependencies}(\tvar{\graph})$
			\State $\tvar{comps} \gets \tprim{ConnectedComponents}(\ganc)$ \Comment{Subgraphs of $\ganc$}
			\State $\gadeps \gets (\tvar{comps}, \{\})$ \Comment{Each comp is a vertex of $\gadeps$}
			\For{$\tvar{comp} \in \tvar{comps}$}
				\For{$\tvar{comp'} \in \tvar{comps}$}
					\If{
						$\exists $ path from $\tvar{c} \in \tvar{comp}$
						to $\tvar{c'} \in \tvar{comp'}$ in $\graph$
					}
						\State \graphPrim{addEdge}{\gadeps}{\tvar{comp}, \tvar{comp'}}
					\EndIf
				\EndFor
			\EndFor
			\State \textbf{assert} \gadeps{} has no cycles
		%\State \textbf{return} vertices (i.e., sets of ancillae) of \gadeps{} in topological order
		\State \textbf{return} $\tvar{comps}$ in topological order according to $\gadeps$
		\EndFunction
		\State%%%%%%%%%%%%%%%%%%%%%%%%%%%%%%%%%%%%%%
		\Function{AncillaDependencies}{$\graph$}
		\Comment{Dependency graph on ancilla qubits}
		\State $V_a \gets \{\tvar{v}.\tfield{qbit} \mid \tvar{v} \in
		\tvar{\graph}, \tvar{v}.\tfield{isAnc}\}$
			\State $E_a \gets \{(\tvar{v}.\tfield{qbit},
			\tvar{v'}.\tfield{qbit}) \mid \tvar{v} \ctrlEdge \tvar{v'} \in G\} \cap V_a \times V_a$
			\State $G_a \gets (V_a, E_a)$
		\EndFunction
        \recordLineNumber
    \end{algorithmic}
    \caption{Partitioning the uncomputation problem into independent subproblems.}
    \label{alg:partition}
\end{figure*}

\subsection{Partitioning}\label{app:partitioning}

\cref{alg:partition} shows the algorithm for partitioning the input graph.

%%%%%%%%%%%%%%%%%%%%%%%%%%%%%%%%%%%%%%%%%%%%%%

\begin{figure}[t!]
    \footnotesize
	\begin{algorithmic}[1]
		\continueLineNumber
		\Function{getAvailCtrl}{\tvar{c}: Vertex, \tvar{\un{v}}: Vertex, \tvar{I}:Set[Qb]}\label{lin:reqompControlStrategy-getAvailableControl}
			%\Comment{return $\tvar{\un{c}} \equiv \tvar{c}$, available for \tvar{\un{v}}}
			\State $\tvar{\un{c}} \gets
			\graphPrim{getVertex}{\graphu}{\tvar{c}.\tfield{qbit},
			\tstring{"last"}, \tvar{c}.\tfield{valIdx}}$\label{lin:reqompControlStrategy-default-start}
			\If{\graphPrim{isAvailable}{\graphu}{\tvar{\un{c}},
			\tvar{\un{v}}}}
				\State \textbf{return} $\tvar{\un{c}}$\label{lin:reqompControlStrategy-default-end}
			\Else
				\State $\tvar{\un{c}'} \gets \graphPrim{getVertex}{\graphu}{\tvar{c}.\tfield{qbit}, \tstring{"last"}}$\label{lin:reqompControlStrategy-evolve-start}
				\State $\tvar{\un{c}} \gets \tfunc{evolveVertexUntil}(\tvar{\un{c}'},
				\tvar{c}.\tfield{valIdx}, I)$\label{lin:getAvailCtrl-evolve}
				\State \textbf{return} \tvar{\un{c}}\label{lin:reqompControlStrategy-evolve-end}
			\EndIf
		\EndFunction

		\State%%%%%%%%%%%%%%%%%%%%%%%%%%%%%%%%%%%%%%

		\Function{assertFullyEvolved}{\null}
		    \State \Comment{Abort if values were evolved incorrectly}
			\For{$\tvar{v} \in$ final vertices in \graph}
				\State $\tvar{\un{v}} \gets \graphPrim{getVertex}{\graphu}{\tvar{v}.\tfield{qbit}, \tstring{"last"}}$
				\If{\tvar{\un{v}}.\tfield{isAnc}}
					\State \textbf{assert} \tvar{\un{v}}.\tfield{valIdx} = 0
				\Else
					\State \textbf{assert} \tvar{\un{v}}.\tfield{valIdx} = \tvar{v}.\tfield{valIdx}
				\EndIf
			\EndFor
		\EndFunction
		\recordLineNumber
	\end{algorithmic}

	\begin{algorithmic}[1]
		\continueLineNumber
		\State%%%%%%%%%%%%%%%%%%%%%%%%%%%
		\Function{getPath}{\tvar{q}: Qubit, \tvar{from}: int, \tvar{to}: int}
			%\Comment{evolve \tvar{qbit} to \tvar{valIdx}}
			%\State \textbf{assert} no call to $\tfunc{evolveVertexUntil}$ with
			%argument $\tvar{qbit}$ in call stack
			%\Comment{\todo{Anouk: correct place to check?}}
			%\Comment{\todo{show necessity of check in example}}
			%\Comment{\todo{show successful recursion in example}}
			\State $p \gets \tfunc{shortestPathInValueGraph}(q, from, to)$
			\If{$q.\tfield{isAnc}$}
				\State \textbf{return} $p$
			\EndIf
			\State $p' \gets []$
			\State $\tvar{r} \gets {[}from + 1, to{]} \text{ if } from < to \text{ else } {[}to, from - 1{]}$ 
			\State $v \gets from$
			\For{$i$ in $r$}
				\If{$i \notin p$ and $q_{i.0} \notin \un{G}$}
					\State $p' \gets p' + \tfunc{shortestPathInValueGraph}(q, v, i)$
					\State $v \gets i$
				\EndIf
			\EndFor
			\State $p' \gets p' + \tfunc{shortestPathInValueGraph}(q, v, to)$
			\State \textbf{return} $p$
		\EndFunction
		\recordLineNumber
	\end{algorithmic}
    \caption{Convenience functions leveraged by Reqomp.}
    \label{alg:reqomp-convenience}
\end{figure}

\subsection{Reqomp Convenience Methods}\label{app:reqomp-convenience-methods}

\cref{alg:reqomp-convenience} shows convenience functions omitted from Reqomp.

\para{GetPath}
The function $\tfunc{getPath}$ used by $\tfunc{evolveVertexUntil}$ is shown in \cref{alg:reqomp-convenience}. For ancilla variables, it simply returns the shortest path between the two values in the value graph. However for non ancilla variables, it forces the computation of intermediate values that may not have been computed yet. This could happend for a circuit such as:

$$    \Qcircuit @C=0.4em @R=0.1em @!R {
					& \lstick{{q}} & \gate{X} & \gate{X} & \gate{H}
				}$$

Here the value graph is:
\begin{center}
	\begin{tikzpicture}
		%%%%%%%%%
		% NODES %
		%%%%%%%%%
	
		\matrix[row sep=2mm,column sep=0mm, inner sep=0mm, nodes={inner sep=1mm}] {
			\node (s0)  {$q_1$}; &\node (r0p)  {\phantom{llll}}; & \node (s1)  {$q_0$}; &\node (r1p)  {\phantom{llll}}; & \node (s2)  {$q_2$}; \\
		};
	
		\path[->]          (s0)  edge  [bend left=20] node[above] {\scriptsize $X$} (s1);
		\path[->]          (s1)  edge  [bend left=20] node[below] {\scriptsize $X$} (s0);
		\path[->]          (s1)  edge   node[above] {\scriptsize $H$} (s2);
	\end{tikzpicture}
\end{center}
%$$q_1 \leftrightarrows q_0 \rightarrow q_2$$
Therefore, if we want to compute $q_2$ from $q_0$, the shortest path is simply $q_0 \targEdge q_2$. However as $H$ is not qfree, once $q_2$ has been computed, it can never be uncomputed again, and therefore, we can never compute $q_1$, which may be needed for some later computation. To correct this, we introduce $q_1$ (if it has not already been computed in $\un{G}$) in the path, giving:

$$q_0 \targEdge q_1 \targEdge q_0 \targEdge q_2$$

\begin{figure*}[t]
	\footnotesize
	\begin{algorithmic}[1]
		\continueLineNumber
		\Function{getLinearStrat}{$\tvar{cc}$: Qubit, $\nq$: int}
			\State $\tvar{sortedAncillae} \gets \tfunc{topoSort}(\tvar{cc})$
			\State \textbf{return} $ [(\tvar{sortedAncillae}[i], b) \text{ for }(i, b) \in \tfunc{stepsDP}(|\tvar{sortedAncillae}|, \nq,
			\false)]$
		\EndFunction
		\State%%%%%%%%%%%%%%%%%%%%%%%%%%%%%%%%%%%%%%
		\Function{stepsDP}{$\na$: int, $\nq$: int, $\tvar{uncLast}$: bool}
			\If{return value was computed previously}
				\State \textbf{return} previously computed value\Comment{memoization}
			\EndIf
			\If{$\na = 0$}
				\State \textbf{return} $[]$
			\EndIf
			\If{$\na = 1$}
				\If{$\nq = 0$}
					\State \textbf{return} \textbf{null}
				\EndIf
				\If{\tvar{uncLast}}
					\State \textbf{return} $[(0, \true), (0, \false)]$
				\Else
					\State \textbf{return} $[(0, \true)]$
				\EndIf
			\EndIf
			\For{$\tvar{m} \in \{1, \dots, \na - 1\}$}
				\If{\tvar{uncLast}}
					\State $\tvar{toM} \gets \tfunc{stepsDP}(\tvar{m}, \nq, \false)$
					\Comment{$0 \rightarrow \tvar{m}$}
					\State $\tvar{fromM} \gets [(i + m, b) \text{ for } (i, b) \in \tfunc{stepsDP}(\na - \tvar{m}, \nq - 1,\true)]$
					\Comment{$\tvar{m} \rightleftarrows \tvar{\na}$}
					\State $\tvar{cleanM} \gets [(i, \lnot b) \text{ for } (i, b) \in \tfunc{reverse}(\tfunc{stepsDP}(m,\nq,\false))]$
					\Comment{$0 \leftarrow \tvar{m}$}
				\Else
					\State $\tvar{toM} \gets \tfunc{stepsDP}(\tvar{m}, \nq, \false)$
					\Comment{$0 \rightarrow \tvar{m}$}
					\State $\tvar{fromM} \gets [(i + m, b) \text{ for } (i, b) \in \tprim{stepsDP}(\na - m, \nq - 1,\false)]$
					\Comment{$\tvar{m} \rightleftarrows \tvar{\na}$}
					\State $\tvar{cleanM} \gets [(i, \lnot b) \text{ for } (i, b) \in \tfunc{reverse}(\tfunc{stepsDP}(m,\nq-1,\false))]$
					\Comment{$0 \leftarrow \tvar{m}$}
				\EndIf
				\State $\tvar{steps_m} \gets \tvar{toM} +
				\tvar{fromM} + \tvar{cleanM}$
			\EndFor
			\State \textbf{return} $\argmin_{\tvar{steps_m}} \tprim{cost}(\tvar{steps_m})$
		\EndFunction
		\recordLineNumber
	\end{algorithmic}
	\caption{Optimal uncomputation strategy for linear graphs.}
	\label{alg:linear-steps}
\end{figure*}

\subsection{Linear Steps}

\cref{alg:linear-steps} shows \tfunc{getLinearStrat}. It is adapted from \cite{knill_analysis_1995}: we added the $\tvar{uncLast}$ parameters that allows us to apply it to ancillae only (that is we want all qubits to be computed once then uncomputed whereas the original algorithm did not uncompute the last qubit in the dependency line).
	\section{Formal Correctness Proof}\label{app:proof}

In the following, we provide a formal proof that Reqomp synthesizes correct
uncomputation according to \cref{def:correct}.

\subsection{Definitions and Helper Lemmas}\label{app:definition-lemmas}

We first define what we consider to be a valid circuit graph, following
\cite{paradis_unqomp_2021}:
\begin{restatable}[Valid Circuit Graph]{mydef}{defvalid}
	\label{def:valid}
    A circuit graph is valid iff
	\begin{enumerate}[label=(\roman*)]
		\item its init vertices have no incoming target
    edge while gate vertices have exactly one, 
		\item all its vertices have at
    most one outgoing target edge
		\item its anti-dependency edges can be
    reconstructed from its control and target edges according to the rule
    discussed in \cref{sec:circ-graphs:circ-g}, 
		\item the
    number of incoming control edges of each gate vertex $v$ matches the number
    of controls of the gate of $v$
		\item $G$ is acyclic.
	\end{enumerate}
\end{restatable}

In a valid circuit graph, we can define for any non init vertex $n$ its predecessor $\tfunc{pred}(n)$ as the only vertex $m$ such that $m \targEdge n$ (the target edge from $m$ goes to $n$). We can also define for any qubit $q$ its last vertex $\tfunc{last}(q)$: it is the only vertex on qubit $q$ with no outgoing target edge.

We now recall the well-valued circuit graph definition.
\begin{restatable}[Well-valued Circuit Graph]{mydef}{defvalid}
    We say a valid circuit graph is well valued iff:
	\begin{enumerate}[label=(\roman*)]
		\item all vertex names are of the form $q_{s.i}$ where $q$ is the name
		of the vertex qubit, $s$ and $i$ are natural numbers
		\item there are no duplicate vertices
		\item the init vertex on each qubit is named $q_{0.0}$ and for any
		$q_{s.i}$ in $G$, $q_{s.0}$ is in $G$
		\item any gate vertex $q_{s.i}$ with $s>0$ satisfies one of the following:
		\subitem \textbf{(fwd)}
		$\tfunc{valIdx}(\tfunc{pred}(q_{s.i})) =
		\tfunc{valIdx}(\tfunc{pred}(q_{s.0}))$ and $q_{s.i}$ and $q_{s.0}$ have
		the same gate and same control vertices (up to their instance indices)
		\subitem \textbf{(bwd)} if we denote $s' = \tfunc{valIdx}(\tfunc{pred}(q_{s.i}))$,
		we have that  (i)~$\tfunc{valIdx}(\tfunc{pred}(q_{s'.0})) = s$, (ii)~$q_{s.i}.gate$ is qfree and equal to $q_{s'.0}.gate^\dagger$, and (iii)~both $q_{s.i}$ and $q_{s'.0}$ have the same controls (up to instance indices).
\end{enumerate}
\end{restatable}

Vertices in a well-valued circuit graph are of the shape $q_{s.i}$, where we call $s$ its value index ($\tvar{valIdx}$ in the
algorithms) and $i$ its instance index. $i$ is $0$ for the first occurrence of
$q_s$ in the graph, but otherwise we only use its value to ensure uniqueness of
the vertex names. 

Due to the following lemma, it suffices to only consider valid and well-valued
circuit graphs:

\begin{restatable}[evolveVertex Correctness]{mylem}{lemextended}
	\label{lem:evolveVertexCorrect}
    For a valid and well-valued circuit graph $G$, any number of calls to evolveVertex results in a valid and well-valued circuit graph $\un{G}$ such that (i) $\{q_{s.0} \in \un{G}\} $ is a subset of $ \{q_{s.0} \in G\}$ and (ii) for any $q_{s.0}$ in $G \cap \un{G}$, it has the same gate and control vertices (up to instance index) in both graphs.
\end{restatable}
\begin{proof}
	By induction on the depth of calls to evolveVertex.
\end{proof}

We then define the extended graph $\extended{G}$ of a circuit graph $G$.
Roughly, we want $\extended{G}$ to keep a copy of every vertex $q_{s.i}$ in $G$,
saved on a fresh qubit $\copyqb{q_{s.i}}$. 	For a graph $G$ with one qubit and
two vertices, we show $\extended{G}$ in \cref{fig:extended-graph}.

\begin{figure}
	\centering
	\begin{tikzpicture}
		%%%%%%%%%
		% NODES %
		%%%%%%%%%
	
		\matrix[row sep=5mm,column sep=6mm] {
			\node (x00) [style=gate] {$x_{0.0}$}; &
			\node (x00copy) [style=copygate] {$\copyqb{x_{0.0}}_{0.0}$}; &
			\node (x10copy) [style=copygate] {$\copyqb{x_{1.0}}_{0.0}$}; \\

			&
			\node (x00copy1) [style=copygate] {$\copyqb{x_{0.0}}_{0.1} : CX$}; &
			\\

			\node (x10) [style=gate] {$x_{1.0} : H$}; &
			 & \\

			 & &
			\node (x10copy1) [style=copygate] {$\copyqb{x_{1.0}}_{1.0}:  CX$}; \\
		};

		%%%%%%%%%%%%%%%%
		% TARGET EDGES %
		%%%%%%%%%%%%%%%%
	
		\draw[style=targ] (x00) -- (x10);
		\draw[style=targ] (x00copy) -- (x00copy1);
		\draw[style=targ] (x10copy) -- (x10copy1);
	
		%%%%%%%%%%%%%%%%%
		% CONTROL EDGES %
		%%%%%%%%%%%%%%%%%
	
		\draw[style=ctrl] (x00) -- (x00copy1);
		\draw[style=ctrl] (x10) -- (x10copy1);

		%%%%%%%%%%%%%%%%%
		% AVAIL EDGES %
		%%%%%%%%%%%%%%%%%
	
		%\draw[style=avail] (q10) -- (c10);
		%\draw[style=avail] (q01) -- (c10);
		%\draw[style=ctrl] (c01) -- (q11);
	\end{tikzpicture}

	\caption{Extended graph example, copy vertices are shown in green.}
    \label{fig:extended-graph}
\end{figure}

\begin{restatable}[Extended Graph]{mydef}{defextended}
	\label{def:extended}
    For any circuit graph $G = (V, E)$, we define its extended graph $\extended{G} = (V_e, E_e)$ as follows:
	\begin{align*}
		V_e =& V \cup \left\{\copyqb{q_{s.i}}_{0.0}, \copyqb{q_{s.i}}_{1.0}  \mid q_{s.i} \in V\right\} \\
		E_e =& E \cup \left\{\copyqb{q_{s.i}}_{0.0} \targEdge \copyqb{q_{s.i}}_{1.0} \mid q_{s.i} \in V\right\}  \\
	&\cup \left\{q_{s.i} \ctrlEdge\copyqb{q_{s.i}}_{1.0} \mid q_{s.i} \in V\right\} 
	\end{align*}
\end{restatable}
	For each $q_{s.i}$ in $V$, we have added a new qubit $\copyqb{q_{s.i}}$,
	with one init vertex and one gate vertex $CX$ controlled by $q_{s.i}$. In
	the following we refer to those added qubits as
	$\copyqb{V}$. Note that while $\copyqb{q_{s.i}}_{1.0}$ is a vertex, $\copyqb{q_{s.i}}$ is qubit.

As the extended graph is a valid graph, it corresponds to a circuit and
therefore its semantics $\fme{\extended{G}}$ is well defined. For
a given input state $\varphi$ to $G$, this allows us to define:

\begin{restatable}[Projected Coefficients]{mydef}{defprojcoeff}
	\label{def:projcoeff}
    For a fixed input state $\varphi$ to the circuit graph $G = (V, E)$, we define the projected coefficients of $G$ as the unique complex numbers $\copysem{G}_p$ such that:

	\begin{align*}&\fme{\extended{G}} \varphi \otimes \ket{0 ... 0}_{\copyqb{V}} =\\
		 &\sum_{p : \extended{G}.\tvar{qbs} \rightarrow \{0, 1\}} \copysem{G}_p \ket{p(G.qbs)}_{G.qbs} \otimes \ket{p(\copyqb{V})}_{\copyqb{V}}
	\end{align*}
	where $p(Q) = (p(q^{(1)}), ..., p(q^{(n)}))$ for qubits $Q = \{q^{(1)}...
	q^{(n)}\}$.
\end{restatable}

Using these coefficients, we can prove the following three lemmas. 
First, the semantics of the circuit graph $G$ can be expressed in terms
of its projected coefficients $\copysem{G}_p$:
\begin{restatable}[Projected Coefficients for Graph Semantics]{mylem}{lemextended}
	\label{lem:graphsemascopysem}
    For a circuit graph $G$ we have:
	\begin{align*}&\fme{G} \varphi =
		&\sum_{p : \extended{G}.\tvar{qbs} \rightarrow \{0, 1\}} \copysem{G}_p \ket{p(G.\tvar{qbs})}
   \end{align*}
\end{restatable}
\begin{proof}
	We can prove this by induction on the number of gates in $G$.
\end{proof}

Second, copies have consistent values. Specifically, for a given qubit $q$ and
valIdx $s$, all $\copyqb{q_{s.i}}$ hold the same value as $\copyqb{q_{s.i}}$, and the value of $q$ is the same as the
copy of the last vertex on $q$:
\begin{restatable}[Null Projected Coefficients]{mylem}{lemextended}
	\label{lem:nullprojcoeffs}
    For a valid and well-valued circuit graph $G = (V, E)$ and $p : \extended{G}.\tvar{qbs}
    \rightarrow \{0, 1\}$, we have $\copysem{G}_p = 0$ if 
	\begin{enumerate}[label=(\roman*)]
		\item $p(\copyqb{q_{s.i}}) \neq p(\copyqb{q_{s.0}})$ for some $q_{s.i}$, or
		\item $p(q) \neq p(\copyqb{\tfunc{last}(q)})$ for some qubit $q$.
	\end{enumerate}
\end{restatable}
\begin{proof}
	We prove \cref{lem:nullprojcoeffs} in \cref{sec:lemmas-proofs}.
\end{proof}

Finally, if $\copysem{G}_p \neq 0$, it depends only on the gates used for the
first computation of each $q_{s.0}$. 

\begin{restatable}[Projected Coefficients Values]{mylem}{lemextended}
	\label{lem:valprojcoeffs}
    For a circuit graph $G$ and $p : \extended{G}.\tvar{qbs} \rightarrow \{0,
    1\}$, we have that if $\copysem{G}_p \neq 0$ then:
	\begin{align*}\copysem{G}_p =&
		\alpha_{p(q^{(0)}_{0.0} ... q^{(n)}_{0.0})} \prod_{\substack{q_{s.0} \in G \\ s \neq 0}} \gamma^{q_{s.0}}
   \end{align*}
\end{restatable}

Here the $\alpha$
describe the initial state:
 $$\varphi = \sum_{k \in \{0, 1\}^m} \alpha_k \ket{k}.$$
 and the $\gamma$ are gate
 coefficients defined such that $\fme{g} \ket{c} \ket{t} = \sum_{t' =
 0}^1\gamma^{g}_{t, c \rightarrow t'} \ket{c} \ket{t'}$ for a gate $g$ and
 $t$ in $\{0, 1\}$ and $c \in \{0,1\}^m$. We have further shortened 
 $$\gamma^{q_{s.0}} = \gamma^{q_{s.0}.\tvar{gate}}_{p(\copyqb{pred(q_{s.0})}), p(\copyqb{ctrls(q_{s.0})}) \rightarrow p(\copyqb{q_{s.0}})}$$
\begin{proof}
	We prove \cref{lem:valprojcoeffs} in \cref{sec:lemmas-proofs}.
\end{proof}

\subsection{Main Proof}

Using \crefrange{lem:evolveVertexCorrect}{lem:valprojcoeffs}, we can prove the
correctness of Reqomp:

\begin{restatable}[Correctness]{mythm}{thmCorrectness}
	\label{thm:correctness}
	Have $G$ a circuit graph built from a circuit with $n$ qubits, of which $m$ are ancilla variables. Without loss of generality, we can assume
	that those ancilla variables $A= \left( a^{(1)}, \dots, a^{(m)} \right)$ are the
	first $m$ qubits of $\depgraph$. Let $\textsc{Reqomp}(G, A)=\graphu$. If
	
	\begin{alignat}{10}
		&\ket{0 \cdots 0}_A \otimes \varphi
		&\ \xmapsto{\fme{G}}
		&\sum_{k \in \{0,1\}^m} \lambda_{k}
		&&\ket{k}_A
		&&\otimes \phi_{k}, \text{ then} \label{eq:correctness-assumption}\\
		&\ket{0 \cdots 0}_A \otimes \varphi
		&\ \xmapsto{\fme{\graphu}}
		&\sum_{k \in \{0,1\}^m} \lambda_{k}
		&&\ket{0 \cdots 0}_A
		&&\otimes \phi_{k}. \label{eq:correctness-conclusion}
	\end{alignat}
\end{restatable}

Note that this is an equivalent rewrite of \cref{def:correct}.

\begin{proof}
	We first make the values of the non-ancilla qubits explicit, and denote $R =
	G.\tvar{qbs} \backslash A$. This allows us to rewrite
	\cref{eq:correctness-assumption} as :
	\begin{align}\fme{G}\ket{0 \cdots 0}_A \otimes \varphi
		=
		&\sum_{\substack{k \in \{0,1\}^m \\ k' \in \{0, 1\}^{n-m}}} \lambda_{kk'}
		&&\ket{k}_A \ket{k'}_R \label{eq:rewrite-g-effect}
	\end{align}
	Similarly for $\graphu$ we can write:
	\begin{align}\fme{\un{G}}\ket{0 \cdots 0}_A \otimes \varphi
		=
		&\sum_{\substack{k \in \{0,1\}^m \\ k' \in \{0, 1\}^{n-m}}} \un{\lambda_{kk'}}
		&&\ket{k}_A \ket{k'}_R \label{eq:rewrite-ung-effect}
	\end{align}

	Note that here we use $\un{\lambda}$ to refer to a coefficient in $\un{G}$, and not to the complex conjugate of $\lambda$.

	To prove the theorem, it is hence enough to prove that for all $k'$, 
	\begin{align*}
		\un{\lambda_{kk'}} = \begin{cases}
			0 & \text{ if } k \neq 0 \qquad \text{(i)} \\
			\sum_k \lambda_{kk'} & \text{ if } k = 0 \qquad \text{(ii)}
		\end{cases}
	\end{align*}

	To do so, we first identify \cref{eq:rewrite-ung-effect} with
	\cref{lem:graphsemascopysem}. This gives us that:

	\begin{align}
		\un{\lambda_{kk'}} = \sum_{\substack{p : \extended{\un{G}}.qbs \rightarrow \{0, 1\}\\ p(\un{G}.qbs) = kk'}} \copysem{\un{G}_p} \label{eq:val-lambda-un}
	\end{align}

	The assertion at
	\cref{lin:reqomp-assertFullyEvolved} in the \reqomp algorithm (\cref{alg:apply-start}) and \cref{lem:nullprojcoeffs} then
	give that for any ancilla qubit $a^{(i)}$, if $p(a^{(i)}) \neq
	p(\copyqb{a^{(i)}_{0.0}})$, then $\copysem{\un{G}_p}$ is null. As
	$\copyqb{a^{(i)}_{0.0}}$ copies the initial state of the ancilla, we then
	get that if $k \neq 0$, then $\un{\lambda_{kk'}} = 0$, proving (i).

	To prove (ii), we first note that \cref{eq:val-lambda-un} holds analogously
	for $G$, allowing us to derive the following. Here, we denote $V_0 =
	\{q_{s.0} \in V\}$. We then have for any $k'$ in $\{0, 1\}^{n-m}$:

	\begin{align}
		\sum_{k \in \bools^m }{\lambda_{kk'}} &= \sum_{k \in \bools^m} \sum_{\substack{p : \extended{G}.qbs \rightarrow \{0, 1\}\\ p(G.qbs) = kk'}} \copysem{G_p} \\
		&= \sum_{\substack{p : \extended{G}.qbs \rightarrow \{0, 1\}\\ p(R) = k'}} \copysem{G_p} \\
		&= \sum_{p_0 : V_0 \rightarrow \bools}\sum_{\substack{p : \extended{G}.qbs \rightarrow \{0, 1\}\\ p(R) = k' \\ p_{|V_0} = p_0}} \copysem{G_p} \label{eq:processing-lambda}
	\end{align}

	Using \cref{lem:nullprojcoeffs}, we have that for any $p_0 : V_0 \rightarrow
	\bools$, there is a unique $p_0^+ : \extended{G}.qbs \rightarrow \bools$
	such that $p_{|V_0} = p_0$ and $\copysem{G}_p$ is not known to be null. We
	can hence further rewrite \cref{eq:processing-lambda}:

	\begin{align}
		\sum_{k \in \bools^m }{\lambda_{kk'}} = \sum_{\substack{p_0 : V_0 \rightarrow \bools \\ p_0(R_0) = k'}} \copysem{G}_{p_0^+}, \label{eq:processing-lambda2}
	\end{align}
	where we denoted $R_0 = \{\copyqb{q_{s.0}} \mid \tfunc{last}(q) = q_{s.i}, q \in R\}$. Similarly, we write the same equation for $\un{G}$ and $\un{\lambda_{kk'}}$:

	\begin{align}
		\sum_{k \in \bools^m }{\un{\lambda_{kk'}}} = \sum_{\substack{p_0 : \un{V_0} \rightarrow \bools \\ p_0(\un{R_0}) = k'}} \copysem{\un{G}}_{p_0^+} %\label{eq:processing-lambda2-un}
	\end{align}

	Using that $\un{\lambda_{kk'}}$ is null if $k \neq 0$, we finally get:

	\begin{align}
		\un{\lambda_{0k'}} = \sum_{\substack{p_0 : \un{V_0} \rightarrow \bools \\ p_0(\un{R_0}) = k'}} \copysem{\un{G}}_{p_0^+} \label{eq:processing-lambda2-un}
	\end{align}

	Now, if $V_0 = \un{V_0}$, as gates are the same for any $q_{s.0}$ in $G$ and $\un{G}$, 
	 \cref{eq:processing-lambda2} and
	\cref{eq:processing-lambda2-un} combined with \cref{lem:valprojcoeffs}  and \cref{lem:evolveVertexCorrect} give
	us (ii).
	
	If this is not the case, there must exist some $q_{s.0}$ in $V_0 \backslash
	\un{V_0}$. For instance this could happen if $G$ contains $q_{0.0} \targEdge
	q_{1.0} \targEdge q_{0.1}$, and $\un{G}$ only contains $q_{0.0}$ : it was
	not necessary to compute $q_{1.0}$ to reach the same final state as in $G$.
	The crucial observation is that this vertex $q_{s.0}$ gate is qfree.
	If $q$ is an ancilla, this is clear as Reqomp would have raised an error
	otherwise.
	Indeed, Reqomp computes all ancillae (in \crefrange{lin:reqomp-detailed-steps-beg}{lin:reqomp-detailed-steps-end}) and check that they are all later uncomputed (\cref{lin:reqomp-assertFullyEvolved}). All operations on ancillae are hence uncomputed, and therefore their gate must be qfree (this is checked in \cref{lin:evolve-vtx-assert-well-val}). If $q$ is not an
	ancilla, it means then $q_{s.0}$ must have been uncomputed in $G$ (as the
	final state of non ancilla qubits in both graphs is the same). As qfree
	gates coefficients $\gamma$ are either $0$ or $1$, having an extra qfree
	gates does not change the result of the sum in \cref{eq:processing-lambda2}, concluding this proof.
\end{proof}

\subsection{Proofs of Helper Lemmas} \label{sec:lemmas-proofs}
We now prove \cref{lem:nullprojcoeffs} and \cref{lem:valprojcoeffs} by induction
on the number of gates in $G$.

\begin{proof}
For a circuit graph $G$ with no gates, an immediate induction on the number of qubits gives both lemmas.

Now suppose both lemma holds for any $G$ with at most $l$ gates. Now have $G' =
(V', E')$ with $l+1$ gates. We can write $G'$ as $G' = G \cdot q_{s.i}$ where $G
= (V, E)$ has $l$ gates and $q_{s.i}$ can be applied last in $G'$. To simplify notations, we assume $q_{s.i}$ has only one control vertex
$c_{t.j}$. If it has 0 or more controls, the reasoning is analogous.

By definition, we know that:

$$\fme{\extended{G}} \varphi = \sum_{p : \extended{G}.qbs \rightarrow \{0,
1\}} \copysem{G}_p \ket{p(G.\tvar{qbs})}_{G.\tvar{qbs}} \ket{p(\copyqb{V})}_{\copyqb{V}}$$

If we now apply $q_{s.i}$ and $\copyqb{q_{s.i}}_{0.1}$ (the CX gate copying
$q_{s.i}$ to a new qubit) to one state
of the sum above, we get for any $p :
\extended{G}.qbs \rightarrow \{0, 1\}$:

\begin{align*}
	&\fme{q_{s.i} \cdot \copyqb{q_{s.i}}_{0.1}} \ket{p(G.\tvar{qbs})}_{G.qbs} \ket{p(\copyqb{V})}_{\copyqb{V}} = \\
	&\sum_{b \in \{0, 1\}} \gamma^{q_{s.i}.gate}_{p(q), p(c) \rightarrow b} \ket{p(G.qbs \backslash\{q\}), b}_{G.qbs} \ket{p(\copyqb{V}), b}_{\copyqb{V'}}
\end{align*}
Here $b$ appears first as the value on the qubit $q$, and second as the value on the copy qubit $\copyqb{q_{s.i}}$.

As $\extended{G'} = \extended{G} \cdot q_{s.i} \cdot \copyqb{q_{s.i}}_{0.1}$, we
can use the above to compute $\copysem{G'}_{p'}$ for any $p' : \extended{G'}.qbs
\rightarrow \bools$. We first notice that if $p'(q) \neq p'(\copyqb{q_{s.i}})$,
then $\copysem{G'}_{p'} = 0$, giving us in \cref{lem:nullprojcoeffs}~(ii) for
$q$. In the following, we hence only consider $p'$ such that $p'(q) =
p'(\copyqb{q_{s.i}})$. We then get:

\begin{align}
	&\copysem{G'}_{p'} =\sum_{b \in \{0, 1\}} \copysem{G}_{\substack{p'_{|\extended{G}.qbs} \\ [q \mapsto b]}} \gamma^{q_{s.i}.gate}_{b, p'(c) \rightarrow p'(q)} %\ket{p'(G'.qbs)}_{G'.qbs} \ket{p'(\copyqb{V'})} _\copyqb{V'}
	\label{eq:valprojv1}
\end{align}

Using the recursion hypothesis, this immediately gives that if for any $q' \neq
q$ if $p'(q') \neq p'(\copyqb{last(q')})$, then $\copysem{G'}_{p'} = 0$, giving
us \cref{lem:nullprojcoeffs}~(ii) for $q' \neq q$. Together with the above, we
hence get \cref{lem:nullprojcoeffs}~(ii).

We now work on proving both \cref{lem:nullprojcoeffs}~(i) and
\cref{lem:valprojcoeffs}. The recursion hypothesis gives us that for any $(q',
s', i') \neq (q, s, i)$, if $p'(\copyqb{q'_{s'.i'}}) \neq p'(\copyqb{q'_{s'.0}})$, then again
$\copysem{G'}_{p'} = 0$. 
We hence only need to establish that if $p'(q_{s.i}) \neq p'(q_{s.0})$ then
$\copysem{G'}_{p'} = 0$, and the value of this coefficient when it is not null
(that is to say \cref{lem:valprojcoeffs}). To do so, we now consider $p'$
consistent with what we have already proven, i.e., $p'$ such that for any $q'$
in $G.qbs$, $p'(q') =
p'(\copyqb{last(q')})$ and for any $(q', s', i') \neq (q, s, i)$,
$p'(q'_{s'.i'}) = p'(q'_{s'.0})$.

We first notice that the recursion hypothesis gives that if $b \neq
p'(\copyqb{pred(q_{s.i})})$, then $\copysem{G}_{\substack{p'_{|\extended{G}.qbs}
\\ \oplus q \mapsto b}}=0$. Hence one of the summands in \cref{eq:valprojv1} is
null:
\begin{align*}
	&\copysem{G'}_{p'} =\copysem{G}_{\substack{p'_{|\extended{G}.qbs} \\ \oplus q \mapsto p'(\copyqb{pred(q_{s.i})})}} \gamma^{q_{s.i}.gate}_{p'(\copyqb{pred(q_{s.i})}), p'(c) \rightarrow p'(q)} %\ket{p'(G'.qbs)}_{G'.qbs} \ket{p'(\copyqb{V'})}_\copyqb{V'}
\end{align*}
Using the constraints on $p'$, we can rewrite this to:
\begin{align*}
	&\copysem{G'}_{p'} =\copysem{G}_{\substack{p'_{|\extended{G}.qbs} \\ \oplus q \mapsto p'(\copyqb{pred(q_{s.i})})}} \gamma^{q_{s.i}.gate}_{p'(\copyqb{pred(q_{s.i})}), p'(\copyqb{c_{t.0}}) \rightarrow p'(\copyqb{q_{s.i}})} %\ket{p'(G'.qbs)}_{G'.qbs} \ket{p'(\copyqb{V'})}_\copyqb{V'}
\end{align*}

Now we distinguish two cases. If this is the first occurence of $q_{s}$, that is
to say $i = 0$, we immediately get \cref{lem:nullprojcoeffs} (i), as $i =0$. For \cref{lem:valprojcoeffs}, by rewriting the equation above using that $i = 0$
\begin{align*}
	&\copysem{G'}_{p'} =\copysem{G}_{\substack{p'_{|\extended{G}.qbs} \\ \oplus q \mapsto p'(\copyqb{pred(q_{s.0})})}} \gamma^{q_{s.0}.gate}_{p'(\copyqb{pred(q_{s.0})}), p'(\copyqb{c_{t.0}}) \rightarrow p'(\copyqb{q_{s.0}})} %\ket{p'(G'.qbs)}_{G'.qbs} \ket{p'(\copyqb{V'})}_\copyqb{V'}
\end{align*}
and using the induction hypothesis, we can conlude. 

On the other hand, if $i \neq 0$, we again need to
distinguish two possibilities: cases \textbf{fwd} and \textbf{bwd} in \cref{def:well-val-4} of \cref{def:well-valued}. % either $pred(q_{s.i})$ is some $q_{s+1.i'}$ or it is some $q_{s-1.i'}$. 
We focus on the later case, as the first is simpler. We denote $q_{s'.i'} = \tfunc{pred}(q_{s.i})$. We
can hence rewrite:

\begin{align*}
	\gamma^{q_{s.i}.gate}_{p'(\copyqb{pred(q_{s.i})}), p'(\copyqb{c_{t.0}}) \rightarrow p'(\copyqb{q_{s.i}})} = \\ \gamma^{q_{s.i}.gate}_{p'(\copyqb{q_{s'.0}}), p'(\copyqb{c_{t.0}}) \rightarrow p'(\copyqb{q_{s.i}})}
\end{align*}

As we know that $G'$ is well-valued, we have that $q_{s.i}.gate = q_{s'.0}.gate
^\dagger$ and that both gates are qfree. Generally, the coefficient for the
reverse of a qfree gate $g$ is
\begin{align*}
	\gamma^{g^\dagger}_{t, c \rightarrow t'} = \gamma^{g}_{t', c \rightarrow t},
\end{align*}
as a qfree gate coefficient can only be $0$ or $1$.

We can hence again rewrite the above coefficient as:
\begin{align*}
	\gamma^{q_{s'.0}.gate}_{p'(\copyqb{q_{s.i}}), p'(\copyqb{c_{t.0}}) \rightarrow p'(\copyqb{q_{s'.0}})}.
\end{align*}
Here, we used that $\left(g^\dagger\right)^\dagger= g$.
Overall this gives us that:
\begin{align*}
	&\copysem{G'}_{p'} =\copysem{G}_{\substack{p'_{|\extended{G}.qbs} \\ \oplus q \mapsto p'(\copyqb{pred(q_{s.0})})}} \gamma^{q_{s'.0}.gate}_{p'(\copyqb{q_{s.i}}), p'(\copyqb{c_{t.0}}) \rightarrow p'(\copyqb{q_{s'.0}})}%\ket{p'(G'.qbs)}_{G'.qbs} \ket{p'(\copyqb{V'})}_\copyqb{V'}
	\label{eq:val-sem-g}
\end{align*}

Now if $p'(\copyqb{q_{s.i}}) \neq p'(q_{s.0})$, let us prove that $\copysem{G'}_{p'}$ is null.
If $\copysem{G}_{\substack{p'_{|\extended{G}.qbs} \\ \oplus q
\mapsto p'(\copyqb{pred(q_{s.i})})}}$ is null this is clear, otherwise using \cref{lem:valprojcoeffs} we get that $\copysem{G}_{\substack{p'_{|\extended{G}.qbs} \\ \oplus q
\mapsto p'(\copyqb{pred(q_{s.i})})}}$ contains $\gamma^{q_{s'.0}.gate}_{p'(\copyqb{q_{s.0}}),
p'(\copyqb{c_{t.0}}) \rightarrow p'(\copyqb{q_{s'.0}})}$.
We hence have that $\copysem{G'}_{p'}$ is a product
which includes the factors 
$\gamma^{q_{s'.0}.gate}_{p'(\copyqb{q_{s.i}}), p'(\copyqb{c_{t.0}}) \rightarrow
p'(\copyqb{q_{s'.0}})}$ and $\gamma^{q_{s'.0}.gate}_{p'(\copyqb{q_{s.0}}),
p'(\copyqb{c_{t.0}}) \rightarrow p'(\copyqb{q_{s'.0}})}$. As $q_{s'.0}.gate$
is qfree, one of those coefficients is null, and hence so is
$\copysem{G'}_{p'}$.

Finally, if $p'(\copyqb{q_{s.i}}) = p'(q_{s.0})$, we get
\cref{lem:nullprojcoeffs} (i) trivially.  If
$\copysem{G}_{\substack{p'_{|\extended{G}.qbs} \\ \oplus q \mapsto
p'(\copyqb{pred(q_{s.i})})}}$ is null, \cref{lem:valprojcoeffs} holds trivially.
Otherwise, we use as above that $\gamma^{q_{s'.0}.gate}_{p'(\copyqb{q_{s.0}}),
p'(\copyqb{c_{t.0}}) \rightarrow p'(\copyqb{q_{s'.0}})}$ is in $\copysem{G}_{\substack{p'_{|\extended{G}.qbs} \\ \oplus q
\mapsto p'(\copyqb{pred(q_{s.i})})}}$. As $\gamma^{q_{s'.0}.gate}_{p'(\copyqb{q_{s.0}}),
p'(\copyqb{c_{t.0}}) \rightarrow p'(\copyqb{q_{s'.0}})}$ is 0 or 1, it is equal
to its squared value, and the recursion hypothesis allows us to conclude. 
\end{proof}

	\section{Evaluation Values}
\label{app:eval}
We show in \cref{tab:trade-offsrange-app} the absolute numerical results on which the relative values in \cref{tab:trade-offsrange} are based. \cref{tab:params-eval} further shows the exact parameters of each circuit used in our evaluation.

\begin{table*}[!htb]
    \caption{%
        Parameters for all examples in \cref{tab:trade-offsrange} and \cref{tab:trade-offsrange-app}.}
    \label{tab:params-eval}

    %\vspace{-0.2cm}
    \centering
    \footnotesize
    \setlength\tabcolsep{4pt}
    \begin{tabular}{p{0.15\textwidth}p{0.8\textwidth}}%
	\textbf{Algorithm}  &  \textbf{Parameters} \\% 
    \hline
    \textbf{Small} &\\
    Adder & 12 qubits per operand\\
    Deutsch-Jozsa& 10 control qubits, with oracle MCX, returning true iff the value is 1111111111 \\
    Grover's algorithm & 5 control qubits,  with oracle MCX, returning true iff the value is 1111111111\\
    IntegerComparator & 12 control qubits, comparing to $i = 463$\\
    MCRY & 12 control qubits, with rotation angle $\theta = 4$\\
    MCX & 12 control qubits\\
    Multiplier & 5 qubits for each operand, and 5 for the result\\
    PiecewiseLinearR & 6 control qubits, function breakpoints are [10, 23, 42, 47, 51, 53, 63], slopes are [39, 32, 77, 27, 77, 4, 74] and offsets are [174, 40, 110, 163, 100, 185, 130]\\
    PolynomialPauliR & 5 control qubits, polynomial coefficients are [2, 2, 2, 2, 2]\\
    WeightedAdder & 10 controls qubits, values for sum are [0, 1, 1, 5, 2, 10, 4, 4, 9, 3]\\
    \textbf{Big} &\\
    Adder & 100 qubits per operand \\
    Deutsch-Jozsa & 100 control qubits, with oracle MCX, returning true iff the value is 1111111111 \\
    Grover's algorithm & 10 control qubits,  with oracle MCX, returning true iff the value is 1111111111 \\
    IntegerComparator & 100 control qubits, comparing to $i = 878234040205782925887743338143$ \\
    MCRY & 200 control qubits, with rotation angle $\theta = 4$ \\
    MCX & 200 control qubits \\
    Multiplier & 16 qubits for each operand, and 5 for the result \\
    PiecewiseLinearR & 40 control qubits, function breakpoints are [63870600266, 81180069351, 185076947411, 350818281077, 590566882159, 677977056232, 866030640015, 949186564661, 978976427282], offsets are [46, 59, 40, 48, 54, 67, 21, 71, 22] and coefficients are [60, 59, 6, 45, 83, 44, 34, 130, 130]\\
    PolynomialPauliR & 10 control qubits, polynomial coefficients are [2, 2, 2, 2, 2, 2, 2, 2, 2, 2] \\
    WeightedAdder & 20 controls qubits, values for sum are [9, 0, 9, 10, 2, 6, 10, 6, 8, 5, 8, 7, 8, 4, 0, 0, 5, 7, 5, 6]\\
\end{tabular}
\end{table*}

\newcommand{\tradeofftablecolheader}{\# Total Qubits & \# Ancillae Qubits & \# CX Gates & \# Gates & Depth}
\newcommand{\shorttradeofftablecolheader}{Q & A & CX & G & D}

\begin{table*}[!htb]
    \caption{%
        \reqomp results for the reductions presented in \cref{tab:trade-offsrange}. We also report Unqomp results.
        Columns \textbf{Max} and \textbf{Min} report the results for the most
        aggressive settings, respectively optimizing only for number of qubits
        and optimizing only for number of gates.
        Columns \textbf{-75\%}, \textbf{-50\%}, and \textbf{-25\%} report the
        gate counts when achieving the respective ancilla reductions.
        Entries "x" indicate that a given ancilla reduction was not achieved. 
        Q is total number of qubits, A is number of ancillae, CX is number of CX gates, G is total number of gates and D is circuit depth. }
    \label{tab:trade-offsrange-app}

    %\vspace{-0.2cm}
    \centering
    \footnotesize
    \setlength\tabcolsep{4pt}
    \newcolumntype{U}{S[table-format=2.1,table-column-width=0.1cm, group-digits = false]}
    \begin{tabular}{l|rrrrr|rrrrr|rrrrr|}%
    &    \multicolumn{15}{c|}{\textbf{Ancilla Reduction}} \\%& \multicolumn{5}{c|}{\textbf{Unqomp}}\\
    &    \multicolumn{5}{c}{\textbf{Max}} & \multicolumn{5}{c}{\textbf{-75\%}} & \multicolumn{5}{c|}{\textbf{-50\%}}\\% & \multicolumn{5}{c}{\textbf{-25\%}}& \multicolumn{5}{c|}{\textbf{Min}} & & & & & \\
	\textbf{Algorithm}  &  \shorttradeofftablecolheader
	&\shorttradeofftablecolheader& \shorttradeofftablecolheader \\% 
	\hline
    \textbf{Small} & & & & & & & & & & & & & & & \\
        % to produce the file: plots.py run gets_vals_all_examples() then run rewrites_csv_output.py then copy paste tradeoffsclean.csv here. Examples names can be changed in plots.py
        Adder&28&4&326&1132&482&x&x&x&x&x&30&6&270&924&346\\
Deutsch-Jozsa&15&4&78&331&162&x&x&x&x&x&15&4&78&331&162\\
Grover&8&2&192&839&456&x&x&x&x&x&x&x&x&x&x\\
IntegerComparator&17&4&110&456&251&x&x&x&x&x&18&5&100&422&231\\
MCRY&17&4&122&486&263&x&x&x&x&x&18&5&104&414&225\\
MCX&17&4&102&405&243&x&x&x&x&x&18&5&96&381&202\\
Multiplier&24&20&420&1500&550&x&x&x&x&x&x&x&x&x&x\\
PiecewiseLinearR&10&3&1000&3851&2188&x&x&x&x&x&10&3&1000&3851&2188\\
PolynomialPauliR&8&2&360&1381&846&x&x&x&x&x&x&x&x&x&x\\
WeightedAdder&25&40&689&2606&1184&x&x&x&x&x&x&x&x&x&x\\

        \textbf{Big} & & & & & & & & & & & & & & & \\
        Adder&207&7&6824&24664&8832&225&25&2820&9792&3005&250&50&2470&8492&2806\\
Deutsch-Jozsa&108&7&2700&10999&5632&125&24&1038&4351&2289&150&49&888&3751&2041\\
Grover&15&4&3600&15312&7734&x&x&x&x&x&15&4&3600&15312&7734\\
IntegerComparator&108&7&2720&11758&6190&125&24&1042&4584&2452&150&49&892&3938&2177\\
MCRY&209&8&7358&29430&15268&250&49&2096&8382&4646&300&99&1796&7182&4134\\
MCX&209&8&7278&29109&15283&250&49&2088&8349&4653&300&99&1788&7149&4141\\
Multiplier&79&64&4688&16768&5764&x&x&x&x&x&x&x&x&x&x\\
PiecewiseLinearR&47&6&31064&124834&67353&51&10&25252&99754&55636&61&20&23812&93512&50941\\
PolynomialPauliR&15&4&30322&119245&65640&x&x&x&x&x&15&4&30322&119245&65640\\
WeightedAdder&38&80&1857&7042&3259&x&x&x&x&x&x&x&x&x&x\\

    & \multicolumn{10}{c|}{\phantom{space}} & & & & & \\\\
    & \multicolumn{10}{c|}{\phantom{space}} & & & & & \\\\
    &    \multicolumn{10}{c|}{\textbf{Ancilla Reduction}} & \multicolumn{5}{c|}{\textbf{Unqomp}}\\
    &   \multicolumn{5}{c}{\textbf{-25\%}}& \multicolumn{5}{c|}{\textbf{Min}} & & & & & \\
	\textbf{Algorithm}  &  \shorttradeofftablecolheader
	&\shorttradeofftablecolheader& \shorttradeofftablecolheader \\% 
	\hline
    \textbf{Small} & & & & & & & & & & & & & & & \\
        % to produce the file: plots.py run gets_vals_all_examples() then run rewrites_csv_output.py then copy paste tradeoffsclean.csv here. Examples names can be changed in plots.py
        Adder&33&9&228&768&275&35&11&200&664&267&36&12&200&664&267\\
Deutsch-Jozsa&17&6&66&283&142&19&10&54&235&126&19&8&54&235&126\\
Grover&8&2&192&839&456&9&3&168&743&386&9&3&168&743&378\\
IntegerComparator&21&8&82&346&172&24&12&68&294&164&24&11&68&310&164\\
MCRY&21&8&86&342&177&24&12&68&270&161&24&11&68&270&161\\
MCX&20&7&84&333&179&23&12&66&261&154&23&10&66&261&154\\
Multiplier&x&x&x&x&x&24&20&420&1500&550&24&20&460&1580&573\\
PiecewiseLinearR&11&4&958&3671&1952&13&6&906&3465&1932&13&6&906&3583&1986\\
PolynomialPauliR&8&2&360&1381&846&9&3&330&1261&693&9&3&330&1261&735\\
WeightedAdder&x&x&x&x&x&25&40&689&2606&1184&25&40&749&2886&1322\\

        \textbf{Big} & & & & & & & & & & & & & & & \\
        Adder&275&75&2120&7192&2494&299&99&1784&5944&2291&300&100&1784&5944&2291\\
Deutsch-Jozsa&174&73&744&3175&1563&199&100&594&2575&1386&199&98&594&2575&1386\\
Grover&17&6&3000&12912&6279&19&8&2550&11112&5902&19&8&2550&11112&5852\\
IntegerComparator&175&74&742&3284&1662&200&100&596&2670&1484&200&99&596&2862&1484\\
MCRY&350&149&1496&5982&3145&400&200&1196&4782&2793&400&199&1196&4782&2793\\
MCX&349&148&1494&5973&3138&399&200&1194&4773&2786&399&198&1194&4773&2786\\
Multiplier&x&x&x&x&x&79&64&4688&16768&5764&79&64&5168&17728&5886\\
PiecewiseLinearR&71&30&22372&87224&45711&81&40&21050&81592&44966&81&40&21050&84852&46752\\
PolynomialPauliR&17&6&26962&105805&55664&19&8&26572&104245&50936&19&8&26572&104245&60488\\
WeightedAdder&x&x&x&x&x&38&80&1857&7042&3259&38&80&1989&7658&3458\\

        \end{tabular}
\end{table*}
% ugly hack from https://tex.stackexchange.com/questions/402627/using-clearpage-after-figure-in-revtex-gives-error-output-routine-didnt-u
\makeatletter\onecolumngrid@push\makeatother
\clearpage
\makeatletter\onecolumngrid@pop\makeatother

\fi

\end{document}